\documentclass{article}

\usepackage{ijcai11}

\usepackage{times}
\usepackage{amsmath,amsfonts,amssymb}
\usepackage{hhline}
\usepackage[dvipsnames,usenames]{xcolor}
\usepackage{pifont} 
\usepackage{subfig}
\usepackage{tikz}
\usetikzlibrary{patterns}
\usetikzlibrary{calc}
\usetikzlibrary{decorations.pathmorphing}
\graphicspath{{5-1-Graphics/}{/}}
\usepgfmodule{plot}
\usepackage{url}

\newtheorem{theorem}{Theorem}
\newtheorem{lemma}[theorem]{Lemma}
\newtheorem{corollary}[theorem]{Corollary}

\newcommand{\ic}{c^\circ}
\newcommand{\NP}{\textsc{NP}}
\newcommand{\PSpace}{\textsc{PSpace}}
\newcommand{\ExpTime}{\textsc{ExpTime}}
\newcommand{\Sat}{\textit{Sat}}

\newcommand{\cB}{\mathcal{B}}%
\newcommand{\cBC}{\mathcal{C}}%
\newcommand{\cBc}{\ensuremath{\mathcal{B}c}}
\newcommand{\cBci}{\ensuremath{\mathcal{B}c^\circ}}
\newcommand{\RCCE}{\ensuremath{\mathcal{RCC}8}}%
\newcommand{\RCCF}{\ensuremath{\mathcal{RCC}5}}%
\newcommand{\BRCCE}{\ensuremath{\mathcal{BRCC}8}}%
\newcommand{\cBCc}{\ensuremath{\mathcal{C}c}}%
\newcommand{\cBCci}{\ensuremath{\mathcal{C}c^\circ}}%

\newcommand{\R}{\mathbb{R}}
\newcommand{\RC}{{\sf RC}}
\newcommand{\RCP}{{\sf RCP}}

\newcommand{\fw}{\mathsf{w}}
\newcommand{\fW}{\mathbf{w}}

\newcommand{\cL}{\mathcal{L}}%

\newcommand{\ti}[2][]{{#2}^{\circ_{#1}}}
\newcommand{\tc}[2][]{{#2}^{-_{#1}}}

\newcommand{\set}[1]{\{#1\}}
\newcommand{\tseq}[1]{\mathfrak{#1}}
\newcommand{\intermediate}[1]{\dot{#1}}
\newcommand{\inner}[1]{\ddot{#1}}
\newcommand{\stack}{\mathsf{stack}}
\newcommand{\stacki}{\ti{\mathsf{stack}}}
\newcommand{\frameFla}{\mathsf{frame}}
\newcommand{\frameFlai}{\ti{\mathsf{frame}}}
\newcommand{\md}[2][] {{\lfloor#2\rfloor_{#1}}}

\newcommand{\qedsymbol}{\ding{113}}
\newenvironment{proof}{\par\noindent\textbf{Proof.}}{\mbox{}\hfill\qedsymbol\par\bigskip}
\newenvironment{swetheorem}[1]{\par\medskip\noindent\textbf{#1.}\hspace*{0.5em}\em}{\par\smallskip}

\renewcommand{\phi}{\varphi}

\title{On the Decidability of Connectedness Constraints in 2D and 3D Euclidean Spaces}
\author{Roman Kontchakov$^1$\!, Yavor Nenov$^2$\!, Ian Pratt-Hartmann$^2$ and Michael Zakharyaschev$^1$\\
{\parbox[t]{60mm}{\centering $^1$Department of Computer Science and Information Systems\\ Birkbeck College London, U.K.}}\hspace*{3em}
\parbox[t]{60mm}{\centering $^2$School of Computer Science\\ University of Manchester, U.K.}}

\begin{document}

\maketitle

\begin{abstract}
We investigate (quantifier-free) spatial constraint languages with
equality, contact and connectedness predicates, as well as Boolean
operations on regions, interpreted over low-dimensional Euclidean
spaces. We show that the complexity of reasoning varies dramatically
depending on the dimension of the space and on the type of regions
considered. For example, the logic with the interior-connectedness
predicate (and without contact) is undecidable over polygons or
regular closed sets in $\R^2$, \ExpTime-complete over polyhedra in
$\R^3$, and \NP-complete over regular closed sets in $\R^3$.
\end{abstract}


\section{Introduction}\label{sec:intro}

A central task in Qualitative Spatial Reasoning is that of determining
whether some described spatial configuration is geometrically
realizable in 2D or 3D Euclidean space. Typically, such a description
is given using a spatial logic---a formal language whose variables
range over (typed) geometrical entities, and whose non-logical
primitives represent geometrical relations and operations involving
those entities. Where the geometrical primitives of the language are
purely topological in character, we speak of a \emph{topological
  logic}; and where the logical syntax is confined to that of
propositional calculus, we speak of a \emph{topological constraint
  language}.

Topological constraint languages have been intensively studied in
Artificial Intelligence over the last two decades.  The best-known of
these, \RCCE{} and \RCCF, employ variables ranging over regular closed
sets in topological spaces, and a collection of eight (respectively,
five) binary predicates standing for some basic topological relations
between these
sets~\cite{ijcai:Egenhofer&Franzosa91,ijcai:Randelletal92,ijcai:Bennett94,ijcai:Renz&Nebel98}. An
important extension of \RCCE, known as \BRCCE{}, additionally features
standard Boolean operations on regular closed
sets~\cite{ijcai:Wolter&Z00ecai}.

A remarkable characteristic of these languages is their
\emph{in}sensitivity to the underlying interpretation.  To show that an
\RCCE-formula is satisfiable in $n$-dimensional Euclidean space, it
suffices to demonstrate its satisfiability in {\em any} topological
space \cite{ijcai:Renz98}; for \BRCCE-formulas, satisfiability in
\emph{any connected} space is enough. This inexpressiveness yields
(relatively) low computational complexity: satisfiability
of~\BRCCE-, \RCCE- and \RCCF-formulas over arbitrary topological
spaces is \NP-complete; satisfiability of~\BRCCE{}-formulas over
connected spaces is \PSpace-complete.

However, satisfiability of spatial constraints by {\em arbitrary}
regular closed sets by no means guarantees realizability by
practically meaningful geometrical objects, where {\em connectedness}
of regions is typically a minimal
requirement~\cite{Borgo96,ijcai:Cohn&Renz08}.  (A connected region is
one which consists of a `single piece.')  It is easy to write
constraints in $\RCCE$ that are satisfiable by connected regular
closed sets over arbitrary topological spaces but not over $\R^2$; in
$\BRCCE$ we can even write formulas satisfiable by connected regular
closed sets over arbitrary spaces but not over $\R^n$ for any $n$.
Worse still: there exist very simple collections of spatial
constraints (involving connectedness) that are satisfiable in the
Euclidean plane, but only by `pathological' sets that cannot plausibly
represent the regions occupied by physical objects~\cite{ijcai:HSL2}.
Unfortunately, little is known about the complexity of topological
constraint satisfaction by non-pathological objects in low-dimensional
Euclidean spaces. One landmark result~\cite{ijcai:iscloes:sss03} in
this area shows that satisfiability of \RCCE-formulas by
\emph{disc-homeomorphs} in $\R^2$ is still \NP-complete, though the
decision procedure is vastly more intricate than in the general
case. In this paper, we investigate the computational properties of
more general and flexible spatial logics with connectedness
constraints interpreted over $\R^2$ and $\R^3$.

We consider two `base' topological
constraint languages.  The language $\cB$ features $=$ as its
only predicate, but has function symbols $+$, $-$, $\cdot$ denoting
the standard operations of fusion, complement and taking common parts
defined for regular closed sets, as well as the constants $1$ and $0$
for the entire space and the empty set. Our second base language,
$\cBC$, additionally features a binary predicate, $C$, denoting the
`contact' relation (two sets are in {\em contact} if they share at
least one point).  The language $\cBC$ is a notational variant
of~\BRCCE{} (and thus an extension of \RCCE), while $\cB$ is the
analogous extension of \RCCF{}. We add to $\cB$ and $\cBC$ one of
two new unary predicates: $c$, representing the property of
connectedness, and $\ic$, representing the (stronger) property of
having a connected \emph{interior}. We denote the resulting languages
by $\cBc$, $\cBci$\!, $\cBCc$ and $\cBCci$\!. We are interested 
in interpretations over ({\em i}) the regular closed sets of $\R^2$
and $\R^3$, and ({\em ii}) the regular closed \emph{polyhedral} sets
of $\R^2$ and $\R^3$.  (A set is polyhedral if it can be defined by
finitely many bounding hyperplanes.) By restricting interpretations to
polyhedra we rule out satisfaction by pathological sets and use the
same `data structure' as in GISs.

When interpreted over {\em arbitrary} topological spaces, the
complexity of reasoning with these languages is known: satisfiability
of $\cBci$-formulas is \NP-complete, while for the other three
languages, it is \ExpTime-complete.  Likewise, the 1D Euclidean case
is completely solved.  For the spaces $\R^n$ ($n \geq 2$), however,
most problems are still open.  All four languages contain
formulas satisfiable by regular closed sets in $\R^2$, but not by
regular closed polygons; in $\R^3$, the analogous result is known
only for $\cBci$ and $\cBCci$. The satisfiability problem for \cBc{},
\cBCc{} and \cBCci{} is \ExpTime-hard (in both polyhedral and
unrestricted cases) for $\R^n$ ($n \geq 2$); however, the only known
upper bound is that satisfiability of $\cBci$-formulas by
polyhedra in $\R^n$ ($n \geq 3$) is \ExpTime-complete.
(See~\cite{ijcai:kphz10} for a summary.)

This paper settles most of these open problems, revealing considerable
differences between the computational properties of constraint
languages with connectedness predicates when interpreted over $\R^2$
and over abstract topological spaces.  Sec.~\ref{sec:sensitivity}
shows that $\cBc$, $\cBci$, $\cBCc$ and $\cBCci$ are all sensitive to
restriction to polyhedra in $\R^n$ ($n \geq
2$). Sec.~\ref{sec:undecidability} establishes an unexpected result:
all these languages are \emph{undecidable} in 2D, both in the
polyhedral and unrestricted cases (\cite{Dornheim} proves
undecidability of the \emph{first-order} versions of these
languages). Sec.~\ref{sec:3d} resolves the open issue of the
complexity of $\cBci$ over regular closed sets (not just polyhedra) in
$\R^3$ by establishing an NP upper bound.  Thus, Qualitative Spatial
Reasoning in Euclidean spaces proves much more challenging if
connectedness of regions is to be taken into account.  We discuss the
obtained results in the context of spatial reasoning in
Sec.~\ref{conclusion}.  Omitted proofs can be found in the appendix. 


\section{Constraint Languages with Connectedness}\label{sec:preliminaries}

Let $T$ be a topological space. We denote the closure of any $X
\subseteq T$ by $\tc{X}$, its interior by $\ti{X}$ and its boundary by
$\delta X = \tc{X} \setminus \ti{X}$. We call $X$ {\em regular closed}
if $X = \tc{\ti{X}}$, and denote by $\RC(T)$ the set of regular closed
subsets of $T$. Where $T$ is clear from context, we refer to elements
of $\RC(T)$ as {\em regions}. $\RC(T)$ forms a Boolean algebra under
the operations $X + Y = X \cup Y$, $X \cdot Y =
\smash{\tc{\ti{\smash{(X \cap Y)}}}}$ and $-X = \tc{\smash{(T
    \setminus X)}}$. We write $X \leq Y$ for $X \cdot (-Y) = \emptyset$; thus
$X \leq Y$ iff $X \subseteq Y$.  A subset $X \subseteq T$ is
\emph{connected} if it cannot be decomposed into two disjoint,
non-empty sets closed in the subspace topology; $X$ is
\emph{interior-connected} if $\ti{X}$ is connected.

Any $(n-1)$-dimensional hyperplane in $\R^n$, $n \geq 1$, bounds two
elements of $\RC(\R^n)$ called \emph{half-spaces}. We denote by
$\RCP(\R^n)$ the Boolean subalgebra of $\RC(\R^n)$ generated by the
half-spaces, and call the elements of $\RCP(\R^n)$ (regular closed)
\emph{polyhedra}. If $n = 2$, we speak of (regular closed)
\emph{polygons}. Polyhedra may be regarded as `well-behaved' or, in
topologists' parlance, `\emph{tame}.'  In particular, every polyhedron
has finitely many connected components, a property which is not true
of regular closed sets in general.

The topological constraint languages considered here all employ a
countably infinite collection of variables $r_1, r_2, \ldots$ The
language $\cBC$ features binary predicates $=$ and $C$, together with
the individual constants $0$, $1$ and the function symbols $+$,
$\cdot$, $-$. The \emph{terms} $\tau$ and \emph{formulas} $\phi$ of
$\cBC$ are given by:
\begin{align*}
\tau \quad & ::=
\quad r \ \ \mid
\ \ \tau_1 + \tau_2 \ \ \mid
\ \ \tau_1 \cdot \tau_2 \ \ \mid
\ \ - \tau_1 \ \ \mid
\ \ 1 \ \ \mid
\ \ 0,\\
\phi \quad & ::= \quad \tau_1 =  \tau_2 \ \
\mid \ \ C(\tau_1,\tau_2) \ \
\mid \ \ \phi_1 \land \phi_2 \ \
\mid \ \ \neg \phi_1.
\end{align*}
The language $\cB$ is defined analogously, but without the predicate
$C$. If $S \subseteq \RC(T)$ for some topological space $T$, an
\emph{interpretation over} $S$ is a function $\cdot^\mathfrak{I}$
mapping variables $r$ to elements $r^\mathfrak{I} \in S$. We extend
$\cdot^\mathfrak{I}$ to terms $\tau$ by setting $0^\mathfrak{I} =
\emptyset$, $1^\mathfrak{I} = T$, $(\tau_1 + \tau_2)^\mathfrak{I} =
\tau_1^\mathfrak{I} + \tau_2^\mathfrak{I}$, etc. We write
$\mathfrak{I} \models \tau_1 = \tau_2$ iff $\tau_1^\mathfrak{I} =
\tau_2^\mathfrak{I}$, and $\mathfrak{I} \models C(\tau_1,\tau_2)$ iff
$\tau_1^\mathfrak{I} \cap \tau_2^\mathfrak{I} \neq \emptyset$.  We
read $C(\tau_1, \tau_2)$ as `$\tau_1$ \emph{contacts} $\tau_2$.'  The
relation $\models$ is extended to non-atomic formulas in the obvious
way. A formula $\phi$ is \emph{satisfiable over} $S$ if 
$\mathfrak{I} \models \phi$ for some interpretation $\mathfrak{I}$ over $S$. 

Turning to languages with connectedness predicates, we define $\cBc$
and $\cBCc$ to be extensions of $\cB$ and $\cBC$ with the unary
predicate $c$. We set $\mathfrak I \models c(\tau)$ iff
$\smash{\tau^{\mathfrak I}}$ is connected in the topological space under
consideration. Similarly, we define $\cBci$ and $\cBCci$ to be
extensions of $\cB$ and $\cBC$ with the predicate $\ic$, setting
$\mathfrak I \models \ic(\tau)$ iff $\ti{\smash{(\tau^{\mathfrak
      I})}}$ is connected. $\Sat(\mathcal{L},S)$ is the set of $\mathcal{L}$-formulas
satisfiable over $S$, where $\cL$
is one of $\cBc$, $\cBCc$, $\cBci$ or $\cBCci$ (the
topological space is implicit in this notation, but will always be
clear from context). We shall be concerned with $\Sat(\cL, S)$, where
$S$ is $\RC(\R^n)$ or $\RCP(\R^n)$ for $n=2,3$.

To illustrate, consider the $\cBci$-formulas $\phi_k$
given by
\begin{equation}
\bigwedge_{1 \leq i \leq k}\hspace*{-0.5em}  \bigl( \ic(r_i) \land (r_i \neq 0) \bigr)
   \land \bigwedge_{i < j} \bigl( \ic(r_i + r_j) \land (r_i \cdot r_j =0) \bigr).
\label{eq:ex1}
\end{equation}
One can show that $\phi_3$ is satisfiable over $\RCP(\R^n)$, $n \geq
2$, but not over $\RCP(\R)$, as no three intervals with non-empty,
disjoint interiors can be in pairwise contact. Also, $\phi_5$ is
satisfiable over $\RCP(\R^n)$, for $n\geq 3$, but not over
$\RCP(\R^2)$, as the graph $K_5$ is non-planar.  Thus, $\cBci$ is
sensitive to the dimension of the space.
Or again, consider the $\cBci$-formula 
\begin{equation}\label{eq:wiggly}
\bigwedge_{1 \leq i \leq 3}\hspace*{-0.5em} \ic(r_i) \ \ \land\ \ \ic(r_1 + r_2 + r_3) \ \ \land\ \ \bigwedge_{2 \leq i \leq 3} \neg\ic(r_1 + r_i).
\end{equation}
One can show that~\eqref{eq:wiggly} is satisfiable over $\RC(\R^n)$,
for any $n\ge 2$ (see, e.g., Fig.~\ref{fig:wiggly}), but not over
$\RCP(\R^n)$. Thus $\cBci$ is sensitive to tameness in Euclidean
spaces.
\begin{figure}[ht]
\begin{center}
\resizebox{3.1cm}{!}{\input{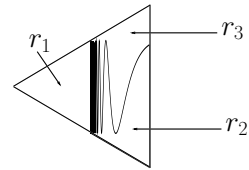}}
\end{center}\vspace*{-3mm}
\caption{Three regions in $\RC(\R^2)$ satisfying \eqref{eq:wiggly}.}
\label{fig:wiggly}
\end{figure}
It is known~\cite{ijcai:kphz10} that, for the Euclidean {\em plane},
the same is true of $\cBc$ and $\cBCc$: there is a $\cBc$-formula
satisfiable over $\RC(\R^2)$, but not over $\RCP(\R^2)$. (The 
example required to show this is far more complicated than the
\cBci-formula~\eqref{eq:wiggly}.)  In the next section, we prove that
any of $\cBc$, $\cBCc$ and $\cBCci$ contains formulas satisfiable
over $\RC(\R^n)$, for every $n \geq 2$, but only by regions with
infinitely many components. Thus, all four of our languages are
sensitive to tameness in all dimensions greater than one.


\section{Regions with Infinitely Many Components}\label{sec:sensitivity}

Fix $n \ge 2$ and let $d_0,d_1,d_2,d_3$ be regions partitioning $\R^n$:
\begin{align}
\label{eq:InfPart1}
\textstyle \big( \sum_{0 \leq i \le 3} d_i =1 \big) \quad\land\quad \bigwedge_{0 \leq i<j\leq3}(d_i\cdot d_j=0).	
%
\end{align}
We construct formulas forcing the $d_i$ to have infinitely many
connected components. To this end we require non-empty regions $a_i$
contained in $d_i$, and a non-empty region $t$:
\begin{align}\label{eq:basic-regions}
\textstyle\bigwedge_{0 \leq i \leq 3} \bigl((a_i \ne 0) \land (a_i \leq d_i)\bigr) \quad\land\quad (t\ne 0).
\end{align}
The configuration of regions we have in mind is depicted in
Fig.~\ref{fig:InfCmpSat}, where components of the $d_i$ are arranged
like the layers of an onion. The `innermost' component of $d_0$ is
surrounded by a component of $d_1$, which in turn is surrounded by a
component of $d_2$, and so on. The region $t$ passes through every
layer, but avoids the $a_i$. To enforce a configuration of this sort,
we need the following three formulas, for $0 \leq i \leq 3$:
\begin{align}
\label{eq:InfContact}	&c(a_i+d_{\md{i+1}}+t),	 \\
\label{notC}	&\neg C(a_i,d_{\md{i+1}}\cdot (-a_{\md{i+1}})) \ \ \land \ \ \neg C(a_i, t), \\
\label{eq:InfNTriv1}    &\neg C(d_i, d_{\md{i+2}}),
\end{align}
where $\md{k}= k\, {\rm mod}\, 4$. Formulas~\eqref{eq:InfContact} and
\eqref{notC} ensure that each component of $a_i$ is in contact with
$a_{\md{i+1}}$, while \eqref{eq:InfNTriv1} ensures that no component of
$d_i$ can touch any component of $d_{\md{i+2}}$.
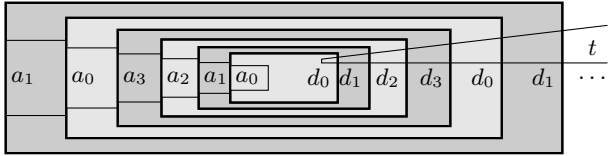
\begin{figure}[h]
\begin{center}
\begin{tikzpicture}
\small
	\coordinate (Px) at (-3.7,0);
	\coordinate (Qx) at (1.2,0);
	\coordinate (Py) at (0,1);
	\coordinate (Pby) at (0,.5);
	\foreach \i/\a/\b in {1/40/20,0/25/10,3/40/20,2/25/10,1/40/20,0/25/10}
	{
		\draw[thick,fill=black!\b] ($(Px)-(Py)$) rectangle ($(Py)-(Px)$);
		\draw[thin,fill=black!\b] ($(Px)-(Pby)$) rectangle ($(Qx)+(Pby)$);
		\draw[thick] ($(Px)-(Py)$) rectangle ($(Py)-(Px)$);
		\node at ($(Px)+(0.23,0)$) {$a_\i$};
		\node at ($(-0.25,0)-(Px)$) {$d_\i$};
		\coordinate (Px) at ($(Px)!0.15!(0,0)+(0.25,0)$);
		\coordinate (Py) at ($(Py)!0.2!(0,0)$);
		\coordinate (Pby) at ($(Pby)!0.2!(0,0)$);
		\coordinate (Qx) at ($(Qx)!0.15!(0,0)-(.2,0)$);
	}
	\draw (4.3,0.7)--(.5,0.25)--(.5,0.2)--(4.3,0.2);
	\node at (4.1,.4) {$t$};
	\node at (4.1,0) {$\ldots$};
\end{tikzpicture}		
\end{center}
\vspace*{-2mm}
\caption{Regions satisfying $\phi_\infty$
.}\label{fig:InfCmpSat}	
\end{figure}

Denote by $\phi_\infty$ the conjunction of the above
constraints. Fig.~\ref{fig:InfCmpSat} shows how $\phi_\infty$ can be
satisfied over $\RC(\R^2)$. By cylindrification, it is also satisfiable
over any $\RC(\R^n)$, for $n> 2$.

The arguments of this section are based on the following
property of regular closed subsets of Euclidean spaces:
\begin{lemma}\label{lma:ourNewman}
If $X \in \RC(\R^n)$ is connected, then every component of $-X$ 
has a connected boundary.
\end{lemma}

The proof of this lemma, which follows from a result
in~\cite{ijcai:Newman64}, can be found in Appendix~\ref{sec:sensitivityA}. 
The result fails for other familiar spaces such as the torus.
\begin{theorem}\label{theo:inftyCc}
There is a $\cBCc$-formula 
satisfiable over $\RC(\R^n)$, $n
\geq 2$, but not by regions with finitely many components.
\end{theorem}
\begin{proof}
Let $\phi_\infty$ be as above.
To simplify the presentation, we ignore the difference between variables
and the regions they stand for, writing, for example, $a_i$ instead of
$a_i^\mathfrak{I}$.  We construct a sequence of disjoint components
$X_i$ of $d_{\md{i}}$ and open sets $V_i$ connecting $X_i$ to
$X_{i+1}$ (Fig.~\ref{fig:InfCmpConstr}). By the first conjunct
of~\eqref{eq:basic-regions}, let $X_0$ be a component of $d_0$
containing points in $a_0$. Suppose $X_i$ has been constructed.
By~\eqref{eq:InfContact} and~\eqref{notC}, $X_i$ is in contact with
$a_{\md{i+1}}$. Using~\eqref{eq:InfNTriv1} and the fact that
$\R^n$ is locally connected, one can find a
component $X_{i+1}$ of $d_{\md{i+1}}$ which has points in $a_{i+1}$,
and a connected open set $V_i$ such that $V_i \cap X_i$ and $V_i \cap
X_{i+1}$ are non-empty, but $V_i \cap d_{\md{i+2}}$ is empty.
\begin{figure}[h]
\begin{center}
\begin{tikzpicture}
	\newcommand{\iterateboxes}[3][]
	{
		\coordinate (Plx) at (0,0);
		\coordinate (Prx) at (7.5,0);
		\coordinate (Py) at (0,1);
		#1
		\foreach \x/\y in {3/2,2/1,1/0}
		{
			#2
			
			\coordinate (Plx) at ($(Plx)+(.25,0)$);
			\coordinate (Prx) at ($(Prx)-(2,0)$);
			\coordinate (Py) at ($(Py)-(0,.25)$);
		}
		#3
	}
	\small
	\iterateboxes
	[\node at ($(.5,0)+(Prx)$) {$\ldots$};]
	{
		\draw[draw=black,fill=gray!30] ($(Plx)-(Py)$) rectangle ($(Prx)+(Py)$);
		\draw[draw=black,thick,fill=gray!0] ($(.125,.125)+(Plx)-(Py)$) rectangle ($(Prx)+(Py)-(1.45,.125)$);
		\node at ($(-.7,-0)+(Prx)$) {$X_\x$};
	}
	{
		\draw[draw=black,fill=gray!30] ($(Plx)-(Py)$) rectangle ($(Prx)+(Py)$);
		\node at ($(-.5,-0)+(Prx)$) {$X_0$};
	}
	\iterateboxes
	{
		\filldraw ($(Prx)-(2,0)$) coordinate(P);
		\filldraw ($(Prx)-(1.45,0)$) coordinate(Q);
		\draw ($(P)!.5!(Q)$) ellipse (.5 and 0.2);
		\node at ($(P)!.5!(Q)$) {$V_\y$};
	}{}
\end{tikzpicture}
\end{center}
\vspace*{-2mm}
\caption{The sequence $\{X_i,V_i\}_{i\geq0}$ generated by $\phi_\infty$.
($S_{i+1}$ and $R_{i+1}$ are the `holes' of $X_{i+1}$ containing $X_i$ and $X_{i+2}$.)
}
\label{fig:InfCmpConstr}
\end{figure}
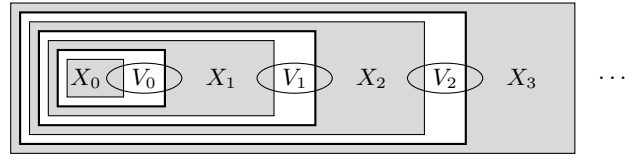

To see that the $X_i$ are distinct, let $S_{i+1}$ and $R_{i+1}$ be the
components of $-X_{i+1}$ containing $X_i$ and $X_{i+2}$,
respectively. It suffices to show $S_{i+1} \subseteq\ti{S}_{i+2}$.
Note that the connected set $V_i$ must intersect $\delta S_{i+1}$.
Evidently, $\delta S_{i+1} \subseteq X_{i+1} \subseteq d_{\md{i+1}}$.
Also, $\delta S_{i+1} \subseteq -X_{i+1}$; hence,
by~\eqref{eq:InfPart1} and~\eqref{eq:InfNTriv1}, $\delta S_{i+1}
\subseteq d_{i} \cup d_{\md{i+2}}$.  By Lemma~\ref{lma:ourNewman},
$\delta S_{i+1}$ is connected, and therefore, by~\eqref{eq:InfNTriv1},
is entirely contained either in $d_{\md{i}}$ or in
$d_{\md{i+2}}$. Since $V_i \cap \delta S_{i+1} \neq \emptyset$ and
$V_i \cap d_{\md{i+2}} = \emptyset$, we have $\delta S_{i+1} \not
\subseteq d_{\md{i+2}}$, so $\delta S_{i+1} \subseteq d_i$. Similarly,
$\delta R_{i+1}\subseteq d_{i+2}$.  By~\eqref{eq:InfNTriv1}, then,
$\delta S_{i+1} \cap \delta R_{i+1} = \emptyset$, and since $S_{i+1}$
and $R_{i+1}$ are components of the same set, they are
disjoint. Hence, $S_{i+1}\subseteq \ti{(-R_{i+1})}$, and since
$X_{i+2}\subseteq R_{i+1}$, also $S_{i+1}\subseteq
\ti{(-X_{i+2})}$. So, $S_{i+1}$ lies in the interior of
a component of $-X_{i+2}$, and since $\delta S_{i+1}\subseteq
X_{i+1}\subseteq S_{i+2}$, that component must be $S_{i+2}$.
\end{proof}

\vspace*{-2mm}

Now we show how the $\cBCc$-formula $\phi_\infty$ can be transformed
to $\cBCci$- and $\cBc$-formulas with similar properties.  Note first
that all occurrences of $c$ in $\phi_\infty$ have positive polarity.
Let $\ti{\phi}_\infty$ be the result of replacing them with the
predicate $\ic$. In Fig.~\ref{fig:InfCmpSat}, the connected regions
mentioned in~\eqref{eq:InfContact} are in fact interior-connected;
hence $\ti{\phi}_\infty$ is satisfiable over $\RC(\R^n)$. Since
interior-connectedness implies connectedness, $\ti{\phi}_\infty$
entails $\phi_\infty$, and we obtain:
\begin{corollary}\label{cor:inftyCci}
There is a $\cBCci$-formula 
satisfiable over $\RC(\R^n)$, $n
\geq 2$, but not by regions with finitely many components.
\end{corollary}

To construct a $\cBc$-formula, we observe
that all occurrences of $C$ in $\phi_\infty$ are negative. We eliminate
these using the predicate $c$. Consider, for example, the formula
$\neg C(a_i, t$) in~\eqref{notC}.
By inspection of Fig.~\ref{fig:InfCmpSat},
one can find regions $r_1$, $r_2$ satisfying
\begin{equation}
\label{eq:contactTrick}
c(r_1) \wedge c(r_2) \wedge (a_i \leq r_1) \wedge (t \leq r_2)
\wedge \neg c(r_1+r_2).
\end{equation}
On the other hand, \eqref{eq:contactTrick} entails $\neg C(a_i,t)$. By
treating all other non-contact relations similarly, we obtain a
$\cBc$-formula $\psi_\infty$ that is satisfiable over $\RC(\R^n)$, and
that entails $\phi_\infty$. Thus:
\begin{corollary}\label{cor:inftyBc}
There is a $\cBc$-formula satisfiable over $\RC(\R^n)$, $n \geq 2$,
but not by regions with finitely many components.
\end{corollary}

Obtaining a $\cBci$ analogue is complicated by the fact that we must
enforce non-contact constraints using $\ic$ (rather than $c$). In the
Euclidean plane, this can be done using \emph{planarity constraints};
see Appendix~\ref{sec:sensitivityA}.
\begin{theorem}\label{theo:inftyBci}
There is a $\cBci$-formula satisfiable over $\RC(\R^2)$, but not by
regions with finitely many components.
\end{theorem}

Theorem~\ref{theo:inftyCc} and Corollary~\ref{cor:inftyBc} entail that, if
$\cL$ is $\cBc$ or $\cBCc$, then $\Sat(\cL,\RC(\R^n)) \neq
\Sat(\cL,\RCP(\R^n))$ for $n \geq 2$.  Theorem~\ref{theo:inftyBci}
fails for $\RC(\R^n)$ with $n\geq 3$ (Sec.~\ref{sec:3d}).  However, we
know from~\eqref{eq:wiggly} that $\Sat(\cBci,\RC(\R^n)) \neq
\Sat(\cBci,\RCP(\R^n))$ for all $n \geq 2$. Theorem~\ref{theo:inftyCc}
fails in the 1D case; moreover, $\Sat(\cL,\RC(\R))=\Sat(\cL,\RCP(\R))$
only in the case $\cL = \cBc$ or $\cBci$~\cite{ijcai:kphz10}.



\section{Undecidability in the Plane}\label{sec:undecidability}

Let $\cL$ be any of $\cBc$, $\cBCc$, $\cBci$ or $\cBCci$. In this
section, we show, via a reduction of the {\em Post correspondence
  problem} (PCP), that $\Sat(\cL,\RC(\R^2))$ is r.e.-hard, and
$\Sat(\cL,\RCP(\R^2))$ is r.e.-complete. An {\em instance} of the PCP
is a quadruple $\fW = (S, T, \fw_1, \fw_2)$ where $S$ and $T$ are
finite alphabets, and each $\fw_i$ is a word morphism from $T^*$ to
$S^*$. We may assume that $S = \set{0,1}$ and $\fw_i(t)$ is non-empty
for any $t \in T$. The instance $\fW$ is {\em positive} if there
exists a non-empty $\tau \in T^*$ such that $\fw_1(\tau) = \fw_2(\tau)$. The set
of positive PCP-instances is known to be r.e.-complete. The reduction
can only be given in outline here: full details are given in 
Appendix~\ref{sec:UndecidabilityB}.


To deal with arbitrary regular closed subsets of $\RC(\R^2)$, we use
the technique of `wrapping' a region inside two bigger ones. Let us
say that a \emph{3-region} is a triple $\tseq{a} =
(a,\intermediate{a},\inner{a})$ of elements of $\RC(\R^2)$ such that
$0 \neq \inner{a} \ll \intermediate{a} \ll a$, where $r \ll s$
abbreviates $\neg C(r, -s)$. It helps to think of $\tseq{a} =
(a,\intermediate{a},\inner{a})$ as consisting of a kernel,
$\inner{a}$, encased in two protective layers of shell. As a simple
example, consider the sequence of 3-regions $\tseq{a}_1, \tseq{a}_2,
\tseq{a}_3$ depicted in Fig.~\ref{fig:stack}, where the inner-most
regions form a sequence of externally touching polygons.
\begin{figure}[h]
\begin{center}
\begin{tikzpicture}[	s0/.style={fill=Gray!40,fill opacity=0.5},			
						s1/.style={dashed,fill=Gray!20,fill opacity=0.3},			
						s2/.style={dotted,fill=white,fill opacity=0}] 		
	
	\newcount\mod
	\newcount\jj
	
	\coordinate (H0) at (2,0); 
	\coordinate (V0) at (0,0.8); 
	\coordinate (H1) at (2,0); 
	\coordinate (V1) at (0,.4); 
	
	\coordinate (HW) at (.4,0); 
	\coordinate (VW) at (0,.2); 
	\coordinate (P) at (0,0);
	
	\foreach \j in {2,1,0}
	{
	   \foreach \i in {1,2,3}
		{			
		    \pgfmathsetcount{\mod}{mod(\i,2)}
			\coordinate (H) at ($\j*(0.2,0)$); 
			\coordinate (V) at ($\j*(0,0.33)$); 
			
			\filldraw[s\j] ($(P)-(H)-(V)-.5*(V\the\mod)+(VW)$)
				--++($(V\the\mod)+2*(V)-2*(VW)$)
				--++($(HW)+(VW)$)
				--++($(H\the\mod)+2*(H)-2*(HW)$)
				--++($(HW)-(VW)$)
				--++($-1*(V\the\mod)-2*(V)+2*(VW)$)
				--++($-1*(HW)-(VW)$)
				--++($-1*(H\the\mod)-2*(H)+2*(HW)$)
				--++($-1*(HW)+(VW)$);			
			\ifnum \j=0 \node at ($(P)-0.5*(V\the\mod)+0.5*(H\the\mod)+(0,.2)$) {\small $\inner{a}_\i$}; \fi
			\ifnum \j=1 \node at ($(P)-0.5*(V\the\mod)+0.5*(H\the\mod)-.5*(V)+(0.,-0.05)$) {\small $\intermediate{a}_\i$}; \fi
			\ifnum \j=2 \node at ($(P)-0.5*(V\the\mod)+0.5*(H\the\mod)-.75*(V)+(0.,-0.05)$) {\small $a_\i$}; \fi
		    \coordinate (P) at ($(P)+(H\the\mod)$);
		}
	\coordinate (P) at (0,0);
	}
\end{tikzpicture}
\end{center}
\vspace*{-2mm}
\caption{A chain of 3-regions satisfying $\stack(\tseq{a}_1,
\tseq{a}_2, \tseq{a}_3)$.
}\label{fig:stack}
\end{figure}
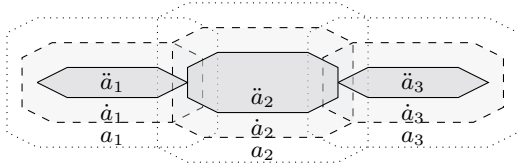
When describing arrangements of 3-regions, we use the variable
$\tseq{r}$ for the triple of variables $(r, \intermediate{r},
\inner{r})$, taking the conjuncts $\inner{r} \neq 0$, $\inner{r} \ll
\intermediate{r}$ and $\intermediate{r} \ll r$ to be implicit. As with
ordinary variables, we often ignore the difference between 3-region
variables and the 3-regions they stand for.

For $k \geq 3$, define the formula $\stack(\tseq{a}_1, \ldots,
\tseq{a}_k)$ by 
\begin{equation*}
\bigwedge_{1 \leq i \leq k}
     c(\intermediate{a}_i + \inner{a}_{i + 1} + \cdots + \inner{a}_k) \ \  \ \ 
\land \ \ 
\bigwedge_{j - i > 1} \neg C(a_i,a_j).
\end{equation*}
Thus, the triple of 3-regions in Fig.~\ref{fig:stack} satisfies
$\stack(\tseq{a}_1, \tseq{a}_2, \tseq{a}_3)$. This formula plays a
crucial role in our proof.  If $\stack(\tseq{a}_1,\ldots,\tseq{a}_k)$
holds, then any point $p_0$ in the inner shell $\intermediate{a}_1$ of
$\tseq{a}_1$ can be connected to any point $p_k$ in the kernel
$\inner{a}_k$ of $\tseq{a}_k$ via a Jordan arc
$\gamma_1\cdots\gamma_k$ whose $i$th segment, $\gamma_i$, never leaves
the outer shell $a_i$ of $\tseq{a}_i$. Moreover, each $\gamma_i$
intersects the inner shell $\intermediate{a}_{i+1}$ of
$\tseq{a}_{i+1}$, for $1 \leq i <k$.

This technique allows us to write $\cBCc$-formulas whose satisfying
regions are guaranteed to contain various networks of arcs, exhibiting
almost any desired pattern of intersections. Now recall the
construction of Sec.~\ref{sec:sensitivity}, where constraints on the
variables $d_0, \ldots, d_3$ were used to enforce `cyclic' patterns of
components. Using $\stack(\tseq{a}_1,\ldots,\tseq{a}_k)$, we can
write a formula with the property that the regions in any satisfying
assignment are forced to contain the pattern of arcs having the form
shown in Fig.~\ref{fig:Summary1}.
\begin{figure}[h]
\begin{center}
\begin{tikzpicture}[>=latex]
\draw (0.3,0.2) rectangle +(7.9,2.1);
\draw[thick] (0.3,0.9) -- ++(6,0);
\foreach \x/\y in {1/1,2/2,3/3,6/{n}}
{
    \draw[thick] (\x + 0.1,0.9) -- ++(0,1.4);
    \filldraw[black] (\x + 0.3,0.9) circle(0.04);
    \draw[ultra thin] (\x + 0.3,0.9) -- ++ (0,-0.25);
    \draw[<->,dashed] (\x - 0.7,0.75) to node [label=below:{\small $\zeta_{\y}$}] {} (\x + 0.3,0.75);
    \draw[<->,dashed] (\x - 0.05,0.9) to node [label=left:{\small $\eta_{\y}$}] {} (\x - 0.05,2.3);
}
\filldraw[black] (0.3,0.9) circle(0.04);
\filldraw[black] (5.3,0.9) circle(0.04);
\draw[ultra thin] (5.3,0.8) -- ++ (0,-0.25);
\end{tikzpicture}
\end{center}
\vspace*{-2mm}
\caption{Encoding the PCP: Stage 1.}
\label{fig:Summary1}
\end{figure}
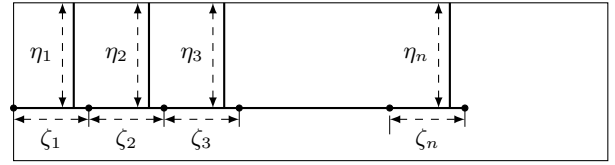
%
%
These arcs define a `window,' containing a sequence $\set{\zeta_i}$ of
`horizontal' arcs ($1 \leq i \leq n$), each connected by a
corresponding `vertical arc,' $\eta_i$, to some point on the `top
edge.' We can ensure that each $\zeta_i$ is included in a region
$a_{\md{i}}$, and each $\eta_i$ ($1 \leq i \leq n$) in a region
$b_{\md{i}}$, where $\md{i}$ now indicates $i \ \mbox{mod}\ 3$. By
repeating the construction, a second pair of arc-sequences,
$\set{\zeta'_i}$ and $\set{\eta'_i}$ ($1 \leq i \leq n'$) can be
established, but with each $\eta'_i$ connecting $\zeta'_i$ to the
`bottom edge.' Again, we can ensure each $\zeta'_i$ is included in a
region $a'_{\md{i}}$ and each $\eta'_i$ in a region $b'_{\md{i}}$ ($1 \leq
i \leq n'$). Further, we can ensure that the final horizontal arcs
$\zeta_n$ and $\zeta'_{n'}$ (but no others) are joined by an arc
$\zeta^*$ lying in a region $z^*$.
\begin{figure}[h]
\centering
\begin{tikzpicture}[>=latex]
\draw (0.3,0.2) rectangle +(7.9,2.1);
\draw[gray,thick] (0.3,0.8) -- ++(6,0);
\draw (6.3,1.6) -- (7,1.6) -- (7,0.8) -- (6.3,0.8);
\draw[thick] (0.3,1.6) -- ++(6,0);
\foreach \x/\y in {1/1,2/2,3/3,6/{n}}
{
    \draw[thick,gray] (\x + 0.1,0.8) -- ++(0,1.5);
    \draw[thick] (\x - 0.1,0.2) -- ++(0,1.4);
    \filldraw[gray] (\x + 0.3,0.8) circle(0.04);
    \filldraw[black] (\x + 0.3,1.6) circle(0.04);
    \draw[ultra thin] (\x + 0.3,1.6) -- ++ (0,0.25);
    \draw[<->,dashed] (\x + 0.3,1.75) to node [label=above:{\footnotesize $\zeta'_{\y}$}] {} (\x - 0.7,1.75);
    \draw[<->,dashed] (\x - 0.25,1.6) to node {} (\x - 0.25,0.2);
    \node at (\x-0.45,1.15) {{\footnotesize $\eta'_{\y}$}};
}
\filldraw[gray] (0.3,0.8) circle(0.04);
\filldraw[gray] (5.3,0.8) circle(0.04);
\filldraw[black] (0.3,1.6) circle(0.04);
\filldraw[black] (5.3,1.6) circle(0.04);
\draw[ultra thin] (5.3,1.6) -- ++ (0,0.25);
%
\node (last) at (6.9,1.3) [label=right:{\small $\zeta^*$}] {};
\end{tikzpicture}
\caption{Encoding the PCP: Stage 2.}
\label{fig:Summary2}
\end{figure}
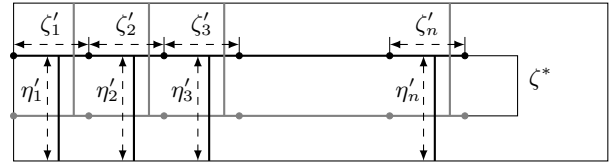
%
%
%
%
%
%
The crucial step is to match up these arc-sequences.  To do so, we
write $\neg C(a'_i, b_j) \wedge \neg C(a_i, b'_j) \wedge \neg C(b_i +
b'_i, b_j + b'_j + z^*)$, for all $i$, $j$ ($0 \leq i, j < 3$, $i \neq
j$). A simple argument based on planarity considerations then ensures
that the upper and lower sequences of arcs must cross (essentially) as
shown in Fig.~\ref{fig:Summary2}. In particular, we are guaranteed
that $n = n'$ (without specifying the value $n$), and that, for all $1
\leq i \leq n$, $\zeta_i$ is connected by $\eta_i$ (and also by
$\eta'_i$) to $\zeta'_i$.

Having established the configuration of Fig.~\ref{fig:Summary2}, we
write $(b_i \leq l_0 + l_1) \wedge \neg C(b_i \cdot l_0, b_i \cdot
l_1)$, for $0\leq i < 3$, ensuring that each $\eta_i$ is included in
exactly one of $l_0$, $l_1$. These inclusions naturally define a word
$\sigma$ over the alphabet $\set{0,1}$.  Next, we write
$\cBCc$-constraints which organize the sequences of arcs
$\set{\zeta_i}$ and $\set{\zeta'_i}$ (independently) into consecutive
blocks. These blocks of arcs can then be put in 1--1 correspondence
using essentially the same construction used to put the individual
arcs in 1--1 correspondence. Each pair of corresponding blocks can now
be made to lie in exactly one region from a collection $t_1, \ldots,
t_\ell$. We think of the $t_j$ as representing the letters of the
alphabet $T$, so that the labelling of the blocks with these elements
defines a word $\tau \in T^*$. It is then straightforward to write
non-contact constraints involving the arcs $\zeta_i$ ensuring that
$\sigma = \fw_1(\tau)$ and non-contact constraints involving the arcs
$\zeta'_i$ ensuring that $\sigma = \fw_2(\tau)$. Let $\phi_\fW$ be the
conjunction of all the foregoing $\cBCc$-formulas. Thus, if $\phi_\fW$
is satisfiable over $\RC(\R^2)$, then $\fW$ is a positive instance of
the PCP. On the other hand, if $\fW$ is a positive instance of the
PCP, then one can construct a tuple satisfying $\phi_\fW$ over
$\RCP(\R^2)$ by `thickening' the above collections of arcs into
polygons in the obvious way. So, $\fW$ is positive iff $\phi_\fW$ is
satisfiable over $\RC(\R^2)$ iff $\phi_\fW$ is satisfiable over
$\RCP(\R^2)$. This shows r.e.-hardness of $\Sat(\cBCc,\RC(\R^2))$ and
$\Sat(\cBCc,\RCP(\R^2))$.  Membership of the latter problem in r.e.~is
immediate because all polygons may be assumed to have vertices with
rational coordinates, and so may be effectively enumerated.  Using the
techniques of Corollaries~\ref{cor:inftyCci}--\ref{cor:inftyBc} and
Theorem~\ref{theo:inftyBci}, we obtain:
\begin{theorem}
\label{theo:undecidable}
For $\mathcal{L}\in \{\cBci, \cBc, \cBCci, \cBCc\}$,
$\Sat(\mathcal{L},\RC(\R^2))$ is r.e.-hard, and
$\Sat(\mathcal{L},\RCP(\R^2))$ is r.e.-complete.
\end{theorem}

The complexity of $\Sat(\mathcal{L},\RC(\R^3))$ remains open for the
languages $\mathcal{L}\in \{\cBc, \cBCci, \cBCc\}$. However, as we
shall see in the next section, for $\cBci$ it drops dramatically.



\section{\cBci{} in 3D}\label{sec:3d}

In this section, we consider the complexity of satisfying
\cBci-constraints by polyhedra and regular closed sets in
three-dimensional Euclidean space. Our analysis rests on an important
connection between geometrical and graph-theoretic interpretations. We
begin by briefly discussing the results of~\cite{ijcai:kp-hwz10} for
the {\em polyhedral} case.

Recall that every partial order $(W,R)$, where $R$ is a transitive and
reflexive relation on $W$, can be regarded as a topological space by
taking $X \subseteq W$ to be open just in case $x \in X$ and $xRy$
imply $y\in X$. Such topologies are called \emph{Aleksandrov spaces}.
If $(W,R)$ contains no proper paths of length greater than 2, we call
$(W,R)$ a \emph{quasi-saw} (Fig.~\ref{fig:broom}).  If, in
addition, no $x \in W$ has more than two proper $R$-successors, we
call $(W,R)$ a \emph{$2$-quasi-saw}.  The properties of 2-quasi-saws
we need are as follows~\cite{ijcai:kp-hwz10}:
\begin{itemize}\itemsep=0pt
\item[--] satisfiability of $\cBc$-formulas in arbitrary topological
  spaces coincides with satisfiability in 2-quasi-saws, and is
  \ExpTime-complete;

\item[--] $X \subseteq W$ is connected in a 2-quasi-saw $(W,R)$ iff it is interior-connected in $(W,R)$.
\end{itemize}
The following construction lets us apply these results to the problem
$\Sat(\cBci,\RCP(\R^3))$.  Say that a \emph{connected partition} in
$\RCP(\R^3)$ is a tuple $X_1,\dots,X_k$ of non-empty polyhedra having
connected and pairwise disjoint interiors, which sum to the entire
space $\R^3$. The \emph{neighbourhood graph} $(V,E)$ of this partition
has vertices $V = \{X_1, \ldots, X_k\}$ and edges $E = \{\{X_i, X_j\}
\mid i \ne j \text{ and } \ti{(X_i + X_j)} \text{ is connected}\}$
(Fig.~\ref{fig:conn-part}).
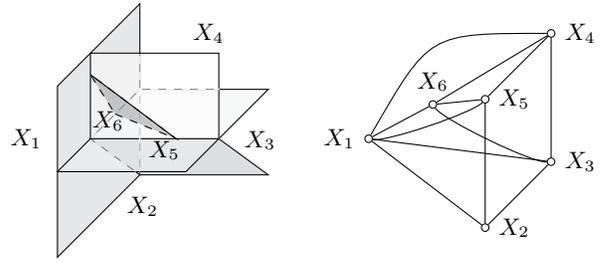
\begin{figure}[h]
\begin{center}
\begin{tikzpicture}[point/.style={circle,draw=black,minimum size=1mm,inner sep=0pt},scale=0.2mm]
\filldraw[fill=Gray!20] (0,-2,-3) -- (0,2,-3) -- (0,2,2) -- (0,-2,2) --   cycle;
\filldraw[fill=Gray!60,fill opacity=0.5] (0,0,0) -- (3,0,0) -- (3,-2,-3) -- (0,-2,-3) --   cycle;
\filldraw[fill=Gray!10,fill opacity=0.5] (0,0,-3) -- (0,0,2) -- (3,0,2) -- (3,0,-3) --   cycle;
\filldraw[fill=white,fill opacity=0.5] (0,0,0) -- (3,0,0) -- (3,2,0) -- (0,2,0) --   cycle;
\begin{scope}[dashed,draw=white]
\clip (0,0,0) -- (3,0,0) -- (3,2,0) -- (0,2,0) --   cycle;
\draw (0,-2,-3) -- (0,2,-3);
\draw (0,0,2) -- (0,0,-3) -- (3,0,-3);
\end{scope}
\begin{scope}[dashed,draw=white]
\clip (0,0,-3) -- (0,0,2) -- (3,0,2) -- (3,0,-3) --   cycle;
\draw (0,-2,-3) -- (0,2,-3);
\draw (0,-2,-3) -- (0,0,0);
\end{scope}
\filldraw[dashed,fill=Gray,fill opacity=0.5] (0,0,-1.5) -- (2,0,0) -- (0,1.5,0)  --   cycle;
\draw (2,0,0) -- (0,1.5,0);
\node (1) at (-1.5,0,0) {\small $X_1$};
\node (2) at (1.8,-1,1.5) {\small $X_2$};
\node (3) at (3,-1,-2.5) {\small $X_3$};
\node (4) at (1.8,1.5,-2.5) {\small $X_4$};
\node (5) at (2.5,.5,2) {\small $X_5$};
\node (6) at (0.3,0.3,-0.3) {\small $X_6$};
\begin{scope}[xshift=80mm]
\node [label=left:{\small $X_1$}] (v1) at (-1.5,0,0) [point] {};
\node [label=right:{\small $X_2$}](v2) at (1.8,-1.5,1.5)[point] {};
\node [label=right:{\small $X_3$}](v3) at (1.8,-1.5,-2.5)[point] {};
\node [label=right:{\small $X_4$}](v4) at (1.8,1.5,-2.5)[point] {};
\node [label=right:{\small $X_5$}](v5) at (1.8,1.5,1.5)[point] {};
\node [label=above:{\small $X_6$}] (v6) at (0,0.8,0)[point] {};
\draw (v1) -- (v2);
\draw (v2) -- (v3);
v\draw (v1) -- (v3);
\draw (v3) -- (v4);
\draw (v2) -- (v5);
\draw (v1) to [bend right, looseness=0.3] (v5);
\draw (v1) to [bend left, looseness=1.5] (v4);
\draw (v3) to [bend left, looseness=0.3] (v6);
\draw (v4) -- (v5);
\draw (v1) -- (v6);
\draw (v4) -- (v6);
\draw (v5) -- (v6);
\end{scope}
\end{tikzpicture}
\end{center}
\vspace*{-2mm}
\caption{A connected partition and its neighbourhood graph.}\label{fig:conn-part}
\end{figure}
One can show that {\em every} 
connected graph is the neighbourhood graph of
some connected partition in $\RCP(\R^3)$. Furthermore, every
neighbourhood graph $(V,E)$ gives rise to a 2-quasi-saw, namely, $(W_0
\cup W_1, R)$, where $W_0 = V$, $W_1 = \{z_{x,y} \mid \{x,y\} \in
E\}$, and $R$ is the reflexive closure of $\{(z_{x,y}, x), (z_{x,y},
y) \mid \{x,y\} \in E \}$. From this, we see that ({\em i}) a
$\cBci$-formula $\varphi$ is satisfiable over $\RCP(\R^3)$ iff ({\em
  ii}) $\varphi$ is satisfiable over a connected $2$-quasi-saw iff ({\em
  iii}) the $\cBc$-formula $\phi^{\bullet}$, obtained from $\phi$ by
replacing every occurrence of $\ic$ with $c$, is satisfiable over a
connected 2-quasi-saw. Thus, $\Sat(\cBci, \RCP(\R^3))$ is
\ExpTime-complete.

The picture changes if we allow variables to range over $\RC(\R^3)$
rather than $\RCP(\R^3)$. Note first that the $\cBci$-formula
\eqref{eq:wiggly} is not satisfiable over 2-quasi-saws, but has a
quasi-saw model as in Fig.~\ref{fig:broom}.
\begin{figure}[ht]
\centering\begin{tikzpicture}[>=latex, point/.style={circle,draw=black,minimum size=1mm,inner sep=0pt}]
\node (r1) at (2,0.8)  [point,label=left:{$x_1$}] {};
\node (r2) at (4,0.8)  [point,label=right:{$x_2$}] {};
\node (r3) at (6,0.8)  [point,label=right:{$x_3$}] {};
\node (r) at (4,0)  [point,label=below:{$z$}] {};
\draw[->] (r) to node [below] {\scriptsize $R$} (r1);
\draw[->] (r) to node [right] {\scriptsize $R$} (r2);
\draw[->] (r) to node [below] {\scriptsize $R$} (r3);
\node at (8,0) {\small $W_1$ = depth 1};
\node at (8,0.8) {\small $W_0$ = depth 0};
\end{tikzpicture}
\vspace*{-3mm}
\caption{A quasi-saw model $\mathfrak{I}$ of~\eqref{eq:wiggly}: $r_i^\mathfrak{I} = \{x_i,z\}$.}
\label{fig:broom}
\end{figure}
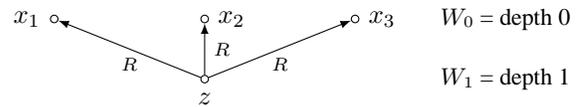
Some extra geometrical work will show now that ({\em iv}) a
$\cBci$-formula is satisfiable over $\RC(\R^3)$ iff ({\em v}) it is
satisfiable over a connected quasi-saw.  And as shown
in~\cite{ijcai:kp-hwz10}, satisfiability of $\cBci$-formulas in
connected spaces coincides with satisfiability over connected
quasi-saws, and is \NP-complete.
\begin{theorem}\label{theo:BciRCR3}
The problem $\Sat(\cBci,\RC(\R^3))$ is \NP-complete.
\end{theorem}
\begin{proof}
From the preceding discussion, it suffices to show that ({\em v})
implies ({\em iv}) for any $\cBci$-formula $\phi$. So suppose
$\mathfrak A \models \varphi$, with $\mathfrak A$ based on a finite
connected quasi-saw $(W_0\cup W_1,R)$, where $W_i$ contains all points
of depth $i \in \{0,1\}$ (Fig.~\ref{fig:broom}).  Without loss of
generality we will assume that there is a special point $z_0$ of depth
1 such that $z_0 R x$ for all $x$ of depth 0.  We show how
$\mathfrak{A}$ can be embedded into $\RC(\R^3)$.

Take pairwise disjoint \emph{closed} balls $B^1_x$, for $x$ of depth
0, and pairwise disjoint \emph{open} balls $D_z$, for all $z$ of depth
1 except $z_0$ (we assume the $D_z$ are disjoint from the
$B^1_x$). Let $D_{z_0}$ be the closure of the complement of all
$B_x^1$ and $D_z$.

We expand the $B^1_x$ to sets $B_x$ in such a way that
\begin{itemize}\itemsep=0pt
\item[(A)] the $B_x$ form a connected partition in $\RC(\R^3)$, that
  is, they are regular closed and sum up to $\R^3$, and their
  interiors are non-empty, connected and pairwise disjoint;

\item[(B)] every point in $D_z$ is either
in the interior of some $B_x$ with $zRx$, or on the boundary of \emph{all} of the $B_x$ with $zRx$.
\end{itemize}
The required $B_x$ are constructed as follows.
Let $q_1,
q_2, \ldots$ be an enumeration of all the points in the interiors of $D_z$ with \emph{rational} coordinates.
For $x\in W_0$, we set $B_x$ to be the closure of the infinite union
$\bigcup_{k=1}^\infty \ti{(B_x^k)}$, where the regular closed sets
$B_x^k$ are defined inductively as follows (Fig.~\ref{fig:apollonian}).
Assuming that the $B^k_x$ are defined, let $q_i$ be the first point in
the list $q_1, q_2, \ldots$ that is not in any $B^k_x$ yet. So, $q_i$ is in the interior of some $D_z$. Take an open ball
$C_{q_i}$ in the interior of $D_z$ centred in $q_i$ and
disjoint from the $B^k_x$. For each $x\in W_0$ with $zRx$, expand
$B_x^k$ by a closed ball in $C_{q_i}$ and a closed `rod' connecting it
to $B_x^1$ in such a way that the ball and the rod are disjoint from
the rest of the $B^k_x$; the result is denoted by $B^{k+1}_x$.
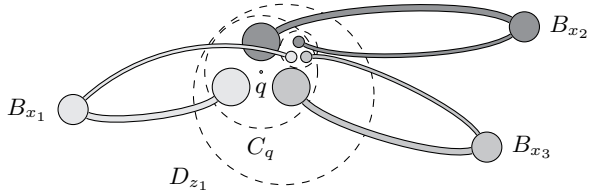
\begin{figure}[h]
\begin{center}
\begin{tikzpicture}[clball/.style={circle,draw=black,minimum size=4mm,inner sep=0pt},
opball/.style={circle,dashed,draw=black,minimum size=25mm,inner sep=0pt}]
\node [label=left:{\small $B_{x_1}$}] (x1) at (0,-3)[clball,fill=Gray!20] {};
\node [label=right:{\small $B_{x_2}$}] (x2) at (6,-1.9)[clball,fill=Gray] {};
\node [label=right:{\small $B_{x_3}$}] (x3) at (5.5,-3.5)[clball,fill=Gray!50] {};
%
\node [label=below left:{\small $D_{z_1}$}] (z1) at (2.8,-2.8)[opball,minimum size=24mm] {};
\node [label=below:{\small $C_q$}](c) at (2.5,-2.5)[opball,minimum size=15mm] {};
\node [label=below:{\small $q$}](q) at (2.5,-2.5)[circle,inner sep=0pt,minimum size=1,draw=black] {};
\node (xc1) at (2.1,-2.7)[clball,fill=Gray!20,minimum size=5mm] {};
\node (xc2) at (2.5,-2.1)[clball,fill=Gray,minimum size=5mm] {};
\node (xc3) at (2.9,-2.7)[clball,fill=Gray!50,minimum size=5mm] {};
\draw[double=Gray!20,double distance=2pt] (xc1) to [bend left, looseness=0.5] (x1);
\draw[double=Gray,double distance=2pt] (xc2) to [bend left, looseness=0.5] (x2);
\draw[double=Gray!50,double distance=2pt] (xc3) to [bend right, looseness=0.5] (x3);
\node (cp) at (3,-2.2)[opball,minimum size=5mm] {};
\node (xcp1) at (2.9,-2.3)[clball,fill=Gray!20,minimum size=1.5mm] {};
\node (xcp2) at (3,-2.1)[clball,fill=Gray,minimum size=1.5mm] {};
\node (xcp3) at (3.1,-2.3)[clball,fill=Gray!50,minimum size=1.5mm] {};
\draw[double=Gray!20,double distance=1pt] (xcp1) to [bend right, looseness=0.9] (x1);
\draw[double=Gray,double distance=1pt] (xcp2) to [bend right, looseness=0.5] (x2);
\draw[double=Gray!50,double distance=1pt] (xcp3) to [bend left, looseness=0.5] (x3);
\end{tikzpicture}
\end{center}
\vspace*{-2mm}
\caption{Filling $D_{z_1}$ with $B_{x_i}$, for $z_1 R x_i$, $i = 1,2,3$.}\label{fig:apollonian}
\end{figure}
Consider a function $f$ that maps regular closed sets $X \subseteq W$ to
$\RC(\R^3)$ so that $f(X)$ is the union of all $B_x$, for $x$ of depth
$0$ in $X$.  By~(A), $f$ preserves $+$, $\cdot$, $-$, $0$ and $1$.
Define an interpretation $\mathfrak{I}$ over $\RC(\R^3)$ by
$r^\mathfrak{I} = f(r^\mathfrak{A})$. To show that $\mathfrak I
\models \varphi$, it remains to prove that $\ti{X}$ is connected iff
$\ti{(f(X))}$ is connected (details are in Appendix~\ref{sec:Bci3D_C}).
\end{proof}


The remarkably diverse computational behaviour of \cBci{} over
$\RC(\R^3)$, $\RCP(\R^3)$ and $\RCP(\R^2)$ can be explained as
follows. To satisfy a \cBci-formula $\phi$ in $\RC(\R^3)$, it suffices
to find polynomially many points in the regions mentioned in $\phi$
(witnessing non-emptiness or non-internal-connectedness constraints),
and then to `inflate' those points to (possibly internally connected)
regular closed sets using the technique of Fig.~\ref{fig:apollonian}.
By contrast, over $\RCP(\R^3)$, one can write a \cBci-formula
analogous to \eqref{eq:contactTrick} stating that two internally
connected polyhedra do not share a 2D face.  Such
`face-contact' constraints can be used to generate constellations of
exponentially many polyhedra simulating runs of alternating Turing
machines on polynomial tapes, leading to \ExpTime-hardness. Finally,
over $\RCP(\R^2)$, planarity considerations endow \cBci{} with the
extra expressive power required to enforce full non-contact constructs
(not possible in higher dimensions), and thus to encode the PCP as
sketched in Sec.~\ref{sec:undecidability}.

%
%


\section{Conclusion}\label{conclusion}

This paper investigated topological constraint
languages featuring connectedness predicates and Boolean operations on
regions.  Unlike their less expressive cousins, \RCCE{} and \RCCF,
such languages are highly sensitive to the spaces over which they
are interpreted, and exhibit more challenging computational
behaviour. Specifically, we demonstrated that the languages $\cBCc$,
$\cBCci$ and $\cBc$ contain formulas satisfiable over $\RC(\R^n)$, $n
\geq 2$, but only by regions with infinitely many components. Using a
related construction, we proved that the satisfiability problem for
any of $\cBc$, $\cBCc$, $\cBci$ and $\cBCci$, interpreted either over
$\RC(\R^2)$ or over its polygonal subalgebra, $\RCP(\R^2)$, is
\emph{undecidable}. Finally, we showed that the satisfiability problem
for $\cBci$, interpreted over $\RC(\R^3)$, is \NP-complete, which
contrasts with \ExpTime-completeness for $\RCP(\R^3)$.  The complexity
of satisfiability for $\cBc$, $\cBCc$ and $\cBCci$ over $\RC(\R^n)$ or
$\RCP(\R^n)$ for $n \geq 3$ remains open. 
The obtained results rely on certain
distinctive topological properties of Euclidean spaces. 
Thus, for example, the argument of
Sec.~\ref{sec:sensitivity} is based on the property of
Lemma~\ref{lma:ourNewman}, while Sec.~\ref{sec:undecidability}
similarly relies on {\em planarity} considerations. In both cases,
however, the moral is the same: the topological spaces of most
interest for Qualitative Spatial Reasoning exhibit special
characteristics which any topological constraint language able to
express connectedness must take into account.

The results of Sec.~\ref{sec:undecidability} pose a challenge for
Qualitative Spatial Reasoning in the Euclidean plane.  On the one hand, the
relatively low complexity of \RCCE{} over disc-homeomorphs suggests
the possibility of usefully extending the expressive power of \RCCE{}
without compromising computational properties. On the other hand, our
results impose severe limits on any such extension.  We observe,
however, that the constructions used in the proofs depend on a strong
interaction between the connectedness predicates and the Boolean
operations on regular closed sets. We believe that by restricting this
interaction one can obtain non-trivial constraint languages with more
acceptable complexity. For example, the extension of \RCCE{} with
connectedness constraints is still in \NP{} for both $\RC(\R^2)$ and
$\RCP(\R^2)$~\cite{ijcai:kphz10}.

\smallskip
\noindent
{\bf Acknowledgments.}\ \ This work was partially supported by the U.K. EPSRC grants EP/E034942/1 and EP/E035248/1.


\bibliographystyle{named}

\cleardoublepage

\appendix

\section{Regions with infinitely many components}
\label{sec:sensitivityA}
First we give detailed proofs of Lemma~\ref{lma:ourNewman} and
Theorem~\ref{theo:inftyCc}.
\begin{theorem}[\cite{ijcai:Newman64}]\label{thm:NewmanBnd} 
If $X$ is a connected subset of $\R^n$, then every connected component
of $\R^n\setminus X$ has a connected boundary.
\end{theorem}

\begin{swetheorem}{Lemma~\ref{lma:ourNewman}}
If $X \in \RC(\R^n)$ is connected, then every component of $-X$ has a
connected boundary.
\end{swetheorem}
\begin{proof}
Let $Y$ be a connected component of $-X$.  Suppose that the boundary
$\beta$ of $Y$ is not connected, and let $\beta_1$ and $\beta_2$ be
two sets separating $\beta$: $\beta_1$ and $\beta_2$ are disjoint,
non-empty, closed subsets of $\beta$ whose union is $\beta$. We will
show that $Y$ is not connected.  We have $Y=\tc{(\bigcup_{i\in
    I}Z_i)}$, for some index set $I$, where the $Z_i$ are distinct
connected components of $\R^n\setminus X$.  By
Theorem~\ref{thm:NewmanBnd},`the boundaries $\alpha_i$ of $Z_i$ are
connected subsets of $\beta$, for each $i\in I$. Hence, either
$\alpha_i\subseteq\beta_1$ or $\alpha_i\subseteq\beta_2$, for
otherwise $\alpha_i\cap\beta_1$ and $\alpha_i\cap\beta_2$ would
separate $\alpha_i$. Let $I_j=\{i\in I\mid \alpha_i\subseteq
\beta_j\}$ and $Y_j=\tc{(\bigcup_{i\in I_j}Z_i)}$, for $j=1,2$.
Clearly, $Y_1$ and $Y_2$ are closed, and $Y=Y_1\cup Y_2$. Hence, it
suffices to show that $Y_1$ and $Y_2$ are disjoint. We know that, for
$j=1,2$,
$$
Y_j=\tc{(\bigcup_{i\in I_j}\alpha_i)}\cup\bigcup_{i\in I_j}Z_i.
$$
Clearly, $\bigcup_{i\in I_1}Z_i$ and $\bigcup_{i\in I_2}Z_i$ are disjoint. We also know that 
$\tc{(\bigcup_{i\in I_1}\alpha_i)}$ and $\tc{(\bigcup_{i\in I_2}\alpha_i)}$ are disjoint, as 
subsets of $\beta_1$ and $\beta_2$, respectively. Finally, $\tc{(\bigcup_{i\in I_j}\alpha_i)}$ 
and $\bigcup_{i\in I_k}Z_i$ are disjoint, for $j,k=1,2$, as subsets of the boundary and the 
interior of $Y$, respectively. So, $Y$ is not connected, which is a contradiction.
\end{proof}

\begin{swetheorem}{Theorem~\ref{theo:inftyCc}}
	If $\mathfrak I$ is an interpretation over $\RC(\R^n)$ such that
	$\mathfrak I \models \phi_\infty$, then every $d_i^\mathfrak{I}$
	has infinitely many components.
\end{swetheorem}
\begin{proof}
To simplify presentation, we ignore the difference  between variables and the regions 
they stand for, writing, for example, $a_i$ instead of $a_i^\mathfrak{I}$. We also set 
$b_i=d_i\cdot(-a_i)$. We construct a sequence of disjoint components $X_i$ of $d_{\md{i}}$ 
and open sets $V_i$ connecting $X_i$ to $X_{i+1}$ (Fig.~\ref{fig:InfCmpConstr}). By the first 
conjunct of~\eqref{eq:basic-regions}, let $X_0$ be a component of $d_0$ containing points in 
$a_0$. Suppose $X_i$ has been constructed, for $i \geq 0$. By~\eqref{eq:InfContact} 
and~\eqref{notC}, there exists a point $q \in X_i \cap a_{\md{i+1}}$. Since 
$q\notin b_{\md{i+1}}\cup d_{\md{i+2}}\cup d_{\md{i+3}}$, and because $\R^n$ is locally connected, 
there exists a connected neighbourhood $V_i$ of $q$ such that 
$V_i\cap (b_{\md{i+1}}\cup d_{\md{i+2}}\cup d_{\md{i+3}})=\emptyset$, and so, 
by~\eqref{eq:InfPart1}, $V_i\subseteq d_{\md{i}}+a_{\md{i+1}}$. Further, since $q\in a_{\md{i+1}}$,
$V_i\cap \ti{a_{\md{i+1}}}\neq \emptyset$. Take $X_{i+1}'$ to be a component of $a_{\md{i+1}}$ that 
intersects $V_i$ and $X_{i+1}$ the component of $d_{\md{i+1}}$ containing $X_{i+1}'$.
	
To see that the $X_i$ are distinct, let $S_{i+1}$ and $R_{i+1}$ be the
components of $-X_{i+1}$ containing $X_i$ and $X_{i+2}$,
respectively. It suffices to show $S_{i+1} \subseteq\ti{S}_{i+2}$.
Note that the connected set $V_i$ must intersect $\delta S_{i+1}$.
Evidently, $\delta S_{i+1} \subseteq X_{i+1} \subseteq d_{\md{i+1}}$.
Also, $\delta S_{i+1} \subseteq -X_{i+1}$; hence,
by~\eqref{eq:InfPart1} and~\eqref{eq:InfNTriv1}, $\delta S_{i+1}
\subseteq d_{i} \cup d_{\md{i+2}}$.  By Lemma~\ref{lma:ourNewman},
$\delta S_{i+1}$ is connected, and therefore, by~\eqref{eq:InfNTriv1},
is entirely contained either in $d_{\md{i}}$ or in
$d_{\md{i+2}}$. Since $V_i \cap \delta S_{i+1} \neq \emptyset$ and
$V_i \cap d_{\md{i+2}} = \emptyset$, we have $\delta S_{i+1} \not
\subseteq d_{\md{i+2}}$, so $\delta S_{i+1} \subseteq d_i$. Similarly,
$\delta R_{i+1}\subseteq d_{i+2}$.  By~\eqref{eq:InfNTriv1}, then,
$\delta S_{i+1} \cap \delta R_{i+1} = \emptyset$, and since $S_{i+1}$
and $R_{i+1}$ are components of the same set, they are
disjoint. Hence, $S_{i+1}\subseteq \ti{(-R_{i+1})}$, and since
$X_{i+2}\subseteq R_{i+1}$, also $S_{i+1}\subseteq
\ti{(-X_{i+2})}$. So, $S_{i+1}$ lies in the interior of
a component of $-X_{i+2}$, and since $\delta S_{i+1}\subseteq
X_{i+1}\subseteq S_{i+2}$, that component must be $S_{i+2}$.
\end{proof}

Now we extend the result to the language $\cBCci$. All occurrences of
$c$ in $\phi_\infty$ have positive polarity.  Let $\ti{\phi}_\infty$
be the result of replacing them with the predicate $\ic$. In the
configuration of Fig.~\ref{fig:InfCmpSat}, all connected regions
mentioned in $\phi_\infty$ are in fact interior-connected; hence
$\ti{\phi}_\infty$ is satisfiable over $\RC(\R^n)$. Since
interior-connectedness implies connectedness, $\ti{\phi}_\infty$
entails $\phi_\infty$ in a common extension of $\cBCci$ and
$\cBCc$. Hence:
\begin{swetheorem}{Corollary~\ref{cor:inftyCci}}
There is a $\cBCci$-formula satisfiable over $\RC(\R^n)$, $n \geq 2$,
but not by regions with finitely many components.
\end{swetheorem}
\begin{figure}[h]
\begin{center}
\begin{tikzpicture}
\scriptsize{
	\coordinate (Px) at (-3.8,0);
	\coordinate (Qx) at (-3.3,0);
	\coordinate (Py) at (0,1.3);
	\coordinate (Pty) at (0.23,-.9);
	\coordinate (Pby) at (0,1);
	\foreach \i/\a/\b in {2/20/0,1/10/0,0/0/30,3/0/0,2/20/0,1/10/0,0/0/30}
	{
		\draw[fill=black!\a] ($(Px)-(Py)$) rectangle ($(Py)-(Px)$);		
		\draw[fill=black!\b] ($(Px)-(Pby)$) rectangle ($(Qx)+(Pby)$);
		\node at ($(Px)+(Pty)$) {$a_\i$};
		\node at ($(-0.2,0)-(Px)$) {$b_\i$};		
		
		\coordinate (Px) at ($(Px)+(.5,0)$);
		\coordinate (Py) at ($(Py)-(0,0.15)$);
		\coordinate (Pty) at ($(Pty)+(0,0.15)$);
		\coordinate (Pby) at ($(Pby)-(0,0.15)$);
		\coordinate (Qx) at ($(Qx)+(.5,0)$);
	}
	\draw (4.3,0.5)--(.5,0.25)--(.5,0.2)--(4.3,0.2);
	\node at (4.1,.35) {$t$};
	\node at (4.1,0) {$\ldots$};
	\draw (-4.3,0.25)--(-.75,0)--(-.75,-0.05)--(-4.3,-0.05);
	\node at (-4.1,.125) {$s$};
}
\end{tikzpicture}	
\end{center}
	\caption{Satisfying $\phi_{\lnot C}^c(a_0, b_1,s,t)$ and $\phi_{\lnot C}^c(a_0, b_2,s,t)$.}
	\label{fig:InfCmpElAiBi}	
\end{figure}
To extend Theorem~\ref{theo:inftyCc} to the language $\cBc$, 
notice that all occurrences of $C$ in $\phi_\infty$ are negative.
We shall eliminate these using only the predicate $c$. We use
the fact that, if the sum of two connected regions is not 
connected, then they must be disjoint. Consider the formula 
\begin{align*}
	\phi_{\lnot C}^c(r,s,r',s'):=c(r+r')\land c(s+s') \hspace{2cm}
		\\\hfill\land \lnot c((r+r')+(s+s')).
\end{align*}
Note that $\phi_{\lnot C}^c(r,s,r',s')$ implies $\lnot C(r,s)$.  We
replace $\lnot C(a_i,t)$ with $\phi_{\lnot
  C}^c(a_i,t,a_0+a_1+a_2+a_3,t)$, which is clearly satisfiable by the
regions on Fig.~\ref{fig:InfCmpSat}. Further, we replace $\lnot C(a_i,
b_{\md{i+1}})$ with $\phi_{\lnot C}^c(a_i,b_{\md{i+1}},s,t)$. As shown
on Fig.~\ref{fig:InfCmpElAiBi}, there exists a region $s$ satisfying
this formula.  Instead of dealing with $\lnot C(d_i,d_{i+2})$, we
consider the equivalent:
\begin{align*}
	\lnot C(a_i,b_{\md{i+2}})\land\lnot C(b_i,a_{\md{i+2}})\land \hspace{3cm}\\
	\lnot C(a_i,a_{\md{i+2}})\land\lnot C(b_i,b_{\md{i+2}}).
\end{align*}
We replace $\lnot C(a_i,b_{\md{i+2}})$ by $\phi^c_{\lnot C}(a_i,b_{\md{i+2}},s,t)$,
which is satisfiable by the regions depicted on Fig.~\ref{fig:InfCmpElAiBi}.
We ignore $\lnot C(b_i,a_{\md{i+2}})$, because it is logically equivalent to 
$\lnot C(a_{i},b_{\md{i+2}})$, for different values of $i$. We replace 
$\lnot C(a_i,a_{\md{i+2}})$ by $\phi_{\lnot C}^c(a_i,a_{\md{i+2}},a_i',a_{\md{i+2}}')$, 
which is satisfiable by the regions depicted on Fig.~\ref{fig:InfCmpElAiAi2}. The fourth 
conjunct is then treated symmetrically.
\begin{figure}[h]
\begin{center}
\begin{tikzpicture}
\scriptsize{
	\node at (-5.2,0) {$\ldots$};
	\clip (-5,-2) rectangle (2.5,2);
	\coordinate (Px) at (-4.8,0);
	\coordinate (Qx) at (-4.3,0);
	\coordinate (Py) at (0,1.7);
	\coordinate (Pty) at (0.23,0);
	\coordinate (Pby) at (0,1.4);
	\coordinate (Q0) at (-4.8,-1.4);
	\coordinate (Q2) at (-4.8,1.4);
	
	\foreach \i/\a/\b in {0/0/30,3/0/0,2/0/10,1/0/0,0/0/30,3/0/0,2/0/10,1/0/0,0/0/30}
	{
		\draw[fill=black!\a,draw=gray] ($(Px)-(Py)$) rectangle ($(Py)-(Px)$);		
		\draw[fill=black!\b,draw=gray] ($(Px)-(Pby)$) rectangle ($(Qx)+(Pby)$);
		\node at ($(Px)+(Pty)$) {$a_\i$};
		\node at ($(-0.2,0)-(Px)$) {$b_\i$};		
		\ifnum \i=0
			\draw[fill=black!30]  ($(Q0)+(0,0)$) to[out=-135,in=-45]  
			($(Q0)-(1.5,.6)$)--++(-.2,0) to[out=-60,in=-135] ($(Q0)+(.2,0)$)
			node[very near end, below]{$a_0'$}--cycle;
		\fi
		\ifnum \i=2
			\draw[fill=black!10]  ($(Q2)+(0,0)$) to[out=135,in=45]  
			($(Q2)-(1.5,-.6)$)--++(-.2,0) to[out=60,in=135] ($(Q2)+(.2,0)$)
			node[very near end, above]{$a_2'$}--cycle;
		\fi
		\coordinate (Px) at ($(Px)+(.5,0)$);
		\coordinate (Py) at ($(Py)-(0,0.15)$);
		\coordinate (Q0) at ($(Q0)+(.5,.15)$);
		\coordinate (Q2) at ($(Q2)+(.5,-.15)$);
		\coordinate (Pby) at ($(Pby)-(0,0.15)$);
		\coordinate (Qx) at ($(Qx)+(.5,0)$);
	}
	\draw (4.3,0.5)--(.5,0.25)--(.5,0.2)--(4.3,0.2);
	\node at (4.1,.35) {$t$};
	\node at (4.1,0) {$\ldots$};	
}
\end{tikzpicture}		
\end{center}
	\caption{Satisfying $\phi_{\lnot C}^c(a_0, a_2,a_0', a_2')$.}
	\label{fig:InfCmpElAiAi2}	
\end{figure}
Transforming $\phi_\infty$ in the way just described, we obtain a
$\cBc$-formula $\phi_\infty^c$, which implies $\phi_\infty$ (in the
language $\cBCc$) and which is satisfiable by the arrangement of
$\RC(\R^n)$. Hence, we obtain the following:
\begin{swetheorem}{Corollary~\ref{cor:inftyBc}} 
There is a $\cBc$-formula satisfiable over $\RC(\R^n)$, $n \geq 2$,
but not by regions with finitely many components.
\end{swetheorem}

The only remaining task in this section is to prove Theorem~\ref{theo:inftyBci}.
The construction is similar to the one developed in Sec.~\ref{sec:undecidability},
and as such uses similar techniques. We employ the following notation. 
If $\alpha$ is a Jordan arc, and $p$, $q$ are points on $\alpha$ such that 
$q$ occurs after $p$, we denote by $\alpha[p,q]$ the segment of $\alpha$ from 
$p$ to $q$.
Consider the formula $\stacki(a_1,\ldots, a_n)$ given by:
\begin{align*}
	\bigwedge_{1\leq i<n} \left(\ic(a_i+\cdots+a_n)\land a_i\cdot a_{i+1}=0\right)
	\land \bigwedge_{j-i>1} \lnot C(a_i,a_j)
\end{align*}
This formula allows us to construct sequences of arcs in the following
sense:
\begin{lemma}\label{lma:StackLemmai} 
Suppose that the condition $\stacki(a_1,\ldots,a_n)$ obtains,
$n>1$. Then every point $p_1\in \ti a_1$ can be connected to every
point $p_n\in \ti a_n$ by a Jordan arc
$\alpha=\alpha_1\cdots\alpha_{n-1}$ such that for all $i$ \textup{(}$1\leq i<
n$\textup{)}, each segment $\alpha_i\subseteq \ti{(a_i+a_{i+1})}$ is a
non-degenerate Jordan arc starting at some point $p_i\in\ti a_i$.
\end{lemma}
\begin{proof}
	By $\ic(a_1+\cdots+a_n)$, let $\alpha_1'\subseteq
        \ti{(a_1+\cdots+a_n)}$ be a Jordan arc connecting $p_1$ to
        $p_n$ (Fig.~\ref{fig:stacki}). By the non-contact constraints,
        $\alpha_1'$ has to contain points in $\ti a_2$. Let $p_2'$ be
        one such point.  For $2\leq i<n$ we suppose $\alpha_1, \ldots,
        \alpha_{i-2}$, $\alpha'_{i-1}$ and $p'_i$ to have been
        defined, and proceed as follows. By $\ic(a_i+\cdots+a_n)$, let
        $\alpha_i''\subseteq \ti{(a_i+\cdots+a_n)}$ be a Jordan arc
        connecting $p_i'$ to $p_n$. By the non-contact constraints,
        $\alpha_i''$ can intersect
        $\alpha_1\cdots\alpha_{i-2}\alpha_{i-1}'$ only in its final
        segment $\alpha_{i-1}'$. Let $p_{i-1}$ be the first point of
        $\alpha_{i-1}'$ lying on $\alpha_i'$; let $\alpha_{i-1}$ be
        the initial segment of $\alpha_{i-1}'$ ending at $p_{i-1}$;
        and let $\alpha_i'$ be the final segment of $\alpha_i''$
        starting at $p_{i-1}$. It remains only to define
        $\alpha_{n-1}$, and to this end, we simply set
        $\alpha_{n-1}:=\alpha_{n-1}'$. To see that $p_i$, $2\leq i<n$,
        are as required, note that $p_i\in
        \alpha_i\cap\alpha_{i-1}$. By the disjoint constraints $p_i$
        must be in $a_i$. If $p_i$ was in $\delta(a_i)$, it would also
        have to be in $\delta(a_{i-1})$ and $\delta(a_{i+1})$, which
        is forbidden by the disjoint constraints. Hence $p_i\in\ti
        a_i$, $1\leq i\leq n$. Given $a_i\cdot a_{i+1}=0$, $1\leq
        i<n$, this also guarantees that the arcs $\alpha_i$ are
        non-degenerate.
\end{proof}
\begin{figure}[h]\begin{center}
	\begin{tikzpicture}
		\scriptsize
		{
			\draw[very thick] (1,2) circle(1pt) node[below]{$p_1$}--
					(2,2) circle (1pt) node[above left]{$p_2$} node[below,midway]{$\alpha_1$};
			\draw[dotted,-latex] (3,2)--(4,2) node[above,midway]{$\alpha_1'$};
			\draw[ultra thin,-latex,shorten >=1pt] (3,2) node[below]{$p_2'$} 
					to[out=90,in=90] (2,2) node[above,midway]{$\alpha_2''$};
			\fill[ultra thin] (3,2) circle(1pt);
			\draw[ultra thin] (3,2) --(2,2) node[below,midway]{$\alpha_1'$};
			
			\draw[very thick,-latex] (2,2) to[out=-70,in=180] (3.3,1.2) node[midway,below]{$\alpha_2$};
			\begin{scope}[xshift=-.5cm]
				\node at (4.2,1.2) {{\normalsize $\ldots$}};
				
				\draw[very thick] (4.6,1.2) -- (5.7,1.2) node[midway,below]{$\alpha_{n-2}$} 
						node[above left]{$p_{n-1}$} circle(1pt);
				\draw[ultra thin] (5.7,1.2)--(7,1.2) circle(1pt) node[above]{$p_{n-1'}$};
				\draw[dotted,-latex](7,1.2)--++(1,0) node[below,midway]{$\alpha_{n-2}'$};
				\draw[ultra thin,-latex,shorten >=1pt] (7,1.2) to[out =-90,in=-90]
								(5.7,1.2) node[midway,below]{$\alpha_{n-1}''$};
				\draw[very thick,-latex,shorten >=1pt] (5.7,1.2) to[out=90,in=180] 
					(8,2) node[above]{$p_n$} node[midway,above]{$\alpha_{n-1}$};
				\fill[very thick] (8,2) circle(1pt);
			\end{scope}
		}
	\end{tikzpicture}
\end{center}
\caption{The constraint $\stacki(a_1,\ldots, a_n)$ ensures the existence of a Jordan arc 
$\alpha=\alpha_1\cdots\alpha_{n-1}$ which connects a point $p_1\in \ti a_1$ to a point $p_n\in\ti a_n$.}
\label{fig:stacki}
\end{figure}

Consider now the formula $\frameFlai(a_0,\ldots,a_{n-1})$ given by:
\begin{align*}
	\bigwedge_{0\leq i < n} \left(\ic(a_i)\land 
		\ic(a_i+ a_{\md{i+1}})\land a_i\neq 0 \right)\land \\
\bigwedge_{j-i>1}a_i\cdot a_j=0,	
\end{align*}
where $\md{k}$ denotes $k \mbox{ mod } n$.  This formula allows us to
construct Jordan curves in the plane, in the following sense:
\begin{lemma}\label{lma:FrameLemmaInt}
	Let $n\geq 3$, and suppose $\frameFlai(a_0, \ldots,a_{n-1})$.
        Then there exist Jordan arcs $\alpha_0$, \ldots, $\alpha_{n-1}$
        such that $\alpha_0\ldots\alpha_{n-1}$ is a Jordan curve lying in
        the interior of $a_0+\cdots+a_{n-1}$, and $\alpha_i \subseteq
        \ti{(a_i+ a_{\md{i+1}})}$, for all $i$, $0 \leq i < n$.
\end{lemma}
\begin{proof}
For all $i$ ($0 \leq i < n$), pick $p'_i \in \ti{a}_i$, and pick a Jordan arc
$\alpha'_i \subseteq \ti{(a_i + a_{\md{i+1}})}$ from $p_i$ to $p_{\md{i+1}}$.  
For all $i$ ($2 \leq i \leq n$), let $p_{\md{i}}$ be the first point
of $\alpha_{i-1}$ lying on $\alpha_{\md{i}}$, and let $p''_1$ be the
first point of $\alpha'_0$ lying on $\alpha'_1$. For all $i$ ($2 \leq
i < n$), let $\alpha_i = \alpha'_i[p_i, p_{i+1}]$, let $\alpha''_1 =
\alpha'_1[p''_1, p_2]$, and let $\alpha''_0$ denote the section of
$\alpha'_0$ (in the appropriate direction) from $p_0$ to $p''_1$. Now
let $p_1$ be the first point of $\alpha''_0$ lying on
$\alpha''_1$, let $\alpha_0 = \alpha''_0[p_0,p_1]$, and let $\alpha_1 =
\alpha''_1[p_1,p_2]$. It is routine to verify that the arcs
$\alpha_0$, \ldots, $\alpha_{n-1}$ have the required properties.
\end{proof}

We will now show how to separate certain types of regions in the language $\cBci$. 
We make use of Lemma~\ref{lma:FrameLemmaInt} and the following fact.
\begin{lemma}\label{lma:Newman}{\cite[p.~137]{ijcai:Newman64}}
	Let $F$, $G$ be disjoint, closed subsets of $\R^2$ such that
	$\R^2\setminus F$ and $\R^2 \setminus G$ are connected. Then
	$\R^2\setminus (F \cup G)$ is connected.
\end{lemma}
\begin{figure}[h]
	\begin{center}
		\begin{tikzpicture}
		\scriptsize{
			\draw[black,thick] (0,0)--++(0,1.5)--++(2,0);
			\filldraw[black,thick] (0,0) circle(1pt) --++ (2,0) circle(1pt);
			\filldraw[black,thick] (2,0)--(2,1.5) circle(1pt);
			
			\draw (4,.75) node[right]{$M_1$} circle(1pt)--
				(2,.75) circle(1pt) node[left]{$P_0$} node[midway,above]{$\mu_0$};
			\draw (4,.75) to[out=90,in=45] 
				(1,1.5) circle(1pt) node[below]{$P_1$} node[midway,above]{$\mu_1$};
			\draw (4,.75) to[out=-90,in=-45] 
				(1,0) circle(1pt) node[above]{$P_2$} node[midway,below]{$\mu_2$};;
			
			\node at (.4,-.2){$\tau_2$};
			\node at (.3,1.7){$\tau_1$};
			\node at (2.2,1.1){$\tau_0$};
			
			\node at (-.5,.75) {$n_0$};
			\node at (2.6,1.4) {$n_2$};
			\node at (2.6,0) {$n_1$};
		}	
	\end{tikzpicture}
	\end{center}
	\caption{The Jordan curve $\Gamma=\tau_0\tau_1\tau_2$ separating $m_1$ from $m_2$.}
	\label{fig:SepBci}
\end{figure}
We say that a region $r$ is \emph{quasi-bounded} if either $r$ or $-r$ is bounded.
We can now prove the following.
\begin{lemma}\label{lma:Cci2BciStar}
	There exists a $\cBci$-formula $\eta^*(r,s, \bar{v})$ with the
	following properties: \textup{(}i\textup{)} $\eta^*(r,s, \bar{v})$
	entails $\neg C(r,s)$ over $\RC(\R^2)$; \textup{(}ii\textup{)}
	if the regions $r$ and $s$ can be separated by a Jordan curve, 
	then there exist polygons $\bar{v}$ such that 
	$\eta^*(\tau_1,\tau_2, \bar{v})$; \textup{(}iii\textup{)} if $r$,
	$s$ are disjoint polygons such that $r$ is quasi-bounded and
	$\R^2 \setminus (r+s)$ is connected, then there exist polygons
	$\bar{v}$ such that $\eta^*(\tau_1,\tau_2, \bar{v})$.
\end{lemma}
\begin{proof}
	Let $\bar{v}$ be the tuple of variables $(t_0,\ldots,t_5, m_1, m_2)$, and let
	$\eta^*(r,s,\bar{v})$ be the formula
	\begin{multline*}
			\frameFlai(t_0,\ldots,t_5)\land  r\leq m_1 \wedge s\leq m_2 \wedge \\
		  (t_0 + \ldots + t_5) \cdot (m_1 +m_2) = 0  \wedge \bigwedge_{\substack{i=1,3,5\\j=1,2}} \ic(t_i + m_j).
	\end{multline*}
	
	Property ({\em i}) follows by a simple planarity argument. By
	$\frameFlai(t_0,\ldots,t_5)$ and Lemma~\ref{lma:FrameLemmaInt}, 
	let $\alpha_i$, for $0\leq i\leq 5$, be such that 
	$\Gamma=\alpha_0\cdots\alpha_5$ is a Jordan curve included in 
	$\ti{(t_0 + \cdots + t_5)}$. Further, let 
	$\tau_i=\alpha_{2i}\alpha_{2i+1}$, $0\leq i\leq 2$ 
	(Fig.\ref{fig:SepBci}). Note that all points in $a_{2i+1}$,
	$0\leq i\leq 2$, that are on $\Gamma$ are on $\tau_i$. By 
	$\ic(t_{2i+1} + m_1)$, $0\leq i\leq 2$, let 
	$\mu_i\subseteq(m_1+t_{2i+1})^\circ$ be a Jordan arc with 
	endpoints $M_1\in m_1^\circ$ and $T_i\in\tau_i\cap t_{2i+1}^\circ$. 
	We may assume that these arcs intersect only at their common 
	endpoint $M_1$, so that they divide the residual domain of 
	$\Gamma$ which contains $M_1$ into three sub-domains $n_i$, for 
	$0\leq i\leq 2$. The existence of a point $M_2\in m_2$ in any 
	$n_i$, $0\leq i\leq 2$, will contradict $\ic(t_{2i+1} + m_2)$. 
	So, $m_2$ must be contained entirely in the residual domain of 
	$\Gamma$ not containing $M_1$. Similarly, all points in $m_1$ must 
	lie in the residual domain of $\Gamma$ containing $M_1$. It follows 
	that $m_1$ and $m_2$ are disjoint, and by $r\leq m_1$ and 
	$s\leq m_2$, that $r$ and $s$ are disjoint as well.
	For Property ({\em ii}), let $\Gamma$ be a Jordan curve
	separating $r$ and $s$. Now thicken $\Gamma$ to form an annular
	element of $\RCP(\R^2)$, still disjoint from $r$ and $s$, and divide
	this annulus into the three regions $t_0,\ldots,t_5$ as shown
	(up to similar situation) in Fig.~\ref{fig:Cci2BciStar}.
	Choose $m_1$ and $m_2$ to be the connected components
	of $-(t_0+\cdots+t_5)$ containing $r$ and $s$, respectively.	
	For Property ({\em iii}), it is routine using Lemma~\ref{lma:Newman} 
	to show that there exists a piecewise linear Jordan curve
	$\Gamma$  in $\R^2 \setminus(r+s)$ separating $r$ and $s$.
\end{proof}

\begin{figure}[h]
\begin{center}
\begin{tikzpicture}	
	\scriptsize{
	\draw[fill=white,very thick,draw=black] (0,0) rectangle (4,2.5);
	\draw[fill=white!20,draw=black] (0.3,0.3) rectangle (3.7,2.2);
	
	\draw[fill=white,draw=black] (0.8,0.8) rectangle (1.5,1.3);		\node at (1.15,1.05) {$r$};
	\draw[fill=white,draw=black] (2.2,.8) rectangle (3.2,1.5);		\node at (2.7,1.15) {$r$};
	\draw[fill=white,draw=black] (1.3,1.5) rectangle (1.8,2.);		\node at (1.55,1.75) {$r$};
	
	\draw[fill=white,draw=black] (4.2,.5) rectangle (4.6,1.9);		\node at (4.4,1.3) {$s$};				
	
	\draw (4.7,1)--(7,2.5);
	\draw (4.7,1)--(7,0);
	\node at (6, 1.3) {$s$};
	
	\node at (3,2.35) {$t_0$};
	\node at (3.85,.6) {$t_1$};
	\node at (3,.15) {$t_2$};
	\node at (1,.15) {$t_3$};
	\node at (0.15,.6) {$t_4$};
	\node at (0.15,2) {$t_5$};
	\node at (4.3,.15) {$\Gamma$};
	
	\node at (3.4,.5) {$m_1$};
	\node at (4.35,2.3) {$m_2$};
	
	\draw (2,0)--(2,.3);
	\draw (1,2.2)--(1,2.5);	
	\draw (.3,.3)--(.3,0);
	\draw (0,1.35)--(.3,1.35);
	\draw (3.7,1.35)--(4,1.35);	
	\draw (3.7,.3)--(3.7,0);	
	}
\end{tikzpicture}
\end{center}
\caption{Separating disjoint polygons by an annulus.}	\label{fig:Cci2BciStar}
\end{figure}
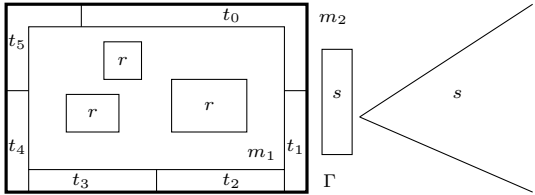
\begin{lemma}\label{lma:Cci2Bci}
	There exists a $\cBci$-formula $\eta(r,s, \bar{v})$ with the following
	properties: \textup{(}i\textup{)} $\eta(r,s, \bar{v})$ entails $\neg
	C(r,s)$ over $\RC(\R^2)$; \textup{(}ii\textup{)} if $r$, $s$ are
	disjoint quasi-bounded polygons, then there exist
	polygons $\bar{v}$ such that $\eta(\tau_1, \tau_2, \bar{v})$.
\end{lemma}
\begin{proof}
	Let $\eta(r,s, \bar{v})$ be the formula
	\begin{equation*}
	r = r_1 + r_2 \wedge s = s_1 + s_2 \wedge \bigwedge_{\substack{1 \leq i \leq 2\\
	  1 \leq j \leq 2}} \eta^*(r_i, s_j, \bar{u}_{i,j}),
	\end{equation*}
	where $\eta^*$ is the formula given in
	Lemma~\ref{lma:Cci2BciStar}. Property ({\em i}) is then immediate. For
	Property ({\em ii}), it is routine to show that there exist  polygons $r_1$, $r_2$
	such that $r = r_1 + r_2$ and $\R^2 \setminus r_i$ is connected for $i
	= 1,2$; let $s_1$, $s_2$ be chosen analogously. Then for all $i$ ($1 \leq i
	\leq 2$) and $j$ ($1 \leq j \leq 2$) we have $r_i \cap s_j =
	\emptyset$ and, by Lemma~\ref{lma:Newman}, $\R^2 \setminus (r_i +
	s_j)$ connected. By Lemma~\ref{lma:Cci2BciStar}, let
	$\bar{u}_{i,j}$ be such that $\eta^*(r_i, s_j, \bar{u}_{i,j})$.
\end{proof}
We are now ready to prove:
\begin{swetheorem}{Theorem~\ref{theo:inftyBci}}
There is a $\cBci$-formula satisfiable over $\RC(\R^2)$, but only by
regions with infinitely many components.
\end{swetheorem}
\begin{proof}
We first write a $\cBCci$-formula, $\phi^*_\infty$ with the required
properties, and then show that all occurrences of $C$ can be
eliminated.  Note that $\phi^*_\infty$ is not the same as the formula
$\ti{\phi}_\infty$ constructed for the proof of
Corollary~\ref{cor:inftyCci}.

Let $s$, $s'$, $a$, $a'$, $b$, $b'$, $a_{i,j}$ and $b_{i,j}$ 
($0 \leq i <2$, $1 \leq j \leq 3$) be variables. The constraints
\begin{align}
	&\frameFlai(s,s',b,b',a,a')\label{eq:BciInf1}\\
	&\stacki(s,b_{i,1},b_{i,2},b_{i,3},b)\label{eq:BciInf2}\\
	&\stacki(b_{\md{i-1},2},a_{i,1},a_{i,2},a_{i,3},a)\label{eq:BciInf3}\\
	&\stacki(a_{\md{i-1},2},b_{i,1},b_{i,2},b_{i,3},b)\label{eq:BciInf4}
\end{align}
are evidently satisfied by the arrangement of Fig.~\ref{fig:InfBci}.
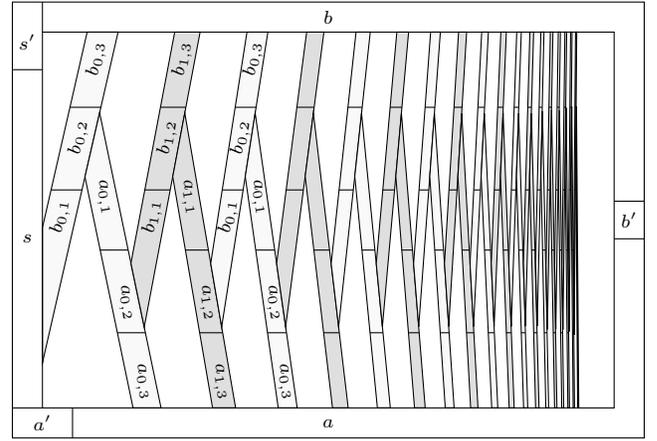
\begin{figure}[h]
\begin{tikzpicture}[scale=1,
		c0/.style={fill=gray!5},
		c1/.style={fill=gray!5},
		c2/.style={fill=gray!25},
		c3/.style={fill=gray!25}
	]
		\scriptsize{
		\newcounter{md}\newcounter{mdf} \newcounter {mdd}
		\coordinate (H) at (0,5);
		\coordinate (P) at (0,0);
		\coordinate (Q) at (H);
		\coordinate (start) at (0,1.5);\coordinate (start') at ($(start)+(1,0)$);
		\coordinate (first) at (0,2.9);\coordinate (first') at ($(first)+(1,0)$);
		\coordinate (second) at (0,4);\coordinate (second') at ($(second)+(1,0)$);
		\coordinate (end) at (H);\coordinate (end') at ($(end)+(1,0)$);
		
		\coordinate (width) at (.8,0);
		
		\foreach \d in {0,...,35}
		{
			\pgfmathparse{.91^\d}			\let\factor=\pgfmathresult 
			\pgfmathparse{2*mod(\d,2)-1}	\let\prt=\pgfmathresult
			\pgfmathsetcounter{mdf}{mod(\d,4)}
			\pgfmathsetcounter{md}{mod(\d,2)}
			\pgfmathsetcounter{mdd}{mod(\d/2,2)}
			
			\coordinate (Pm) at ($(intersection of P--Q and start--start')$);
			\coordinate (Qm) at ($(end-|Pm)+\factor*(width)$);
			\foreach \a/\ind in {-90/0,90/} 
			{
				\coordinate (Pm') at ($(Pm)!{.2cm*\factor}!{\prt*\a}:(Qm)$);
				\coordinate (Qm') at ($(Qm)!{.2cm*\factor}!{-\prt*\a}:(Pm)$);
				\coordinate (P\ind) at ($(intersection of Pm'--Qm' and P--Q)$);
				\coordinate (Q\ind) at ($(intersection of Pm'--Qm' and end--end')$);
			}
			\draw[c\themdf] (P0)--(Q0)--(Q)--(P)--cycle;			
			
			\draw ($(intersection of first--first' and P0--Q0)$) 
				--($(intersection of first--first' and P--Q)$);
			\draw ($(intersection of second--second' and P0--Q0)$)
				--($(intersection of second--second' and P--Q)$);			
			
			\coordinate (start) at ($(H)-(start)$);\coordinate (start') at ($(start)+(1,0)$);
			\coordinate (first) at ($(H)-(first)$);\coordinate (first') at ($(first)+(1,0)$);			
			\coordinate (second) at ($(H)-(second)$);\coordinate (second') at ($(second)+(1,0)$);
			\coordinate (end) at ($(H)-(end)$);\coordinate (end') at ($(end)+(1,0)$);
			
			\ifnum \d<6
				\ifnum \themd=0
					\node[rotate=81] at ($(Pm)!.3!(Qm)$) {$b_{\themdd,1}$};
					\node[rotate=81] at ($(Pm)!.6!(Qm)$) {$b_{\themdd,2}$};
					\node[rotate=81] at ($(Pm)!.9!(Qm)$) {$b_{\themdd,3}$};
				\fi
				\ifnum \themd=1
					\node[rotate=-81] at ($(Pm)!.2!(Qm)$) {$a_{\themdd,1}$};
					\node[rotate=-81] at ($(Pm)!.6!(Qm)$) {$a_{\themdd,2}$};
					\node[rotate=-81] at ($(Pm)!.9!(Qm)$) {$a_{\themdd,3}$};
				\fi
			\fi
		}
		}
		\draw (-0.4,-0.4) rectangle ($(8,.4)+(H)$); 	\draw (0,0) rectangle ($(7.6,0)+(H)$);
		\draw (0,0)--++(-.4,0);
		\draw (H)--++(0,.4);
		\draw ($(H)!.1!(0,0)$)--++(-.4,0);
		\draw ($(.4,0)$)--++(0,-.4);
		\draw ($(H)!.45!(0,0)+(7.6,0)$) --++(.4,0);
		\draw ($(H)!.55!(0,0)+(7.6,0)$) --++(.4,0);
		\node at ($(H)!.55!(0,0)+(-.2,0)$){$s$};
		\node at ($(H)!.025!(0,0)+(-.2,0)$){$s'$};
		\node at (0,-.2){$a'$};
		\node at ($(H)+(3.8,.2)$){$b$};
		\node at ($(H)!.5!(0,0)+(7.8,0)$){$b'$};
		\node at (3.8,-.2){$a$};
		
	\end{tikzpicture}
\caption{A tuple of regions satisfying
  {\eqref{eq:BciInf1}--\eqref{eq:BciInf4}}: the pattern of components of the
  $a_{i,j}$ and $b_{i,j}$ repeats forever.}
\label{fig:InfBci}
\end{figure}

Let $\phi^*_\infty$ be the conjunction of~\eqref{eq:BciInf1}--\eqref{eq:BciInf4}
as well as all conjuncts
\begin{align}
	r\cdot r'=0 & \label{eq:BciInf5},
\end{align}
where $r$ and $r'$ are any two distinct regions depicted on Fig.~\ref{fig:InfBci}. 
Note that the regions $a_{i,j}$ and $b_{i,j}$ have infinitely 
many  connected components. We will now show that this is true for every satisfying 
tuple of $\phi^*_\infty$. 

By \eqref{eq:BciInf1}, we can use Lemma~\ref{lma:FrameLemmaInt} to construct 
a Jordan curve $\Gamma=\sigma\sigma'\beta\beta'\alpha\alpha'$ whose segments 
are Jordan arcs lying in the respective sets $(s+s')^\circ$, $(s'+b)^\circ$, 
$(b+b')^\circ$, $(b'+a)^\circ$, $(a+a')^\circ$, $(a'+s)^\circ$.
Further, let $\sigma_0=\sigma\sigma'$, $\beta_0=\beta\beta'$ and 
$\alpha_0=\alpha\alpha'$ (Fig.~\ref{subfig:BciFrame}). Note that all points in
$s$, $a$ and $b$ that are on $\Gamma$ are on $\sigma_0$, $\alpha_0$ and $\beta_0$, 
respectively. Let $o_0'\in\sigma_0 \cap\ti s$, and let $q^*\in \beta_0\cap\ti b$. 
By \eqref{eq:BciInf2} and Lemma~\ref{lma:StackLemmai} 
we can connect $o_0'$ to $q^*$ by a Jordan arc $\beta_{0,1}'\beta_{0,2}\beta_{0,3}'$ 
whose segments lie in the respective sets $\ti{(s+b_{0,1})}$, 
$\ti{(b_{0,1}+b_{0,2}+b_{0,3})}$ and $\ti{(b+b_{0,3})}$ (Fig.~\ref{subfig:BciBeta0}). 
Let $o_0$ be the last point on $\beta_{0,1}'$ that is on $\sigma_0$ and let $\beta_{0,1}$ 
be the final segment of $\beta_{0,1}'$ starting at $o_0$. Similarly, let $q_0$ 
be the first point on $\beta_{0,3}'$ that is on $\beta_0$ and let $\beta_{0,3}$
be the initial segment of $\beta_{0,3}'$ ending at $q_0$. Hence, the arc $\beta_{0,1}\beta_{0,2}
\beta_{0,3}$ divides one of the regions bounded by $\Gamma$ into two 
sub-regions. We denote the sub-region whose boundary is disjoint from $\alpha_0$ by $U_0$, and 
the other sub-region we denote by $U_0'$. 
Let $\beta_1:=\beta_{0,3}\beta_0[q_0,r]\subseteq \ti{(b+b_{0,3}+b_{1,3})}$.
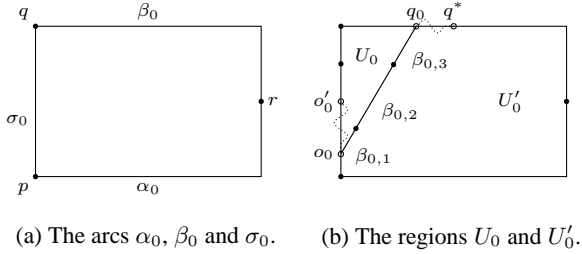
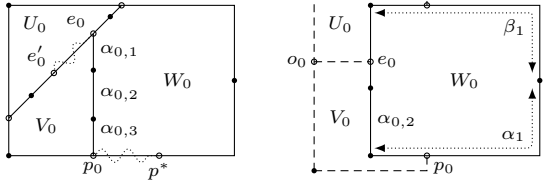
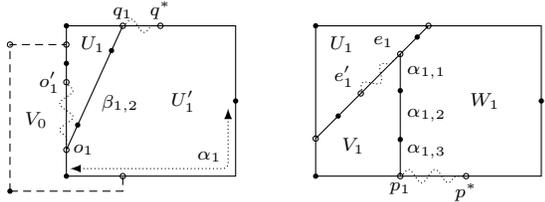
\begin{figure}[h]\begin{center}
\scriptsize{
	\subfloat[The arcs $\alpha_0$, $\beta_0$ and $\sigma_0$.]{
	\begin{tikzpicture}
		\clip (-.5,-.5) rectangle (3.5,2.5);
		\draw (0,0)rectangle(3,2);
		\fill (0,2) circle(1pt) node[above left]{$q$};
		\fill (0,0) circle(1pt) node[below left]{$p$};
		\fill (3,1) circle(1pt) node[right]{$r$};		
		\draw (0,.75)  node[left]{$\sigma_0$};		
		\draw (1.5,2)  node[above]{$\beta_0$};
		\draw (1.5,0)  node[below]{$\alpha_0$};
	\end{tikzpicture}\label{subfig:BciFrame}}
	\subfloat[The regions $U_0$ and $U_0'$.]{
	\begin{tikzpicture}
		\clip (-.5,-.5) rectangle (3.5,2.5);
		\draw (0,0)rectangle(3,2);
		\foreach \x in {(0,2),(0,0),(3,1)} \fill \x circle(1pt);
		
		\fill (0,1.5) circle(1pt);
		\draw (0,1) circle(1pt) coordinate (o0') node[left]{$o_0'$};
		\draw(0,0.3) circle(1pt) coordinate (o0) node[left]{$o_0$};
		\draw[densely dotted,decorate,decoration=snake] (o0') -- (o0);
		
		\draw (1.5,2) circle(1pt) coordinate(qs) node[above]{$q^*$};
		\draw (1,2) circle(1pt) coordinate (q0) node[above]{$q_0$};
		\draw[densely dotted,decorate,decoration=snake](q0)--(qs);

		\fill (o0)--($(o0)!.2!(q0)$) circle(1pt) node[midway,below right]{$\beta_{0,1}$}--
					($(o0)!.7!(q0)$) circle(1pt) node[midway,below right]{$\beta_{0,2}$}--
					(q0) node[midway,below right]{$\beta_{0,3}$};
		\draw (o0)--(q0);
		\node at (.35,1.6){$U_0$};
		\node at (2.25,1){$U_0'$};
	\end{tikzpicture}\label{subfig:BciBeta0}}\\
	\newcommand{\BciEstablishRegionsViUi}[1]{
	\begin{tikzpicture}
		\clip (-.5,-.5) rectangle (3.5,2.5);
		\draw (0,0)rectangle(3,2);
		\foreach \x in {(0,2),(0,0),(3,1)} \fill \x circle(1pt);
		\draw (0,.5) coordinate (p) circle(1pt)--(1.5,2) circle(1pt) coordinate(q);
		\node at (.35,1.75){$U_#1$};		
		\foreach \x in {0.2,0.9} \fill ($(p)!\x!(q)$) circle(1pt);
		\draw ($(p)!.4!(q)$) circle(1pt) coordinate (e0') node[above left ]{$e_#1'$};
		\draw ($(p)!.75!(q)$) circle(1pt) coordinate (e0) node[above left]{$e_#1$};
		\draw[densely dotted,decorate,decoration=snake] (e0') --(e0);
		\draw (2,0) circle(1pt) coordinate(qs) node[below]{$p^*$};
		\draw ($(0,0)!(e0)!(1,0)$) circle(1pt) coordinate (p0) node[below]{$p_#1$}--(e0);
		\draw[densely dotted, decorate,decoration=snake] (p0)--(qs);
		\foreach \x in {0.3,0.7} \fill ($(e0)!\x!(p0)$) circle(1pt);
		\node[right] at ($(e0)!.15!(p0)$) {$\alpha_{#1,1}$};
		\node[right] at ($(e0)!.5!(p0)$) {$\alpha_{#1,2}$};
		\node[right] at ($(e0)!.8!(p0)$) {$\alpha_{#1,3}$};
		\node at (.5,.4){$V_#1$};
		\node at (2.25,1){$W_#1$};
	\end{tikzpicture}\label{subfig:BciAlpha#1}}
	\subfloat[The regions $V_0$ and $W_0$.]{
	\BciEstablishRegionsViUi{0}}
	\subfloat[The regions redrawn.]{\label{subfig:BciRedrawn}
	\begin{tikzpicture}		
		\clip (-.5,-.5) rectangle (3.5,2.5);
		\draw[densely dashed] (0,-.2) -- (0,2.2)--++(1.5,0)--++(0,-.2) coordinate(q0) node[above right]{$q_0$};
		\draw[densely dashed] (0,-.2) --++(1.5,0)--++(0,.2)  coordinate (p0) node[below right]{$p_0$};
		\foreach \x in {(0,2.2),(0,-.2),(3,1)} \fill \x circle(1pt);
		\foreach \x in {(p0),(q0)} \draw \x circle(1pt);
		\coordinate (p) at (.75,2);
		\draw (.75,0) coordinate(q) rectangle (3,2);
		\foreach \x in {0,.55,1} \fill($(p)!\x!(q)$) circle(1pt);
		\draw[densely dashed] (0,1.25) coordinate(o0) node[left]{$o_0$}--++(.75,0) coordinate(e0) node[right]{$e_0$};
		\foreach \x in  {(o0),(e0)} \draw \x  circle(1pt);
		
		\node at (.35,1.75){$U_0$};
		\node at (.35,.5){$V_0$};
		\node at (1.1,.45){$\alpha_{0,2}$};
		\node at (2,1){$W_0$};
		\draw[densely dotted, latex-latex,rounded corners=1pt] ($(p)+(0.05,-.1)$)--
			(2.9,1.9) node[below left]{$\beta_1$}--(2.9,1.1);
		\draw[densely dotted, latex-latex,rounded corners=2pt] ($(q)+(0.05,.1)$)--
			(2.9,.1) node[above left]{$\alpha_1$}--(2.9,.9);
	\end{tikzpicture}}\\
	\subfloat[The regions $U_1$ and $U_1'$.]{\label{subfig:BciBeta1}
	\begin{tikzpicture}		
		\clip (-.5,-.5) rectangle (3.5,2.5);
		\draw[densely dashed] (0,-.2) -- (0,1.75) coordinate(o0)--++(.75,0) coordinate (e0);
		\draw[densely dashed] (0,-.2) --++(1.5,0)--++(0,.2)  coordinate (p0);
		\foreach \x in {(0,-.2),(3,1)} \fill \x circle(1pt);
		\foreach \x in {(p0),(o0),(e0)} \draw \x circle(1pt);
		\coordinate (p) at (.75,2);
		\draw (.75,0) coordinate(q) rectangle (3,2);
		\foreach \x in {0,.25,1} \fill($(p)!\x!(q)$) circle(1pt);
		
		\draw (.75,1.25) circle(1pt) coordinate (o1') node[left]{$o_1'$};
		\draw(.75,0.35) circle(1pt) coordinate (o1) node[right]{$o_1$};
		\draw[densely dotted,decorate,decoration=snake] (o1') -- (o1);
		
		\draw (2,2) circle(1pt) coordinate(qs) node[above]{$q^*$};
		\draw (1.5,2) circle(1pt) coordinate (q1) node[above]{$q_1$};
		\draw[densely dotted,decorate,decoration=snake](q1)--(qs);
		
		\draw (o1)--(q1);
		\foreach \x in {.2,.8} \fill($(o1)!\x!(q1)$) circle(1pt);
		\node[below right] at ($(o1)!.5!(q1)$) {$\beta_{1,2}$};
		
		\node at (.35,.75){$V_0$};
		\node at (1.1,1.75){$U_1$};
		\node at (2.3,1){$U_1'$};
		
		\draw[densely dotted, latex-latex,rounded corners=2pt] ($(q)+(0.05,.1)$)--
			(2.9,.1) node[above left]{$\alpha_1$}--(2.9,.9);
	\end{tikzpicture}}
	\subfloat[The regions $V_1$ and $W_1$.]{\BciEstablishRegionsViUi{1}
	}}
	\end{center} \caption{Establishing infinite sequences of arcs.}
	\label{fig:BciInfCmp}
	\end{figure}

We will now construct a cross-cut $\alpha_{0,1}\alpha_{0,2} \alpha_{0,3}$ in $U_0'$. 
Let $e_0'\in\beta_{0,2}\cap \ti{b_{0,2}}$ and $p^*\in\alpha_0\cap\ti a$. By \eqref{eq:BciInf3} 
and Lemma~\ref{lma:StackLemmai} we can connect $e_0'$ to $p^*$ by a Jordan arc 
$\alpha_{0,1}'\alpha_{0,2}\alpha_{0,3}'$ whose segments lie in the respective 
sets $\ti{(b_{0,2}+a_{0,1})}$, $\ti{(a_{0,1}+a_{0,2}+a_{0,3})}$ and $\ti{(a+a_{0,3})}$ 
(Fig.~\ref{subfig:BciAlpha0}).  Let $e_0$ be the last point on $\alpha_{0,1}'$ that is on 
$\beta_{0,2}$ and let $\alpha_{0,1}$ be the final segment of $\alpha_{0,1}'$ starting at $e_0$.
Similarly, let $p_0$ be the first point on $\alpha_{0,3}'$ that is on $\alpha_0$ and let $\alpha_{0,3}$
be the initial segment of $\alpha_{0,3}'$ ending at $p_0$. By the non-overlapping constraints, 
$\alpha_{0,1}\alpha_{0,2}\alpha_{0,3}$ does not intersect the boundaries of $U_0$ and $U_0'$
except at its endpoints, and hence it is a cross-cut in one of these regions. Moreover, that region has 
to be $U_0'$ since the boundary of $U_0$ is disjoint from $\alpha_0$. So, $\alpha_{0,1}\alpha_{0,2}
\alpha_{0,3}$ divides $U_0'$ into two sub-regions. We denote the sub-region whose boundary contains 
 $\beta_1$ by $W_0$, and the other sub-region we denote by $V_0$. Let 
 $\alpha_1:=\alpha_{0,3}\alpha_0[p_0,r]$ (Fig~\ref{subfig:BciRedrawn}). Note that 
 $\alpha_1\subseteq \ti{(a+a_{0,3}+a_{1,3})}$.
 
 We can now forget about the region $U_0$, and start constructing a cross-cut 
 $\beta_{1,1}\beta_{1,2}\beta_{1,3}$ in $W_0$. As before, let $\beta_{1,1}'\beta_{1,2}\beta_{1,3}'$
 be a Jordan arc connecting a point $o_1'\in\alpha_{0,2}\cap\ti a_{0,2}$ to a point 
 $q^*\in\beta_1\cap\ti b_i$ such that its segments are contained in the respective sets 
 $\ti{(a_{0,2}+b_{1,1})}$, $\ti{(b_{1,1}+b_{1,2}+b_{1,3})}$ and $\ti{(b+b_{1,3})}$. As before, we choose 
 $\beta_{1,1}\subseteq \beta_{1,1}'$ and $\beta_{1,3}\subseteq \beta_{1,3}'$ so that the Jordan arc
 $\beta_{1,1}\beta_{1,2}\beta_{1,3}$ with its endpoints removed is disjoint from the boundaries of 
 $V_0$ and $W_0$. Hence $\beta_{1,1}\beta_{1,2}\beta_{1,3}$ has to be a cross-cut in $V_0$ or 
 $W_0$,  and since the boundary of $V_0$ is disjoint from $\beta_1$ it has to be a cross-cut in  
 $W_0$  (Fig.~\ref{subfig:BciBeta1}). So, $\beta_{1,1}\beta_{1,2}\beta_{1,3}$ separates $W_0$ 
 into two regions $U_1$ and $U_1'$ so that the boundary of $U_1$ is disjoint from $\alpha_1$.
 Let $\beta_2:=\beta_{1,3}\beta_1[q_1,r]\subseteq \ti{(b+b_{0,3}+b_{1,3})}$. Now, we can ignore 
 the region $V_0$, and reasoning as before we can construct a cross-cut 
 $\alpha_{1,1}\alpha_{1,2}\alpha_{1,3}$ in $U_1'$ dividing it into two sub-regions $V_1$ and $W_1$. 
\begin{figure}[h]
\begin{tikzpicture}[scale=1,
		c0/.style={fill=gray!5},
		c1/.style={fill=gray!5},
		c2/.style={fill=gray!25},
		c3/.style={fill=gray!25}
	]
		\clip (-.5,.5) rectangle (8,4.5);
		\scriptsize{
		\coordinate (H) at (0,5);
		\coordinate (P) at (0,0);
		\coordinate (Q) at (H);
		\coordinate (start) at (0,1.5);\coordinate (start') at ($(start)+(1,0)$);
		\coordinate (first) at (0,2.9);\coordinate (first') at ($(first)+(1,0)$);
		\coordinate (second) at (0,4);\coordinate (second') at ($(second)+(1,0)$);
		\coordinate (end) at (H);\coordinate (end') at ($(end)+(1,0)$);
		
		\coordinate (width) at (.8,0);
		
		\foreach \d in {0,...,35}
		{
			\pgfmathparse{.91^\d}			\let\factor=\pgfmathresult 
			\pgfmathparse{2*mod(\d,2)-1}	\let\prt=\pgfmathresult
			\pgfmathsetcounter{mdf}{mod(\d,4)}
			\pgfmathsetcounter{md}{mod(\d,2)}
			\pgfmathsetcounter{mdd}{mod(\d/2,2)}
			
			\coordinate (Pm) at ($(intersection of P--Q and start--start')$);
			\coordinate (Qm) at ($(end-|Pm)+\factor*(width)$);
			\foreach \a/\ind in {-90/0,90/} 
			{
				\coordinate (Pm') at ($(Pm)!{.2cm*\factor}!{\prt*\a}:(Qm)$);
				\coordinate (Qm') at ($(Qm)!{.2cm*\factor}!{-\prt*\a}:(Pm)$);
				\coordinate (P\ind) at ($(intersection of Pm'--Qm' and P--Q)$);
				\coordinate (Q\ind) at ($(intersection of Pm'--Qm' and end--end')$);
			}
			\draw[c\themdf] (P0)--(Q0)--(Q)--(P)--cycle;			
			
			\draw ($(intersection of first--first' and P0--Q0)$) 
				--($(intersection of first--first' and P--Q)$);
			\draw ($(intersection of second--second' and P0--Q0)$)
				--($(intersection of second--second' and P--Q)$);			
			
			\coordinate (start) at ($(H)-(start)$);\coordinate (start') at ($(start)+(1,0)$);
			\coordinate (first) at ($(H)-(first)$);\coordinate (first') at ($(first)+(1,0)$);			
			\coordinate (second) at ($(H)-(second)$);\coordinate (second') at ($(second)+(1,0)$);
			\coordinate (end) at ($(H)-(end)$);\coordinate (end') at ($(end)+(1,0)$);
			
			\ifnum \d<6
				\ifnum \themdd=0
					\ifnum \themd=0
						\node[rotate=81] at ($(Pm)!.6!(Qm)$) {$b_{\themdd,2}$};
					\fi
					\ifnum \themd=1
						\node[rotate=-81] at ($(Pm)!.6!(Qm)$) {$a_{\themdd,2}$};
					\fi
				\fi
			\fi
		}
		}
		\draw[thick] (4,1.6) ellipse (3.8 and 1.1);
		\draw (-0.4,-0.4) rectangle ($(8,.4)+(H)$); 	\draw (0,0) rectangle ($(7.6,0)+(H)$);
		\draw (0,0)--++(-.4,0);
		\draw (H)--++(0,.4);
		\draw ($(H)!.1!(0,0)$)--++(-.4,0);
		\draw ($(.4,0)$)--++(0,-.4);
		\draw ($(H)!.45!(0,0)+(7.6,0)$) --++(.4,0);
		\draw ($(H)!.55!(0,0)+(7.6,0)$) --++(.4,0);
		\node at ($(H)!.55!(0,0)+(-.2,0)$){$s$};
		\node at ($(H)!.025!(0,0)+(-.2,0)$){$s'$};
		\node at (0,-.2){$a'$};
		\node at ($(H)+(3.8,.2)$){$b$};
		\node at ($(H)!.5!(0,0)+(7.8,0)$){$b'$};
		\node at (3.8,-.2){$a$};
	\end{tikzpicture}
\caption {Separating $a_{0,2}$ from $b_{0,2}$ by a Jordan curve.}
\label{fig:BciInftySep}
\end{figure}
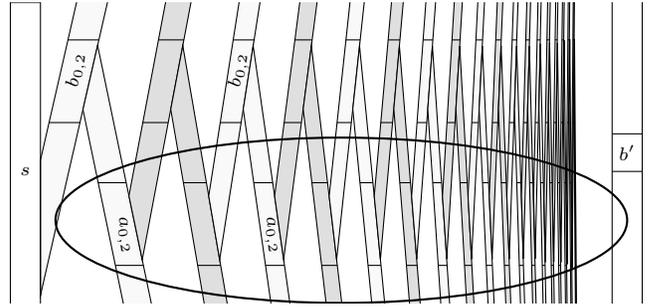
 
 Evidently, this process continues forever. Now, note that by
 construction and \eqref{eq:BciInf5}, $W_{2i}$ contains in its
 interior $\beta_{2i+1,2}$ together with the connected component $c$
 of $b_{1,2}$ which contains $\beta_{2i+1,2}$. On the other hand,
 $W_{2i+2}$ is disjoint from $c$, and since $W_i\subseteq W_j$, $i>j$,
 $b_{1,2}$ has to have infinitely many connected components.

So far we know that the $\cBCci$-formula $\phi^*_{\infty}$ forces
infinitely many components.  Now we replace every conjunct in
$\phi^*_{\infty}$ of the form $\lnot C(r,s)$ by $\eta^*(r,s,\bar v)$,
where $\bar v$ are fresh variables each time. The resulting formula
entails $\phi^*_{\infty}$, so we only have to show that it is still
satisfiable. By Lemma~\ref{lma:Cci2BciStar} (\emph{ii}), it suffices
to separate by Jordan curves every two regions on
Fig.~\ref{fig:InfBci} that are required to be disjoint. It is shown on
Fig.~\ref{fig:BciInftySep} that there exists a curve which separates
the regions $b_{0,2}$ and $a_{0,2}$. All other non-contact constraints
are treated analogously.
\end{proof}

\section{Undecidability of \cBc{} and \cBCc{} in the Euclidean plane}
\label{sec:UndecidabilityB}
In this section, we prove the undecidability of the problems
$\Sat(\cL,\RC(\R^2))$ and $\Sat(\cL,\RCP(\R^2))$, for $\cL$ any of
$\cBc$, $\cBCc$, $\cBci$ or $\cBCci$. We begin with some technical
preliminaries, again employing the notation from the proof of
Theorem~\ref{theo:inftyBci}: if $\alpha$ is a Jordan arc, and $p$, $q$
are points on $\alpha$ such that $q$ occurs after $p$, we denote by
$\alpha[p,q]$ the segment of $\alpha$ from $p$ to $q$. For brevity of
exposition, we allow the case $p= q$, treating $\alpha[p,q]$ as a
(degenerate) Jordan arc.

Our first technical preliminary is to formalize our earlier observations
concerning the formula $\stack(\tseq{a}_1, \ldots, \tseq{a}_n)$,
defined by:
\begin{equation*}
\bigwedge_{1 \leq i \leq n}
     c(\intermediate{a}_i + \inner{a}_{i + 1} + \cdots + \inner{a}_n)
\mbox{}\wedge
\bigwedge_{j - i > 1} \neg C(a_i,a_j).
\end{equation*}
\begin{lemma}
\label{lma:stackLemma}
Let $\tseq{a}_1,\ldots,\tseq{a}_n$ be 3-regions satisfying
$\stack(\tseq{a}_1,\ldots,\tseq{a}_n)$, for $n \geq 3$.  Then, for
every point $p_0\in \intermediate{a}_1$ and every point $p_n \in
\inner{a}_n$, there exist points $p_1, \ldots, p_{n-1}$ and Jordan
arcs $\alpha_1, \ldots, \alpha_n$ such that\textup{:}
\begin{itemize}
\item[\textup{(}i\textup{)}]
$\alpha = \alpha_1 \cdots \alpha_n$ is a Jordan arc from $p_0$ to
$p_n$\textup{; }
\item[\textup{(}ii\textup{)}] for all $i$ \textup{(}$0 \leq i < n$\textup{)}, $p_i \in
\intermediate{a}_{i+1} \cap \alpha_i$\textup{;} and 
\item[\textup{(}iii\textup{)}]
for all $i$ \textup{(}$1 \leq i \leq n$\textup{)}, $\alpha_i \subseteq a_i$.
\end{itemize}
\end{lemma}
\begin{proof}
Since $\intermediate{a}_1 + \inner{a}_2 + \cdots + \inner{a}_n$ is a
connected subset of $\ti{(a_1 + \intermediate{a}_2 + \cdots +
  \intermediate{a}_n)}$, let $\beta_1$ be a Jordan arc connecting
$p_0$ to $p_n$ in $\ti{(a_1 + \intermediate{a}_2 + \cdots +
  \intermediate{a}_n)}$. Since $a_1$ is disjoint from all the $a_i$
except $a_2$, let $p_1$ be the first point of $\beta_1$ lying in
$\intermediate{a}_2$, so $\beta_1[p_0,p_1]\subseteq \ti{a}_1\cup \{
p_1\}$, i.e., the arc $\beta_1[p_0,p_1]$ is either included in
$\ti{a}_1$, or is an end-cut of $\ti{a}_1$. (We do not rule out $p_0 =
p_1$.) Similarly, let $\beta'_2$ be a Jordan arc connecting $p_1$ to
$p_n$ in $\ti{(a_2 + \intermediate{a}_3 + \cdots +
  \intermediate{a}_n)}$, and let $q_1$ be the last point of $\beta'_2$
lying on $\beta_1[p_0,p_1]$. If $q_1 = p_1$, then set $v_1 = p_1$,
$\alpha_1 = \beta_1[p_0,p_1]$, and $\beta_2 = \beta'_2$.  so that the
endpoints of $\beta_2$ are $v_1$ and $p_n$.  Otherwise, we have $q_1
\in \ti{a}_1$.  We can now construct an arc $\gamma_1 \subseteq
\ti{a}_1 \cup \{ p_1 \}$ from $p_1$ to a point $v_1$ on
$\beta'_2[q_1,p_n]$, such that $\gamma_1$ intersects
$\beta_1[p_0,p_1]$ and $\beta'_2[q_1,p_n]$ only at its endpoints,
$p_1$ and $v_1$ (upper diagram in Fig.~\ref{fig:stackLemma}). Let
$\alpha_1 = \beta_1[p_0,p_1]\gamma_1$, and let $\beta_2 =
\beta'_2[v_1,p_n]$.

Since $\beta_2$ contains a point $p_2 \in \intermediate{a}_3$, we may
iterate this procedure, obtaining $\alpha_2, \alpha_3, \ldots
\alpha_{n-1}, \beta_{n}$. We remark that $\alpha_i$ and $\alpha_{i+1}$
have a single point of contact by construction, while $\alpha_i$ and
$\alpha_j$ ($i < j-1$) are disjoint by the constraint $\neg C(a_i,
a_j)$.  Finally, we let $\alpha_n = \beta_n$ (lower diagram in
Fig.~\ref{fig:stackLemma}).
\end{proof}
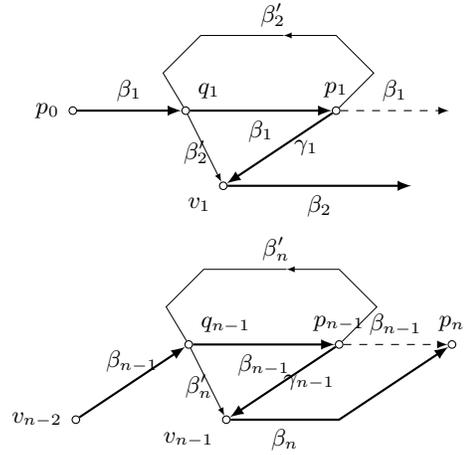
\begin{figure}
\begin{center}
\begin{tikzpicture}[>=latex, point/.style={circle,draw=black,minimum size=1mm,inner sep=0pt}]
						
    \node[point,label=left:{\small $p_0$}] (P0) at (-0.5,0) {};
    \node[point,label=above right:{\small $q_1$}] (Q1) at (1,0) {};
    \node[point,label=below left:{\small $v_1$}] (Q1') at (1.5,-1) {};
    \node[point,label=above:{\small $p_1$}] (P1) at (3,0) {};
    \draw[->,thick] (P0) to node[above] {\footnotesize $\beta_1$} (Q1);
    \draw[->,thick] (Q1) to node[below] {\footnotesize $\beta_1$} (P1);
    \draw[->,dashed] (P1) to node[above] {\footnotesize $\beta_1$} (4.5,0);

	\draw[->,ultra thin] (P1) -- ++(.5,0.5) --++(-.5,.5)  -- (2.3,1);
	\draw[ultra thin] (2.3,1) to node[above, very near start]{\footnotesize $\beta_2'$} (1.2,1) -- ++(-.5,-0.5) -- (Q1);
    \draw[->,ultra thin] (Q1) to node[below,near start]{\footnotesize $\beta_2'$} (Q1');
	\draw[->,thick] (P1) to node[below,near start]{\footnotesize $\gamma_1$} (Q1');
	
	\draw[->,thick] (Q1') to node[below] {\footnotesize $\beta_2$} (4,-1);
\end{tikzpicture}
\begin{tikzpicture}[>=latex, point/.style={circle,draw=black,minimum size=1mm,inner sep=0pt}]
    \node[point,label=left:{\small $v_{n-2}$}] (Qn2) at (6.5,-1) {};
    \node[point,label=above right:{\small $q_{n-1}$}] (Qn1) at (8,0) {};
    \node[point,label=below left:{\small $v_{n-1}$}] (Qn1') at (8.5,-1) {};
    \node[point,label=above:{\small $p_{n-1}$}] (Pn1) at (10,0) {};
    \node[point,label=above:{\small $p_n$}] (Pn) at (11.5,0) {};
    \draw[->,thick] (Qn2) to node[above] {\footnotesize $\beta_{n-1}$} (Qn1);
    \draw[->,thick] (Qn1) to node[below] {\footnotesize $\beta_{n-1}$} (Pn1);
    \draw[->,dashed] (Pn1) to node[above] {\footnotesize $\beta_{n-1}$} (Pn);

	\draw[->,ultra thin] (Pn1) -- ++(.5,0.5) --++(-.5,.5)  -- (9.3,1);
	\draw[ultra thin] (9.3,1) to node[above, very near start]{\footnotesize $\beta_n'$} (8.2,1) -- ++(-.5,-0.5) -- (Qn1);
    \draw[->,ultra thin] (Qn1) to node[below,near start]{\footnotesize $\beta_n'$} (Qn1');
	\draw[->,thick] (Pn1) to node[below,near start]{\footnotesize $\gamma_{n-1}$} (Qn1');
	
	\draw[->,thick] (Qn1') to node[below] {\footnotesize $\beta_n$} (10,-1) to (Pn);
\end{tikzpicture}
\end{center}
\caption{Proof of Lemma~\ref{lma:stackLemma}.}\label{fig:stackLemma}
\end{figure}

In fact, we can add a `switch' $w$ to the formula $\stack(\tseq{a}_1,
\ldots, \tseq{a}_n)$, in the following sense.  If $w$ is a region
variable, consider the formula $\stack_w(\tseq{a}_1, \ldots, \tseq{a}_n)$
\begin{align*}
\neg C(w\cdot \intermediate{a}_1, (-w) \cdot \intermediate{a}_1) \ \wedge \ \stack((-w)\cdot \tseq{a}_1, \tseq{a}_2, \ldots, \tseq{a}_n),
\end{align*}
where $w \cdot \tseq{a}$ denotes the 3-region $(w\cdot a,w\cdot
\intermediate{a},w\cdot \inner{a})$. The first conjunct of
$\stack_w(\tseq{a}_1, \ldots, \tseq{a}_n)$ ensures that any component
of $\intermediate{a}_1$ is either included in $w$ or included in
$-w$. The second conjunct then has the same effect as
$\stack(\tseq{a}_1, \ldots, \tseq{a}_n)$ for those components of
$\intermediate{a}_1$ included in $-w$. That is, if $p \in
\intermediate{a}_1 \cdot (-w)$, we can find an arc $\alpha_1 \cdots
\alpha_n$ starting at $p$, with the properties of
Lemma~\ref{lma:stackLemma}.  However, if $p \in \intermediate{a} \cdot
w$, no such arc need exist.  Thus, $w$ functions so as to
`de-activate' the formula $\stack_w(\tseq{a}_1, \ldots, \tseq{a}_n)$
for any component of $\intermediate{a}_1$ included in it.

As a further application of Lemma~\ref{lma:stackLemma}, consider the
formula $\frameFla(\tseq{a}_0, \ldots, \tseq{a}_n)$ given by:
\begin{align}
  \stack(\tseq{a}_0, \ldots, \tseq{a}_{n-1}) \wedge \lnot C(a_n,a_1+\ldots+a_{n-2})\wedge\nonumber\\
     c(\intermediate a_n)\wedge \intermediate{a}_0\cdot \intermediate a_n\neq 0\wedge 
	\inner{a}_{n-1}\cdot \intermediate a_n\neq 0.
\label{eq:frame}
\end{align}
This formula allows us to construct Jordan curves in the plane, in
the following sense:
\begin{lemma}
\label{lma:FrameLemma}
Let $n\geq 3$, and suppose $\frameFla(\tseq{a}_0, \ldots,
\tseq{a}_n)$.  Then there exist Jordan arcs $\gamma_0$, \ldots,
$\gamma_n$ such that $\gamma_0\ldots\gamma_n$ is a Jordan
curve, and $\gamma_i \subseteq a_i$, for all $i$, $0 \leq i \leq n$.
\end{lemma}
\begin{proof}
By $\stack(\tseq a_0,\ldots,\tseq a_{n-1})$, let
$\alpha_0,\ldots,\alpha_{n-1}$ be Jordan arcs in the respective
regions $a_0,\ldots,a_{n-1}$ such that, 
$\alpha=\alpha_0\cdots\alpha_{n-1}$ is a Jordan arc connecting 
a point $p'\in \intermediate{a}_0\cdot \intermediate a_n$ to a 
point $q'\in \inner{a}_{n-1}\cdot \intermediate a_n$ (see 
Fig.~\ref{fig:FrameLemma}). Because $\intermediate a_n$ is a 
connected subset of the interior of $a_n$, let $\alpha_n\subseteq \ti a_n$ 
be an arc connecting $p'$ and $q'$. Note that $\alpha_n$ does not intersect 
$\alpha_i$, for $1\leq i< n-1$. Let $p$ be the last point on $\alpha_0$ that 
is on $\alpha_n$ (possibly $p'$), and $q$ be the first point on $\alpha_{n-1}$ 
that is on $\alpha_n$ (possibly $q'$). Let $\gamma_0$ be the final segment of 
$\alpha_0$ starting at $p$. Let $\gamma_i:=\alpha_i$, for $1\leq i\leq n-2$.
Let $\gamma_{n-1}$ be the initial segment of $\alpha_{n-1}$ ending at $q$.
Finally, take $\gamma_n$ to be the segment of $\alpha_n$ between $p$ and $q$.
Evidently, the arcs $\gamma_i$, $0\leq i\leq n$, are as required.
\end{proof}
\begin{figure}[h]
\begin{center}
\begin{tikzpicture}
	\scriptsize{
		\draw (.5,1) circle(1pt) node[below]{$p'$};			
		\draw (.5,1) -- (1.5,1);
		\draw (1.5,1)--(1.85,1) --++(0,-1) node[left,near start]{$\alpha_0$} circle(1pt)
			--++(-1.45,0) node[midway,below]{$\alpha_1$}node[left]{$\ldots$};
		
		\draw (-.5,1) circle(1pt) node[below]{$q'$};
		\draw (-.5,1) -- (-1.5,1);
		\draw (-1.5,1)--(-1.85,1) --++(0,-1) node[right,near start]{$\alpha_{n-1}$} circle(1pt)
			--++(1.7,0) node[midway,below]{$\alpha_{n-2}$};
		
		\begin{scope}[xshift=4cm]
			\draw (.5,1) circle(1pt) node[below]{$p'$};
			\draw (1.5,1) circle(1pt) node[below]{$p$};
			\draw[decorate,decoration=snake,densely dotted] (.5,1) -- (1.5,1);
			\draw[dashed] (.5,1) -- (1.5,1);
			\draw (1.5,1)--(1.85,1) --++(0,-1) node[left,midway]{$\gamma_0$} circle(1pt)
				--++(-1.45,0) node[midway,below]{$\gamma_1=\alpha_1$}node[left]{$\ldots$};
			
			\draw (-.5,1) circle(1pt) node[below]{$q'$};
			\draw (-1.5,1) circle(1pt) node[below]{$q$};
			\draw[decorate,decoration=snake,densely dotted] (-.5,1) -- (-1.5,1);
			\draw[dashed] (-.5,1) -- (-1.5,1);
			\draw (-1.5,1)--(-1.85,1) --++(0,-1) node[right,midway]{$\gamma_{n-1}$} circle(1pt)
				--++(1.7,0) node[midway,below]{$\gamma_{n-2}=\alpha_{n-2}$};
			
			\draw(1.5,1) to[out=30,in=150] (-1.5,1) node[midway,below]{$\gamma_n$};
		\end{scope}
	}
\end{tikzpicture}
\end{center}
\caption{Establishing a Jordan curve.}
\label{fig:FrameLemma}
\end{figure}
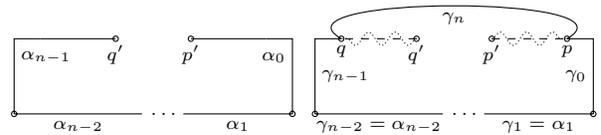

Our final technical preliminary is a simple device for labelling arcs
in diagrams.
\begin{lemma}
\label{lma:labelling}
\label{lma:labels}
Suppose $r$, $t_1$, \ldots, $t_\ell$ are regions such that
\begin{equation}
\label{eq:labelling}
(r \leq t_1 +
\cdots + t_\ell) \wedge \bigwedge_{1 \leq i < j \leq \ell} \neg C(r
\cdot t_i, r \cdot t_j),
\end{equation}
and let $X$ be a connected subset of $r$. Then $X$ is included in
exactly one of the $t_i$, $1 \leq i \leq \ell$.
\end{lemma}
\begin{proof}
If $X \cap t_1$ and $X \cap t_2$ are non-empty, then $X \cap t_1$
and $X \cap (t_2 + \cdots + t_\ell)$ partition $X$ into
non-empty, non-intersecting sets, closed in $X$.
\end{proof}
When~\eqref{eq:labelling} holds, we may think of the regions $t_1,
\ldots, t_\ell$ as `labels' for any connected $X \subseteq r$---and,
in particular, for any Jordan arc $\alpha \subseteq r$. Hence, any
sequence $\alpha_1, \ldots, \alpha_n$ of such arcs encodes a word over
the alphabet $\set{t_1, \ldots, t_\ell}$.

\bigskip

The remainder of this section is given over to a proof of
\begin{swetheorem}{Theorem~\ref{theo:undecidable}}
For $\mathcal{L}\in \{\cBci, \cBc, \cBCci, \cBCc\}$,
$\Sat(\mathcal{L},\RC(\R^2))$ is r.e.-hard, and
$\Sat(\mathcal{L},\RCP(\R^2))$ is r.e.-complete.
\end{swetheorem}

We have already established the upper bounds; we consider here only
the lower bounds, beginning with an outline of our proof strategy.
Let a PCP-instance $\fW = (\set{0,1}, T, \fw_1, \fw_2)$ be given,
where $T$ is a finite alphabet, and $\fw_i\colon T^* \rightarrow
\set{0,1}^*$ a word-morphism ($i = 1,2$).  We call the elements of $T$
    {\em tiles}, and, for each tile $t$, we call $\fw_1(t)$ the {\em
      lower word} of $t$, and $\fw_2(t)$ the {\em upper word} of
    $t$. Thus, $\fW$ asks whether there is a sequence of tiles
    (repeats allowed) such that the concatenation of their upper words
    is the same as the concatenation of their lower words. We shall
    henceforth restrict all (upper and lower) words on tiles to be
    non-empty.  This restriction simplifies the encoding below, and
    does not affect the undecidability of the PCP.

We define a formula $\phi_\fW$ consisting of a large conjunction of
$\cBCc$-literals, which, for ease of understanding, we introduce in
groups. Whenever conjuncts are introduced, it can be readily checked
that---provided $\fW$ is positive---they are satisfiable by elements
of $\RCP(\R^2)$. (Figs.~\ref{fig:concrete1} and~\ref{fig:concrete2}
depict {\em part} of a satisfying assignment; this drawing is
additionally useful as an aid to intuition throughout the course of
the proof.)  The main object of the proof is to show that, conversely,
if $\phi_\fW$ is satisfied by any tuple in $\RC(\R^2)$, then $\fW$
must be positive. Thus, the following are equivalent:
\begin{enumerate}
\item $\fW$ is positive;
\item $\phi_\fW$ is satisfiable over $\RCP(\R^2)$;
\item $\phi_\fW$ is satisfiable over $\RC(\R^2)$.
\end{enumerate}
This establishes the r.e.-hardness of $\Sat(\cL,\RC(\R^2))$ and
$\Sat(\cL,\RCP(\R^2))$ for $\cL = \cBCc$; we then extend the result to the
languages $\cBc$, $\cBCci$ and $\cBci$.

\begin{figure}
\resizebox{9cm}{!}{\input{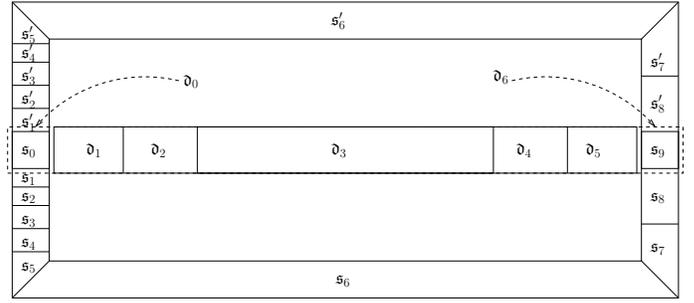}}
\caption{A tuple
of 3-regions satisfying~{\eqref{eq:PCPFrame}}--{\eqref{eq:PCPCord}}. The 3-regions
$\tseq{d}_0$ and $\tseq{d}_6$ are shown in dotted lines.}
\label{fig:concrete1}
\end{figure}
The proof proceeds in five stages.
\bigskip

\noindent
\textbf{Stage 1.} In the first stage, we define an assemblage of arcs
that will serve as a scaffolding for the ensuing construction.
Consider the arrangement of polygonal 3-regions depicted in
Fig.~\ref{fig:concrete1}, assigned to the 3-region variables
$\tseq{s}_0, \dots, \tseq{s}_9$, $\tseq{s}_8', \dots, \tseq{s}_1'$,
$\tseq{d}_0,\dots, \tseq{d}_6$
as indicated. 
It is easy to verify that this arrangement can be made to
satisfy the following formulas:
\begin{align}
\label{eq:PCPFrame}
& \frameFla(\tseq{s}_0, \tseq{s}_1,\dots, \tseq{s}_8,\tseq{s}_9,\tseq{s}_8',\dots, \tseq{s}_1'),\\
\label{eq:PCPCord:endpoints}
& (s_0 \leq \intermediate{t}_0) \land (s_9 \leq \inner{t}_6),\\
\label{eq:PCPCord} & \stack(\tseq{d}_0,\dots,\tseq{d}_6).
\end{align}
And trivially, the arrangement can be made to satisfy any formula
\begin{equation}\label{eq:pcp:C}
\neg C(r,r')
\end{equation}
for which the corresponding 3-regions $\tseq{r}$ and $\tseq{r}'$ are
drawn as not being in contact. (Remember, $r$ is the outer-most shell
of the 3-region $\tseq{r}$, and similarly for $r'$.)  Thus, for
example, \eqref{eq:pcp:C} includes $\neg C(s_0, d_1)$, but not $\neg
C(s_0, d_0)$ of $\neg C(d_0, d_1)$.

Now suppose $\tseq{s}_0, \dots, \tseq{s}_9$, $\tseq{s}_8', \dots,
\tseq{s}_1'$, $\tseq{d}_0,\dots, \tseq{d}_6$ is {\em any} collection
of 3-regions (not necessarily polygonal)
satisfying~\eqref{eq:PCPFrame}--\eqref{eq:pcp:C}.  By
Lemma~\ref{lma:FrameLemma} and~\eqref{eq:PCPFrame}, let $\gamma_0,
\dots, \gamma_9,\gamma_8',\dots,\gamma_1'$ be Jordan arcs included in
the respective regions $s_0, \dots, s_9,s_8',\dots,s_1'$, such that
$\Gamma = \gamma_0 \cdots \gamma_9 \cdot \gamma_8' \cdots \gamma_1'$
is a Jordan curve (note that $\gamma_i'$ and $\gamma_i$ have opposite
directions). We select points $\tilde{o}_{1}$ on $\gamma_0$ and
$\tilde{o}_2$ on $\gamma_9$ (see Fig.~\ref{fig:arcs0}).
By~\eqref{eq:PCPCord:endpoints}, $\tilde{o}_1 \in \intermediate{t}_0$
and $\tilde{o}_2 \in \inner{t}_6$. By Lemma~\ref{lma:stackLemma}
and~\eqref{eq:PCPCord}, let $\tilde{\chi}_1$, $\chi_2$, $\tilde{\chi}_3$
be Jordan arcs in the respective regions
\begin{equation*}
(d_0 + d_1),\qquad (d_2 + d_3 + d_4),\qquad (d_5+ d_6)
\end{equation*}
such that $\tilde{\chi}_1 \chi_2 \tilde{\chi}_3$ is a Jordan arc from
$\tilde{o}_1$ to $\tilde{o}_2$. Let $o_{1}$ be the last point of
$\tilde{\chi}_1$ lying on $\Gamma$, and let $\chi_1$ be the final
segment of $\tilde{\chi}_1$, starting at $o_1$.  Let $o_2$ be the
first point of $\tilde{\chi}_3$ lying on $\Gamma$, and let $\chi_3$ be
the initial segment of $\tilde{\chi}_3$, ending at $o_2$.
\begin{figure}
\resizebox{8.5cm}{!}{\input{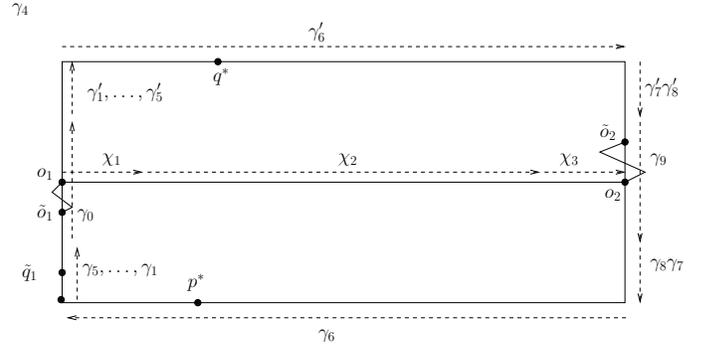}}
\caption{The arcs $\gamma_0, \ldots, \gamma_9$ and $\chi_1, \ldots \chi_3$.}
\label{fig:arcs0}
\end{figure}
By~\eqref{eq:pcp:C}, we see that the arc $\chi = \chi_1\chi_2\chi_3$
intersects $\Gamma$ only in its endpoints, and is thus a chord of
$\Gamma$, as shown in Fig.~\ref{fig:arcs0}.

A word is required concerning the generality of this diagram.  The
reader is to imagine the figure drawn on a {\em spherical} canvas, of
which the sheet of paper or computer screen in front of him is simply
a small part.  This sphere represents the plane with a `point' at
infinity, under the usual stereographic projection. We do not say
where this point at infinity is, other than that it never lies on a
drawn arc.  In this way, a diagram in which the spherical canvas is
divided into $n$ cells represents $n$ different configurations in the
plane---one for each of the cells in which the point at infinity may
be located. For example, Fig~.\ref{fig:arcs0} represents three
topologically distinct configurations in $\R^2$, and, as such, depicts
the arcs $\gamma_0, \ldots, \gamma_9$, $\gamma'_1, \ldots, \gamma'_8$,
$\chi_1$, $\chi_2$, $\chi_3$ and points $o_1$, $o_2$ in full
generality.  All diagrams in this proof are to be interpreted in this
way.  We stress that our `spherical diagrams' are simply a convenient
device for using one drawing to represent several possible
configurations in the Euclidean plane: in particular, we are
interested only in the satisfiability of of $\cBCc$-formulas over
$\RCP(\R^2)$ and $\RC(\R^2)$, not over the regular closed algebra of
any other space!  For ease of reference, we refer to the the two
rectangles in Fig~.\ref{fig:arcs0} as the `upper window' and `lower
window', it being understood that these are simply handy labels: in
particular, either of these `windows' (but not both) may be unbounded.

\bigskip

\noindent
\textbf{Stage 2.}  In this stage, we we construct two sequences of
arcs, $\set{\zeta_i}$, $\set{\eta_i}$ of indeterminate length $n \geq
1$, such that the members of the former sequence all lie in the lower
window. Here and in the sequel, we write $\md{k}$ to denote $k$ modulo
3.  Let $\tseq{a}$, $\tseq{b}$, $\tseq{a}_{i,j}$ and $\tseq{b}_{i,j}$
($0 \leq i < 3$, $1 \leq j \leq 6$) be 3-region variables, let $z$ be
an ordinary region-variable, and consider the formulas
\begin{align}
\label{eq:aSeq1}
& (s_6 \leq \inner{a}) \wedge (s'_6 \leq \inner{b})
\wedge (s_3 \leq \intermediate{a}_{0,3}), \\
\label{eq:aSeq:b}
& \stack_z(\tseq{a}_{\md{i-1},3}, \tseq{b}_{i,1}, \ldots, \tseq{b}_{i,6}, \tseq{b}),\\
\label{eq:aSeq:a}
& \stack(\tseq{b}_{i,3}, \tseq{a}_{i,1}, \ldots, \tseq{a}_{i,6}, \tseq{a}).
\end{align}
The arrangement of polygonal 3-regions depicted in
Fig.~\ref{fig:concrete2} (with $z$ assigned appropriately) is one such
satisfying assignment.
\begin{figure}
\resizebox{9cm}{!}{\input{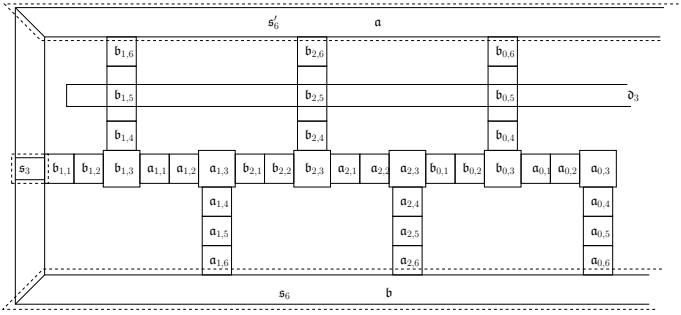}}
\caption{A tuple of 3-regions
  satisfying~{\eqref{eq:aSeq1}}--{\eqref{eq:aSeq:a}}.
   The arrangement
  of components of the $\tseq{a}_{i,j}$ and $\tseq{b}_{i,j}$ repeats
  an indeterminate number of times. The 3-regions $\tseq{a}$,
  $\tseq{b}$ and one component of $\tseq{a}_{0,3}$ are shown in dotted lines.
The 3-regions $\tseq{s}_3$, $\tseq{s}_6$, $\tseq{s}'_6$ and $\tseq{d}_3$
are as in {Fig~\ref{fig:concrete2}}, but not drawn to scale.}
\label{fig:concrete2}
\end{figure}
We stipulate that~\eqref{eq:pcp:C} applies now to all regions depicted
in either Fig~\ref{fig:concrete1} or Fig~\ref{fig:concrete2}. Again,
    these additional constraints are evidently satisfiable.

It will be convenient in this stage to rename the arcs $\gamma_6$ and
$\gamma'_6$ as $\lambda_0$ and $\mu_0$, respectively.  Thus,
$\lambda_0$ forms the bottom edge of the lower window, and $\mu_0$ the
top edge of the upper window. Likewise, we rename $\gamma_3$ as
$\alpha_0$, forming part of the left-hand side of the lower
window. Let $\tilde{q}_{1,1}$ be any point of $\alpha_0$, $p^*$ any
point of $\lambda_0$, and $q^*$ any point of $\mu_0$ (see
Fig.~\ref{fig:arcs0}).  By~\eqref{eq:aSeq1}, then, $\tilde{q}_{1,1}
\in \intermediate{a}_{0,3}$, $p^* \in \inner{a}$, and $q^* \in
\inner{b}$.  Adding the constraint
\begin{equation*}
\neg C(s_3,z),
\end{equation*}
further ensures that $\tilde{q}_{1,1} \in -z$.  By
Lemma~\ref{lma:stackLemma} and~\eqref{eq:aSeq:b}, we may draw an arc
$\tilde{\beta}_1$ from $\tilde{q}_{1,1}$ to $q^*$, with successive
segments $\tilde{\beta}_{1,1}$, $\beta_{1,2}$, \ldots, $\beta_{1,5}$,
$\tilde{\beta}_{1,6}$ lying in the respective regions $a_{0,3} +
b_{1,1}$, $b_{1,2}$, \dots, $b_{1,5}$, $b_{1,6} + b$; further, we can
guarantee that $\beta_{1,2}$ contains a point $\tilde{p}_{1,1} \in
\intermediate{b}_{1,3}$.  Denote the last point of $\beta_{1,5}$ by
$q_{1,2}$. Also, let $q_{1,1}$ be the last point of $\tilde{\beta}_1$
lying on $\alpha_0$, and $q_{1,3}$ the first point of
$\tilde{\beta}_1$ lying on $\mu_0$ Finally, let $\beta_1$ be the
segment of $\tilde{\beta}_1$ between $q_{1,1}$ and $q_{1,2}$; and we
let $\mu_1$ be the segment of $\tilde{\beta}_1$ from $q_{1,2}$ to
$q_{1,3}$ followed by the final segment of $\mu_0$ from $q_{1,3}$.
(Fig.~\ref{subfig:arcs1}).  By repeatedly using the constraints
in~\eqref{eq:pcp:C}, it is easy to see that that $\beta_1$ together
with the initial segment of $\mu_1$ up to $q_{1,3}$ form a chord of
$\Gamma$. Adding the constraints
\begin{equation*}
c(b_{0,5} + d_3),
\end{equation*}
and taking into account the constraints in~\eqref{eq:pcp:C} ensures
that $\beta_1$ and $\chi$ lie in the same residual domain of $\Gamma$,
as shown.  The wiggly lines indicate that we do not care about the
exact positions of $\tilde{q}_{1,1}$ or $q^*$; otherwise,
Fig.~\ref{subfig:arcs1}) is again completely general.
\begin{figure}
\begin{center}
\subfloat[The arc $\beta_1$.]{
\label{subfig:arcs1}
\resizebox{8cm}{!}{\input{5-1-Graphics/subArcs1.pstex_t}}}\\
\subfloat[The arc $\alpha_1$.]{
\label{subfig:arcs2}
\resizebox{8cm}{!}{\input{5-1-Graphics/subArcs2.pstex_t}}}\\
\subfloat[The arc $\beta_2$.]{
\label{subfig:arcs3}
\resizebox{8cm}{!}{\input{5-1-Graphics/subArcs3.pstex_t}}}\\
\subfloat[The arc $\alpha_2$.]{
\label{subfig:arcs4}
\resizebox{8cm}{!}{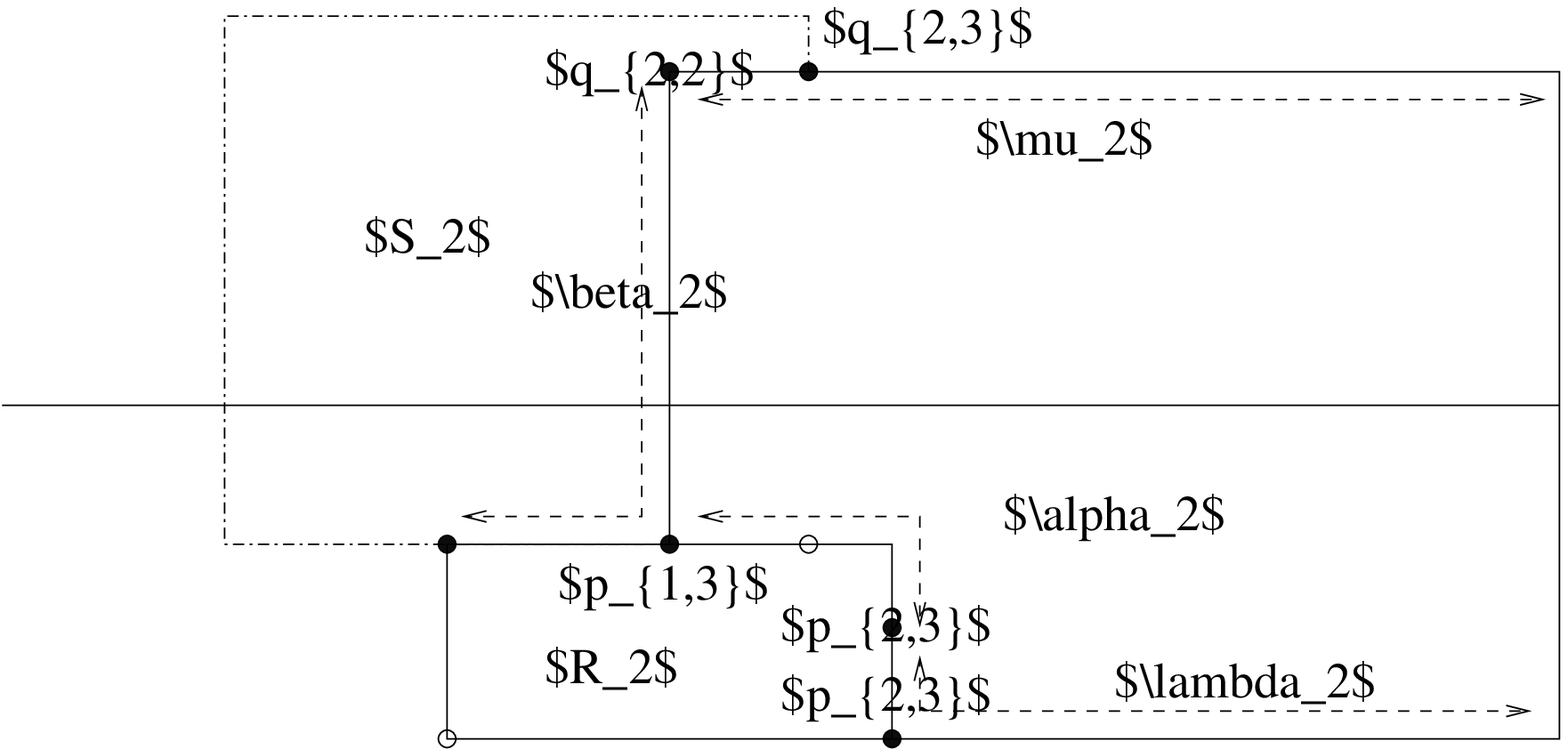}}
\end{center}
\caption{Construction of the arcs $\set{\alpha_i}$ and $\set{\beta_i}$}
\label{fig:arcsAlphaBeta}
\end{figure}
Note that $\mu_1$ lies entirely in $b_{1,6} + b$, and hence
certainly in the region
\begin{equation*}
b^* = b + b_{0,6} + b_{1,6} + b_{2,6}.
\end{equation*}

Recall that $\tilde{p}_{1,1} \in \intermediate{b}_{1,3}$, and $p^* \in
\inner{a}$.  By Lemma~\ref{lma:stackLemma} and~\eqref{eq:aSeq:a}, we may draw an arc
$\tilde{\alpha}_1$ from $\tilde{p}_{1,1}$ to $p^*$, with successive
segments $\tilde{\alpha}_{1,1}$, $\alpha_{1,2}$, \ldots,
$\alpha_{1,5}$, $\tilde{\alpha}_{1,6}$ lying in the respective regions
$b_{1,3} + a_{1,1}$, $a_{1,2}$, \dots, $a_{1,5}$ $a_{1,6} + a$;
further, we can guarantee that the segment lying in $a_{1,3}$ contains
a point $\tilde{q}_{2,1}\in \intermediate{a}_{1,3}$.  Denote the last
point of $\alpha_{1,5}$ by $p_{1,2}$.  Also, let $p_{1,1}$ be the last
point of $\tilde{\alpha}_1$ lying on $\beta_1$, and $p_{1,3}$ the
first point of $\tilde{\alpha}_1$ lying on $\lambda_0$.
From~\eqref{eq:pcp:C}, these points must be arranged as shown in
Fig.~\ref{subfig:arcs2}.  Let $\alpha_1$ be the segment of
$\tilde{\alpha}_1$ between $p_{1,1}$ and $p_{1,2}$.  Noting
that~\eqref{eq:pcp:C} entails
\begin{align*}
& \neg C(a_{1,k}, s_0 + s_9 + d_0+ \cdots + d_5 ) & & \qquad{1 \leq k
    \leq 6},
\end{align*}
we can be sure that $\alpha_1$ lies entirely in the `lower' window,
whence $\beta_1$ crosses the central chord, $\chi$, at least once. Let
$o_1$ be the first such point (measured along $\chi$ from left to
right).  Finally, let $\lambda_1$ be the segment of $\tilde{\alpha}_1$
between $p_{1,2}$ and $p_{1,3}$, followed by the final segment of
$\lambda_0$ from $p_{1,3}$. Note that $\lambda_1$ lies entirely in
$a_{1,6} + a$, and hence certainly in the region
\begin{equation*}
a^* = a + a_{0,6} + a_{1,6} + a_{2,6}.
\end{equation*}
We remark that, in Fig.~\ref{subfig:arcs2}, the arcs $\beta_1$ and
$\mu_1$ have been slightly re-drawn, for clarity.  The region marked
$S_1$ may now be forgotten, and is suppressed in Figs.~\ref{subfig:arcs3}
and~\ref{subfig:arcs4}.

By construction, the point $\tilde{q}_{2,1}$ lies in some component of
$\intermediate{a}_{1,3}$, and, from the presence of the `switching'
variable $z$ in~\eqref{eq:aSeq:a}, that component is either included
in $z$ or included in $-z$. Suppose the latter.  Then we can repeat
the above construction to obtain an arc $\tilde{\beta}_2$ from
$\tilde{q}_{2,1}$ to $q^*$, with successive segments
$\tilde{\beta}_{2,1}$, $\beta_{2,2}$, \ldots, $\beta_{2,5}$,
$\tilde{\beta}_{2,6}$ lying in the respective regions $a_{1,3} +
b_{2,1}$, $b_{2,2}$, \dots, $b_{2,5}$, $b_{2,6} + b$; further, we can
guarantee that $\beta_{2,2}$ contains a point $\tilde{p}_{2,1} \in
\intermediate{b}_{2,3}$.  Denote the last point of $\beta_{2,5}$ by
$q_{2,2}$. Also, let $q_{2,1}$ be the last point of $\tilde{\beta}_2$
lying on $\alpha_1$, and $q_{2,3}$ the first point of
$\tilde{\beta}_2$ lying on $\mu_1$.  Again, we let $\beta_2$ be the
segment of $\tilde{\beta}_2$ between $q_{2,1}$ and $q_{2,2}$; and we
let $\mu_2$ be the segment of $\tilde{\beta}_2$ from $q_{2,1}$ to
$q_{2,3}$, followed by the final segment of $\mu_1$ from $q_{2,3}$.
Note that $\mu_2$ lies in the set $b^*$. It is easy to see that $\beta_2$
must be drawn as shown in Fig.~\ref{subfig:arcs3}: in particular,
$\beta_2$ cannot enter the interior of the region marked $R_1$. For,
by construction, $\beta_2$ can have only one point of contact with
$\alpha_1$, and the constraints~\eqref{eq:pcp:C} ensure that $\beta_2$
cannot intersect any other part of $\delta R_1$; since $q^* \in a$ is
guaranteed to lie outside $R_1$, we evidently have $\beta_2
\subseteq -R_1$. This observation having been made, $R_1$ may now be
forgotten.

Symmetrically, we construct the arc $\tilde{\alpha}_2 \subseteq
b_{1,3} + a_{2,1} + \cdots + a_{2,6} + a$, and points $p_{2,1}$,
$p_{2,2}$, $p_{2,3}$, together with the arcs arcs $\alpha_2$ and
$\lambda_2$, as shown in Fig.~\ref{subfig:arcs4} (where the region
$R_1$ has been suppressed and the region $S_2$ slightly
re-drawn). Again, we know from~\eqref{eq:pcp:C} that $\alpha_2$
lies entirely in the `lower' window, whence $\beta_2$ must cross the
central chord, $\chi$, at least once. Let $o_2$ be the first such
point (measured along $\chi$ from left to right).

This process continues, generating arcs $\beta_i \subseteq
a_{\md{i-1},3} + b_{\md{i},1} + \cdots + b_{\md{i},5}$ and $\alpha_i
\subseteq b_{\md{i},3} + a_{\md{i},1} + \cdots + a_{\md{i},5}$, as
long as $\alpha_i$ contains a point $\tilde{q}_{i,1} \in -z$. That we
eventually reach a value $i = n$ for which no such point exists
follows from~\eqref{eq:pcp:C}. For the conjuncts $\neg C(b_{i,j},
d_k)$ ($j \neq 5$) together entail $o_i \in b_{\md{i},5}$, for every
$i$ such that $\beta_i$ is defined; and these points cycle on $\chi$
through the regions $b_{0,5}$, $b_{1,5}$ and $b_{2,5}$. If there were
infinitely many $\beta_i$, the $o_i$ would have an accumulation point,
lying in all three regions, contradicting, say, $\neg
C(b_{0,5},b_{1,5})$.  The resulting sequence of arcs and points is
shown, schematically, in Fig.~\ref{fig:arcs2}.
\begin{figure}
\resizebox{8.5cm}{!}{\input{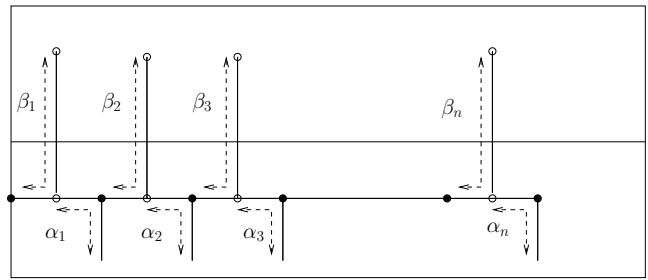}}
\caption{The sequences of arcs $\set{\alpha_i}$ and $\set{\beta_i}$.}
\label{fig:arcs2}
\end{figure}

We finish this stage in the construction by `re-packaging' the arcs
$\set{\alpha_i}$ and $\set{\beta_i}$, as illustrated in
Fig.~\ref{fig:arcs3}. Specifically,
for all $i$ ($1 \leq i \leq n$), let $\zeta_i$ be the initial segment
of $\beta_i$ up to the point $p_{i,1}$ followed by the initial segment
of $\alpha_i$ up to the point $q_{i+1,1}$; and let $\eta_i$ be the
final segment of $\beta_i$ from the point $p_{i,1}$:
\begin{align*}
& \zeta_i = \beta_i[q_{i,1},p_{i,1}]\alpha_i[p_{i,2},q_{i+1,1}]\\
& \eta_i = \beta_i[p_{i,1},q_{i,2}].
\end{align*}
The final
segment of $\alpha_i$ from the point $q_{i+1}$ may be forgotten.
\begin{figure}
\resizebox{8cm}{!}{\input{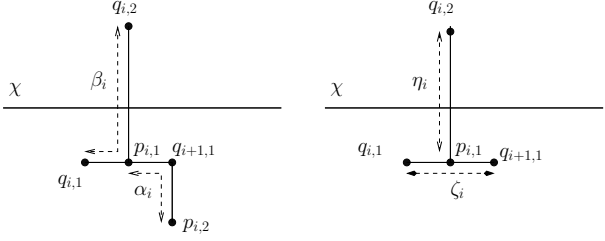}}
\caption{`Re-packaging' of $\alpha_i$ and $\beta_i$ into $\zeta_i$ and
  $\eta_i$: before and after.}
\label{fig:arcs3}
\end{figure}
Defining, for $0 \leq i < 3$,
\begin{eqnarray*}
a_i & = & a_{1-i,3} + b_{i,1} + \cdots + b_{i,4} + a_{i,1} + \cdots + a_{i,4}\\
b_i & = & b_{i,2} + \cdots + b_{i,5},
\end{eqnarray*}
the constraints~\eqref{eq:pcp:C} guarantee that, for $1 \leq i \leq
n$,
\begin{eqnarray*}
\zeta_i & \subseteq & a_{\md{i}}\\ \eta_i & \subseteq & b_{\md{i}}.
\end{eqnarray*}
Observe that the arcs $\zeta_i$ are located entirely in the `lower
window', and that each arc $\eta_i$ connects $\zeta_i$ to some point
$q_{i,2}$, which in turn is connected to a point $q^* \in \lambda_0$
by an arc in $b^*$.

\bigskip

\noindent
\textbf{Stage 3.}  We now repeat Stage~2 symmetrically, with the
`upper' and `lower' windows exchanged. Let $\tseq{a}'_{i,j}$,
$\tseq{b}'_{i,j}$ be 3-region variables (with indices in the same
ranges as for $\tseq{a}_{i,j}$, $\tseq{b}_{i,j}$). Let $\tseq{a}' =
\tseq{b}$, $\tseq{b}' = \tseq{a}$; and let
\begin{eqnarray*}
a'_i & =  & a'_{1-i,3} + b'_{i,1} + \cdots + b'_{i,4} + a'_{i,1} + \cdots + a'_{i,4}\\
b'_i & =  & b'_{i,2} + \cdots + b'_{i,5},
\end{eqnarray*}
for $0 \leq i < 2$. The constraints
\begin{align*}
& (s'_3 \leq \intermediate{a}'_{0,3}) \\
& \stack_z(\tseq{a}'_{\md{i-1},3}, \tseq{b}'_{i,1}, \ldots, \tseq{b}'_{i,6}, \tseq{b}'),\\
& \stack(\tseq{b}'_{i,3}, \tseq{a}'_{k,1}, \ldots, \tseq{a}'_{i,6}, \tseq{a}')\\
& c(b'_{0,5} + d_3)
\end{align*}
then establish sequences of arcs $\set{\zeta'_i}$, $\set{\eta'_i}$,
($1 \leq i \leq n'$) satisfying
\begin{eqnarray*}
\zeta_i' & \subseteq & a_{\md{i}}'\\
\eta_i' & \subseteq & b_{\md{i}}'
\end{eqnarray*}
for $1 \leq i \leq n'$. The arcs $\zeta'_i$ are located entirely in the
`upper window', and each arc $\eta'_i$ connects $\zeta'_i$ to a point
$p_{i,2}$, which in turn is connected to a point $p^*$ by an arc in
the region
\begin{align*}
& {b^*}' = b' + b'_{0,6} + b'_{1,6} + b'_{2,6}.
\end{align*}
Our next task is to write constraints to ensure that $n = n'$, and
that, furthermore, each $\eta_i$ (also each $\eta'_i$) connects
$\zeta_i$ to $\zeta'_i$, for $1\leq i\leq n=n'$. Let $z^*$ be a new
region-variable, and write
\begin{equation*}
\neg C(z^*, s_0 + \cdots + s_9 + s'_1 + \cdots + s'_8 + d_1 + \cdots + d_4 + d_6).
\end{equation*}
Note that $d_5$ does not appear in this constraint, which ensures that
the only arc depicted in Fig.~\ref{fig:arcs0} which $z$ may intersect
is $\chi_3$. Recalling that $\alpha_n$ and $\alpha'_{n'}$ contain
points $q_{n,1}$ and $q'_{n',1}$, respectively, both lying in $z$, the
constraints
\begin{equation*}
c(z) \wedge \neg C(z, -z^*)
\end{equation*}
ensure that $q_{n,1}$ and $q'_{n',1}$ may be joined by an arc, say
$\zeta^*$, lying in $\ti{(z^*)}$, and also lying entirely in the upper
and lower windows, crossing $\chi$ only in $\chi_3$.  Without loss of
generality, we may assume that $\zeta^*$ contacts $\zeta_n$ and
$\zeta'_{n'}$ in just one point. Bearing in mind that the
constraints~\eqref{eq:pcp:C} force $\eta_n$ and $\eta'_{n'}$ to
cross $\chi$ in its central section, $\chi_2$, writing
\begin{align}
& \neg C(b_{i,j}, z)  \wedge \neg C(b'_{i,j}, z)
\label{eq:zeta}
\end{align}
for all $i$ ($0 \leq i < 3$) and $j$ ($1 \leq i \leq 6$) ensures that
$\zeta^*$ is (essentially) as shown in Fig.~\ref{fig:arcZeta}.
\begin{figure}
\resizebox{8.5cm}{!}{\input{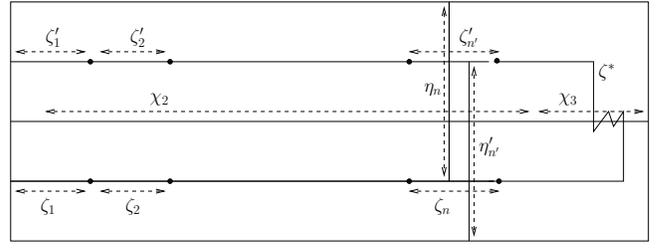}}
\caption{The arc $\zeta^*$.}
\label{fig:arcZeta}
\end{figure}
Now consider the arc $\eta_1$. Recalling that $\eta_1\mu_1$ joins
$\zeta_1$ to the point $q^*$ (on the upper edge of the upper window),
crossing $\chi_2$, we see by inspection of Fig.~\ref{fig:arcZeta}
that~\eqref{eq:zeta} together with
\begin{align*}
&  \neg C(a'_i, b^*)
\end{align*}
for $0 \leq i < 3$
forces $\eta_1$ to cross one of the arcs
$\zeta'_{j'}$ ($1 \leq j' \leq n'$); and the constraints
\begin{align*}
& \neg C(a'_i,b_j)
\end{align*}
for $0 \leq i < 3$, $0 \leq j < 3$, $i \neq j$, ensure that $j' \equiv
1 $ modulo 3. We write the symmetric constraints
\begin{align}
& \neg C(a_i,b'_j)
\label{eq:abPrime}
\end{align}
for $0 \leq i < 3$, $0 \leq j < 3$, $i \neq j$, together with 
\begin{align}
& \neg C(b_i,b'_j)
\label{eq:bbPrime}
\end{align}
for $0 \leq i < j \leq 3$.  Now suppose $j' \geq 4$.  The arc $\eta'_2
\lambda'_2$ must connect $\zeta'_2$ to the point $p^*$ on the bottom
edge of the lower window, which is now impossible without $\eta'_2$
crossing either $\zeta_1$ or $\eta_1$---both forbidden
by~\eqref{eq:abPrime}--\eqref{eq:bbPrime}.  Thus, $\eta_1$ intersects
$\zeta'_j$ if and only if $j = 1$. Symmetrically, $\eta'_1$ intersects
$\zeta_j$ if and only if $j = 1$.  And the reasoning can now be
repeated for $\eta_2$, $\eta_2'$, $\eta_3$, $\eta_3'$ \ldots, leading
to the 1--1 correspondence depicted in
Fig.~\ref{fig:arcCorrespondence}.
\begin{figure}
\resizebox{8cm}{!}{\input{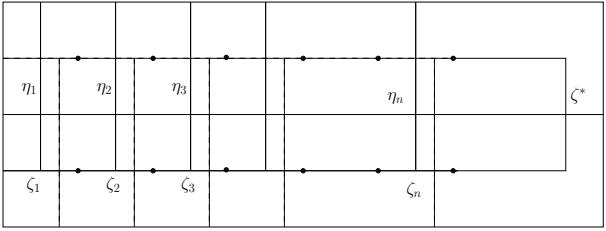}}
\caption{The 1--1 correspondence between the $\zeta_i$ and the
  $\zeta'_i$ established by the $\eta_i$ and the $\eta'_i$.}
\label{fig:arcCorrespondence}
\end{figure}
In particular, we are guaranteed that $n = n'$.

\bigskip

\noindent
\textbf{Stage 4.}  Recall the given PCP-instance, $\fW = (\set{0,1},
T, \fw_1, \fw_2)$.  We think of $T$ as a set of `tiles', and the
morphisms $\fw_1$, $\fw_2$ as specifying, respectively, the `lower'
and `upper' strings of each tile.  In this stage, we shall `label' the arcs
$\zeta_1, \ldots, \zeta_n$, with elements of $\set{0,1}$, thus
defining a word $\sigma$ over this alphabet. Using a slightly more
complicated labelling scheme, we shall label the arcs $\eta_1, \ldots,
\eta_n$ so as to define a word $\tau$ (of length $m \leq n$) over the
alphabet $T$; likewise we shall label the arcs $\eta'_1, \ldots, \eta'_n$
so as to define another word $\tau'$ (of length $m' \leq n$) over $T$.

We begin with the $\zeta_i$. Consider the constraints
\begin{equation*}
b_i \leq l_0 + l_1 \wedge \neg C(b_i \cdot l_0, b_i \cdot l_1) \qquad (i = 0,\ 1).
\end{equation*}
By Lemma~\ref{lma:labelling}, in any satisfying assignment over
$\RC(\R^2)$, every arc $\eta_i$ ($1 \leq i \leq n$) is included in
(`labelled with') exactly one of the regions $l_0$ or $l_1$, so that
the sequence of arcs $\eta_1, \ldots, \eta_n$ defines a word $\sigma
\in \set{0,1}^*$, with $|w| = n$.

Turning our attention now to the $\zeta_i$, let us write $T =
\set{t_1, \ldots, t_\ell}$. For all $j$ ($1 \leq j \leq \ell$), we
shall write $\sigma_j = \fw_1(t_j)$ and $\sigma'_j = \fw_2(t_j)$;
further, we denote $|\sigma_j|$ by $u(j)$ and $|\sigma'_j|$ by
$u'(j)$. (Thus, by assumption, the $u(j)$ and $u'(j)$ are all
positive.)

Now let $t_{j,k}$ ($1 \leq j \leq \ell$, $1 \leq k \leq u(j)$) and
$t'_{j,k}$ ($1 \leq j \leq \ell$, $1 \leq k \leq u'(j)$) be fresh
region variables.  We think of $t_{j,k}$ as standing for the $k$th
letter in the word $\sigma_j$, and likewise think of $t'_{j,k}$ as
standing for the $k$th letter in the word $\sigma'_j$.  By
Lemma~\ref{lma:labelling}, we may write constraints ensuring that each
component of either $a_0$, $a_1$ or $a_2$---and hence each of the arcs
$\zeta_1, \ldots, \zeta_n$---is `labelled with' one of the $t_{j,k}$,
in the by-now familiar sense.  Further, we can ensure that these
labels are organized into (contiguous) blocks, $E_1, \ldots, E_{m}$
such that, in the $h$th block, $E_h$, the sequence of labels reads
$t_{j,1}, \ldots, t_{j,{u(j)}}$, for some fixed $j$ ($1 \leq j \leq
\ell$).  This amounts to insisting that: ({\em i}) the very first arc,
$\zeta_1$, must be labelled with $t_{j,1}$ for some $j$; ({\em ii})
if, $\zeta_i$ is labelled with $t_{j,k}$, where $i < n$ and $k <
u(j)$, then the next arc, namely $\zeta_{i+1}$, must be labelled with
the next letter of $\sigma_j$, namely $t_{j,k+1}$; ({\em iii}) if
$\zeta_i$ ($i < n$) is labelled with the final letter of $w_j$, then
the next arc must be labelled with the initial letter of some possibly
different word $\sigma_{j'}$; and ({\em iv}) $\zeta_n$ must be
labelled with the final letter of some word. To do this we simply
write:
\begin{align*}
& \neg C(t_{j,i}, s_3) & & (\mbox{if } i \neq 1)\\
& \neg C(a_k \cdot t_{j,i}, a_{\lfloor k+1 \rfloor} \cdot t_{j',i'}) & &
\begin{array}{l}
\text{($i < u(j)$ and either}\\
\text{$j' \neq j$ or $i' \neq i+1$)}
\end{array}\\
& \neg C(a_k \cdot t_{j,u(j)}, a_{\lfloor k+1 \rfloor} \cdot t_{j',i'}) & &
(\mbox{if } i' \neq 1)\\
& \neg C(t_{j,i}, z^*) & & (\mbox{if } i \neq u(j)),
\end{align*}
where $1 \leq j, j' \leq \ell$, $1 \leq i \leq u(j)$ and $1\leq i'
\leq u(j')$.

Thus, within each block $E_h$, the labels read $t'_{j,1}, \ldots,
t'_{j,u'(j)}$, for some fixed $j$; we write $j(h)$ to denote the
common subscript $j$.  The sequence of indices $j(1), \ldots, j(m)$
corresponding to the successive blocks thus defines a word $\tau =
t_{j(1)}, \ldots t_{j(m)} \in T^*$.

Using corresponding formulas, we label the arcs $\zeta'_i$ ($1 \leq i
\leq n$) with the alphabet $\set{t'_{j,k} \mid 1 \leq j \leq \ell, 1
  \leq k \leq u'(j)}$, so that, in any satisfying assignment over
$\RC(\R^2)$, every arc $\zeta'_i$ ($1 \leq i \leq n$) is labelled with
exactly one of the regions $t'_{j,k}$. Further, we can ensure that
these labels are organized into (say) $m'$ contiguous blocks, $E'_1,
\ldots, E'_{m'}$ such that in the $h$th block, $E_h'$, the sequence of
labels reads $t'_{j,1}, \ldots, t'_{j,u'(j)}$, for some fixed
$j$. Again, writing $j'(h)$ for the common value of $j$, the sequence
of of indices $j'(1), \ldots, j'(m')$ corresponding to the successive
blocks defines a word $\tau' = t_{j'(1)}, \ldots t_{j'(m')} \in T^*$.

\bigskip

\noindent
\textbf{Stage 5.}  The basic job of the foregoing stages was to define
the words $\sigma \in \set{0,1}^*$ and $\tau, \tau' \in T^*$.  In this
stage, we enforce the equations $\sigma = \fw_1(\tau)$, $\sigma =
\fw_2(\tau')$ and $\tau = \tau'$. That is: the PCP-instance $\fW =
(\set{0,1}, T, \fw_1, \fw_2)$ is positive.

We first add the constraints
\begin{align*}
& \neg C(l_h, t_{j,k}) \qquad \text{the $k$'th letter of $\sigma_j$ is not $h$}\\
& \neg C(l_h, t'_{j,k}) \qquad \text{the $k$'th letter of $\sigma'_j$ is not $h$}.
\end{align*}
Since $\eta_i$ is in contact with $\zeta_i$ for all $i$ ($1 \leq i \leq
n$), the string $\sigma \in \set{0,1}^*$ defined by the arcs $\eta_i$
must be identical to the string $\sigma_{j(1)} \cdots
\sigma_{j(m)}$. But this is just to say that $\sigma = \fw_1(\tau)$.
The equation $\fw_2(\tau') = \sigma$ may be secured similarly.

It remains only to show that $\tau = \tau'$.  That is, we must show
that $m = m'$ and that, for all $h$ ($1 \leq h \leq m$), $j(h) =
j'(h)$. The techniques required have in fact already been encountered
in Stage~3.  We first introduce a new pair of variables, $f_0$, $f_1$,
which we refer to as `block colours', and with which we label the arcs
$\zeta_i$ in the fashion of Lemma~\ref{lma:labelling}, using the
constraints:
\begin{align*}
& (a_0 + a_1 + a_2) \leq (f_0 + f_1)\\
& \neg C(f_0 \cdot a_i, f_1 \cdot a_i), & & (0 \leq i <3).
\end{align*}
We force all arcs in each block $E_j$ to have a uniform block colour,
and we force the block colours to alternate by writing, for $0 \leq h <2$,
$1 \leq j,j' \leq \ell$, $1 \leq k < u(j)$ and $0\leq i<3$:
\begin{align*}
& \neg C(f_h \cdot t_{j,k}, f_{\md{h+1}}\cdot t_{j,k+1}), 
\\& \neg C(f_h \cdot t_{j,u(j)}\cdot a_i, f_{h}\cdot t'_{j',1}\cdot a_{\md{i+1}}) 
\end{align*}
Thus, we may speak unambiguously of the colour ($f_0$ or $f_1$) of a
block: if $E_1$ is coloured $f_0$, then $E_2$ will be coloured $f_1$,
$E_3$ coloured $f_0$, and so on.  Using the the {\em same} variables
$f_0$ and $f_1$, we similarly establish a block structure $E'_1,
\ldots, E'_{m'}$ on the arcs $\eta'_i$. (Note that there is no need
for primed versions of $f_0$ and $f_1$.)

Now we can match up the blocks in a 1--1
fashion just as we matched up the individual arcs. Let $\tseq{g}_0$,
$\tseq{g}_1$, $\tseq{g}'_0$ and $\tseq{g}'_1$ be new 3-regions
variables.  We may assume that every arc $\zeta_i$ contains some point
of $\intermediate{b}_{\md{i},1}$. We wish to connect any such arc that
starts a block $E_h$ (i.e. any $\zeta_i$ labelled by $t_{j,1}$ for
some $j$) to the top edge of the upper window, with the connecting arc
depending on the block colour.  Setting $w_k = -(f_k \cdot \sum_{i =
  1}^{i=\ell} t_{j,1})$ ($0 \leq k < 2$), we can do this using the
constraints:
\begin{align*}
& \stack_{w_k}(\tseq{b}_{i,1}, \tseq{g}_k, \tseq{a}) & & (1 \leq k < 2, 
0 \leq i < 3).
\end{align*}
Specifically, the first arc in each block $E_h$ ($1 \leq h \leq m$) is
connected by an arc $\theta_h\tilde{\theta}_h$ to some point on the
upper edge of the upper window, where $\theta_h \subseteq b_{i,1} + g_i$
and $\tilde{\theta}_h \subseteq a$.  Similarly, setting
$w'_k = -(f_k \cdot \sum_{i = 1}^{i=\ell} t'_{j,1})$ ($0 \leq k < 2$), 
the constraints
\begin{align*}
& \stack_{w_k'}(\tseq{b}'_{i,1}, \tseq{g}'_k, \tseq{b}) & & (1 \leq k < 2, 
0 \leq i < 3)
\end{align*}
ensure that the first arc in each block $E_{h'}$ ($1 \leq h' \leq
m'$) is connected by an arc $\theta'_{h'}\tilde{\theta}'_{h'}$ to some
point on the bottom edge of the lower window, where $\theta_{h'}
\subseteq b'_{i,1} + g'_i$ and $\tilde{\theta}'_{h'} \subseteq b$.
Furthermore, from the arrangement of the $\zeta_i$,
$\zeta'_i$ and $\zeta^*$ (Fig.~\ref{fig:arcZeta}) we can easily write
non-contact constraints
forcing each $\theta_h$ to intersect one of the arcs $\zeta'_i$ ($1
\leq i \leq n$), and each $\theta'_h$ to intersect one of the arcs
$\zeta_{i'}$ ($1 \leq i' \leq n$).

We now write the constraints
\begin{align*}
& \neg C(g_k, f_{1-k}) \wedge  \neg C(g'_k, f_{1-k})  & & (0 \leq k < 2).
\end{align*}
Thus, any $\theta_h$ included in $g_k$ must join some arc $\zeta_i$ in
a block with colour $f_k$ to some arc $\zeta'_{i'}$ also in a block
with colour $f_k$; and similarly for the $\theta'_h$.
Adding
\begin{align*}
& \neg C(g_0 + g'_0, g_1 + g'_1)
\end{align*}
then ensures, via reasoning exactly similar to that employed in
Stage~3, that $\theta_1$ connects the block $E_1$ to the block $E'_1$,
$\theta_2$ connects $E_2$ to $E'_2$, and so on; and similarly for the
$\theta'_h$ (as shown, schematically, in
Fig.~\ref{fig:blockCorrespondence}).  Thus, we have a 1--1
correspondence between the two sets of blocks, whence $m = m'$.
\begin{figure}
\resizebox{9cm}{!}{\input{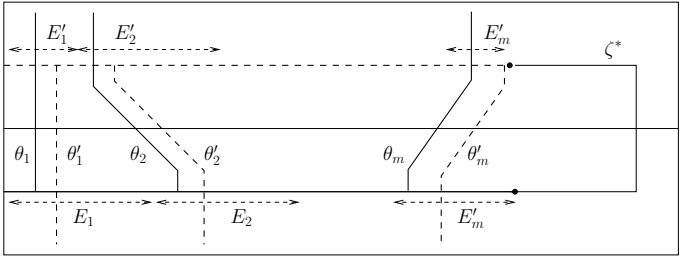}}
\caption{The 1--1 correspondence between the $E_h$ and the
  $E'_h$ established by the $\theta_i$ and the $\theta'_i$.}
\label{fig:blockCorrespondence}
\end{figure}

Finally, we let $d_1, \ldots, d_\ell$ be
new regions variables labelling the components of $g_0$ and of $g_1$,
and hence the arcs $\theta_1, \ldots, \theta_m$:
\begin{align*}
& g_i \leq \sum_{1 \leq j \leq \ell} d_j
\wedge \bigwedge_{1 \leq j \leq \ell} C(d_j \cdot g_i, (-d_j) \cdot g_i)
\end{align*}
for $0 \leq i < 2$. Adding the constraints
\begin{align*}
& \neg C(p_{j,k}, d_{j'}) & & (j \neq j')\\
& \neg C(p'_{j,k}, d_{j'}) & & (j \neq j')
\end{align*}
where $1 \leq j \leq \ell$, $1 \leq k \leq u(j)$ and $1 \leq j' \leq
\ell$, instantly ensures that the sequences of tile indices $j(1),
\ldots, j(m)$ and $j'(1), \ldots, j'(m)$ are identical. In other words,
$\tau = \tau'$.  This completes the proof that $\fW$ is a positive
instance of the PCP.

%
%
%
We have established the r.e.-hardness of $\Sat(\cBCc,\RC(\R^2))$ and
$\Sat(\cBCc,\RCP(\R^2))$. We must now extend these results to the
other languages considered here. We deal with the languages $\cBCci$
and $\cBc$ as in Sec.~\ref{sec:sensitivity}. Let $\ti{\phi}_\fW$ be
the $\cBCci$ formula obtained by replacing all of occurrences of $c$
in $\phi_\fW$ with $\ic$. Since all occurrences of $c$ in $\phi_\fW$
are positive, $\ti{\phi}_\fW$ entails $\phi_\fW$. On the other
hand, the connected regions satisfying $\phi_\fW$ are also
interior-connected, and thus satisfy $\ti{\phi}_\fW$ as well.

For the language $\cBc$, observe that, as in Sec.~\ref{sec:sensitivity},  
all conjuncts of $\phi_\fW$ featuring the predicate $C$ are {\em negative}. 
(Remember that there are additional such literals
implicit in the use of 3-region variables; but let us ignore these for
the moment.) Recall from Sec.~\ref{sec:sensitivityA} that
\begin{align*}
	\phi_{\lnot C}^c(r,s,r',s'):=c(r+r')\land c(s+s') \hspace{2cm}
		\\\hfill\land \lnot c((r+r')+(s+s')),
\end{align*}
and consider the effect of replacing any literal $\neg C(r,s)$
from~\eqref{eq:pcp:C} with the $\cBc$-formula $\phi_{\lnot
  C}^c(r,s,r',s')$ where $r'$ and $s'$ are fresh variables, and let
the formula obtained be $\psi$. It is easy to see that $\psi$ entails
$\phi_\fW$; hence if $\psi$ is satisfiable, then $\fW$ is a positive
instance of the PCP.  To see that $\psi$ is satisfiable, consider the
satisfying tuple of $\phi_\fW$. Note that if $\tseq{r}$ and $\tseq{s}$
are 3-regions whose outer-most elements $r$ and $s$ are disjoint (for
example: $\tseq{r} = \tseq{a}_{0,1}$, $\tseq{s} = \tseq{a}_{0,3}$),
then $r$ and $s$ have finitely many connected components and have
connected complements. Hence, it is easy to find $r'$ and $s'$ in
$\RCP(\R^2)$ satisfying the corresponding formula $\phi_{\lnot
  C}^c(r,s,r',s')$.  Fig.~\ref{fig:connectingRsAndSs} represents the
situation in full generality. (As usual, we assume a spherical
canvas, with the point at infinity not lying on the
boundary of any of the depicted regions.)  We may therefore assume,
that all such literals involving $C$ have been eliminated from
$\phi_\fW$.

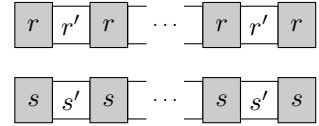
\begin{figure}[h]
	\begin{center}
	\begin{tikzpicture}
		\newcommand{\connectingRsAndSs}[1]
		{
			\draw[fill=black!20] (0,-.3) rectangle (0.5,.3) node[midway]{$#1$};
			\draw (0.5,-.25) rectangle (1,.25) node[midway]{$#1'$};
			\draw[fill=black!20] (1,-.3) rectangle (1.5,.3) node[midway]{$#1$};
			\draw (1.75,-.25) --++ (-.25,0) --++(0,.5)--++(.25,0);
			\node at (2,0){\scriptsize$\ldots$};
			\draw (2.25,-.25) --++ (.25,0) --++(0,.5)--++(-.25,0);
			\draw[fill=black!20] (2.5,-.3) rectangle (3,.3) node[midway]{$#1$};
			\draw (3,-.25) rectangle (3.5,.25) node[midway]{$#1'$};
			\draw[fill=black!20] (3.5,-.3) rectangle (4,.3) node[midway]{$#1$};
		}		
		
		\connectingRsAndSs{r}		
		\begin{scope}[yshift=-1cm]
			\connectingRsAndSs{s}
		\end{scope}
	\end{tikzpicture}
	\end{center}
	\caption{Satisfying $\phi_{\lnot C}^c(r,s,r',s')$}
	\label{fig:connectingRsAndSs}
\end{figure}

We are not quite done, however. We must show that we can replace the
{\em implicit} non-contact constraints that come with the use of
3-region variables by suitable $\cBc$-formulas.  For example, a
3-region variable $\tseq{r}$ involves the implicit constraints
$\neg C(\inner{r}, -\intermediate{r})$ and $\neg
C(\intermediate{r}, -r)$. Since the two conjuncts are identical in form,
we only show how to deal with $\neg C(\intermediate{r}, -r)$.
Because the complement of $-r$ is  in general not connected, a direct 
use of $\phi_{\lnot C}^c$ will result in a formula which is not satisfiable. 
Instead, we represent $-r$ as the sum of two regions $s_1$ and $s_2$
with connected complements, and then proceed as before. In particular, we replace 
$\neg C(\intermediate{r}, -r)$ by:
\begin{eqnarray*}
	-r=s_1+s_2\land  \phi_{\lnot C}^c(\intermediate r,s_1,r_1,s_1) 
	\land \phi_{\lnot C}^c(\intermediate r,s_2,r_2,s_2).
\end{eqnarray*}
For $i=1,2$, $\intermediate r+r_i$ is a connected region that is disjoint from
$s_i$. So, $\intermediate r$ is disjoint from $s_1$ and $s_2$, and hence 
disjoint from their sum $-r:=s_1+s_2$. Fig~\ref{fig:crenellate} shows regions 
$s_i,r_i$, for $i=1,2$, which satisfy the above formula.
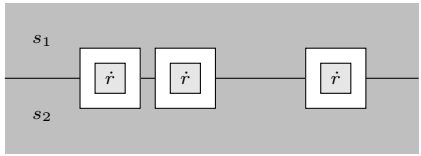
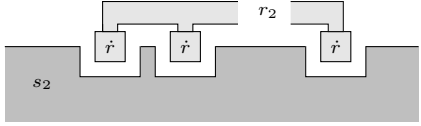
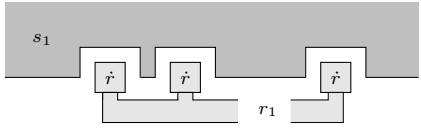
\begin{figure}[h]
	\begin{center}
	\scriptsize{
	\subfloat[The region $-r$ is the sum of $s_1$ and $s_2$.]{
		\begin{tikzpicture}
			\coordinate (P) at (0,0);
			\fill[black!25] (-1,1) rectangle (4.5,0);		
			\node at (-.5,.5) {$s_1$};
			\fill[black!25] (-1,-1) rectangle (4.5,0);
			\node at (-.5,-.5) {$s_2$};
			\draw (-1,0)--(4.5,0);
			\foreach \x in {0,1,2,3}
			{
				\ifnum \x=2					
				\else
					\draw[fill=white] ($(P)+(0,-.4)$) rectangle ($(P)+(.8,.4)$);
					\draw[fill=gray!20] ($(P)+(0.2,-.2)$) rectangle ($(P)+(.6,.2)$) 
						node[midway]{$\intermediate{r}$};
				\fi
				\coordinate (P) at ($(P)+(1,0)$);
			}
		\end{tikzpicture}
	}
	
	\subfloat[The mutually disjoint connected regions $\intermediate r+r_2$ and $s_2$.]
	{
		\begin{tikzpicture}
			\coordinate (P) at (0,0);
			\fill[black!0] (-1,1) rectangle (4.5,0);
			\fill[black!25] (-1,-1) rectangle (4.5,0);
			\node at (-.5,-.5) {$s_2$};
			\draw (-1,0)--(4.5,0);
			\foreach \x in {0,1,2,3}
			{
				\ifnum \x=2					
				\else
					\draw[fill=white] ($(P)+(0,.2pt)$)--($(P)+(0,-.4)$) --($(P)+(.8,-.4)$) -- ($(P)+(0,.2pt)+(.8,0)$);
					\draw[fill=gray!20] ($(P)+(0.2,-.2)$) rectangle ($(P)+(.6,.2)$)
						node[midway]{$\intermediate{r}$};
				\fi
				\coordinate (P) at ($(P)+(1,0)$);
			}			
			\fill[fill=gray!20] (.3,.2)--++(0,.4)--++(1.8,0)--++(0,-.3)--(1.5,.3)--++(0,-.1)--++(-.2,0)--++(0,.1)--(.5,.3)--++(0,-.1)--++(-.2,0)--++(0,.1);
			\draw (.3,.2)--++(0,.4)--++(1.8,0);
			\draw (2.1,.3)--(1.5,.3)--++(0,-.1)--++(-.2,0)--++(0,.1)--(.5,.3)--++(0,-.1)--++(-.2,0)--++(0,.1);		
			\draw[fill=gray!20,draw=black!95] (2.8,.6)--++(0.7,0)--++(0,-.4)--++(-.2,0)--++(0,.1)--(2.8,.3);
			\draw[draw=black!95] (2.5,.45) node {$r_2$};
		\end{tikzpicture}
	}
	
	\subfloat[The mutually disjoint connected regions $\intermediate r+r_1$ and $s_1$.]
	{
		\begin{tikzpicture}
			\coordinate (P) at (0,0);
			\fill[black!25] (-1,1) rectangle (4.5,0);		
			\node at (-.5,.5) {$s_1$};
			\fill[black!0] (-1,-1) rectangle (4.5,0);			
			\draw (-1,0)--(4.5,0);
			\foreach \x in {0,1,2,3}
			{
				\ifnum \x=2					
				\else
					\draw[fill=white] ($(P)-(0,.2pt)$)--($(P)+(0,.4)$) --($(P)+(.8,.4)$) -- ($(P)-(0,.2pt)+(.8,0)$);
					\draw[fill=gray!20] ($(P)+(0.2,-.2)$) rectangle ($(P)+(.6,.2)$) 
						node[midway]{$\intermediate{r}$};
				\fi
				\coordinate (P) at ($(P)+(1,0)$);
			}
			
			\fill[fill=gray!20] (.3,-.2)--++(0,-.4)--++(1.8,0)--++(0,.3)--(1.5,-.3)--++(0,.1)--++(-.2,0)--++(0,-.1)--(.5,-.3)--++(0,.1)--++(-.2,0)--++(0,-.1);
			\draw (.3,-.2)--++(0,-.4)--++(1.8,0);
			\draw (2.1,-.3)--(1.5,-.3)--++(0,.1)--++(-.2,0)--++(0,-.1)--(.5,-.3)--++(0,.1)--++(-.2,0)--++(0,-.1);		
			\draw[fill=gray!20,draw=black!95] (2.8,-.6)--++(0.7,0)--++(0,.4)--++(-.2,0)--++(0,-.1)--(2.8,-.3);
			\draw[draw=black!95] (2.5,-.45) node {$r_1$};	
		\end{tikzpicture}
	}
	}
	\end{center}
	\caption{Eliminating the conjuncts of the form $\lnot C(-r,\intermediate{r})$.} 
	\label{fig:crenellate}
\end{figure}
Let $\psi_\fW$ be the result of replacing all the conjuncts (explicit
or implicit) containing the predicate $C$, as just described. We have
thus shown that, if $\psi_\fW$ is satisfiable over $\RC(\R^2)$, then
$\fW$ is positive, and that, if $\fW$ is positive, then $\psi_\fW$ is
satisfiable over $\RCP(\R^2)$. This completes the proof.

The final case we must deal with is that of $\cBci$. We use the
r.e.-hardness results already established for $\cBCci$, and proceed,
as before, to eliminate occurrences of $C$. Since all the polygons in 
the tuple satisfying $\ti{\phi}_\fW$ are quasi-bounded, we can eliminate 
all occurrences of $C$ from $\ti{\phi}_\fW$ using 
Lemma~\ref{lma:Cci2BciStar}~(\emph{iii}). This completes the proof of 
Theorem~\ref{theo:undecidable}.


\section{$\cBci$ in 3D}\label{sec:Bci3D_C}

\newcommand{\ConRC}{{\sf ConRC}}

Denote by $\ConRC$ the class of all connected topological spaces with regular closed regions. As shown in ~\cite{ijcai:kphz10}, every $\cBci$-formula satisfiable over $\ConRC$ can be satisfied in a finite connected quasi-saw model and the problem $\Sat(\cBci,\ConRC)$ is \NP-complete.

\begin{theorem}
The problems $\Sat(\cBci,\RC(\R^n))$, $n \geq 3$, coincide with $\Sat(\cBci,\ConRC)$, and so are all \NP-complete.
\end{theorem}
\begin{proof}
It suffices to show that every $\cBci$-formula $\varphi$ satisfiable over connected quasi-saws can also be satisfied over any of $\RC(\R^n)$, for $n \geq 3$.
So suppose that $\varphi$ is satisfied in a model $\mathfrak{A}$ based on a finite connected quasi-saw $(W,R)$. Denote by $W_i$ the set of points of depth $i$ in $(W,R)$, for $i=0,1$.
Without loss of generality we may assume that there exists a point $z_0\in W_1$ with $z_0 R x$ for all $x\in W_0$. Indeed, if this is not the case, take the interpretation $\mathfrak B$ obtained by extending $\mathfrak A$ with such a point $z_0$ and setting $z_0 \in r^{\mathfrak B}$ iff $x \in r^{\mathfrak A}$ for some $x \in W_0$. Clearly, we have $\mathfrak A \models (\tau = \tau')$ iff $\mathfrak B \models (\tau = \tau')$, for any terms $\tau$, $\tau'$. To see that $\mathfrak A \models \ic(\tau)$ iff $\mathfrak B \models \ic(\tau)$, recall that $(W,R)$ is connected, and so $\ti{\tau}$ is disconnected in $\mathfrak A$ iff there are two distinct points $x,y \in \tau^{\mathfrak A} \cap W_0$ connected by at least one path in $(W,R)$ and such that no such path lies entirely in $\ti{(\tau^{\mathfrak A})}$. It follows that if $\ti{(\tau^{\mathfrak A})}$ is disconnected then $W_0 \setminus \tau^{\mathfrak A} \neq \emptyset$, and so $z_0 \notin \ti{(\tau^{\mathfrak B})}$.  Thus, by adding $z_0$ to $(W,R)$ we cannot make a disconnected open set in $\mathfrak A$ connected in $\mathfrak B$.

We show now how $\mathfrak{A}$ can be embedded into $\R^n$, for any $n\ge 3$.
First we take pairwise disjoint \emph{closed} balls $B^1_x$ for all $x\in W_0$. We also select pairwise disjoint \emph{open} balls $D_z$ for
$z\in W_1\setminus\{z_0\}$, which are disjoint from all of the $B^1_x$, and take $D_{z_0}$ to be the complement of
\begin{equation*}
\bigcup_{x\in W_0} \ti{(B^1_x)} \ \ \cup \bigcup_{z\in W_1\setminus\{z_0\}} D_z.
\end{equation*}
(Note that $\ti{D_z}$ is connected for each $z\in W_1$; all $D_z$, for $z\in W_1\setminus\{z_0\}$, are open, while $D_{z_0}$ is closed).
We then expand every $B^1_x$ to a set $B_x$ in such a way that the following two properties are satisfied:
\begin{itemize}
\item[(A)] the $B_x$, for $x\in W_0$, form a \emph{connected partition in $\RC(\R^n)$} in the sense that  the $B_x$ are regular closed sets in $\R^n$, whose interiors are non-empty, connected and  pairwise disjoint, and which sum up to the entire space;

\item[(B)] every point in $D_z$, $z\in W_1$, is either

\begin{itemize}
\item[--]
in the interior of some $B_x$ with $zRx$, or

\item[--] on the boundary of \emph{all} of the $B_x$ for which $zRx$.
\end{itemize}
\end{itemize}
The required sets $B_x$ are constructed as follows.
Let $q_1,
q_2, \ldots$ be an enumeration of all the points in $\bigcup_{z\in W_1} \ti{D_z}$ with \emph{rational} coordinates.
For $x\in W_0$, we set $B_x$ to be the closure of the infinite union $\bigcup_{k \in \omega} \ti{(B_x^k)}$, where the regular closed sets $B_x^k$ are defined inductively as follows (see Fig.~\ref{fig:apollonian-app}):
\begin{itemize}
\item[--] Assuming that the $B^k_x$ are already defined, let $q_i$ be the first point in the list $q_1,
q_2, \ldots$ such that $q_i\notin B^k_x$, for all $x\in W_0$. Suppose $q_i \in \ti{D_z}$ for $z\in W_1$. Take an open ball $C_{q_i} \subsetneqq \ti{D_z}$ of radius $< 1/k$ centred in $q_i$ and disjoint from the $B^k_x$. For each $x\in W_0$ with $zRx$, expand $B_x^k$ by a closed ball in $C_{q_i}$ and a closed rod connecting it to $B_x^1$ in such a way that the ball and the rod are disjoint from the rest of the $B^k_x$. The resulting set is denoted by $B^{k+1}_x$.
\end{itemize}
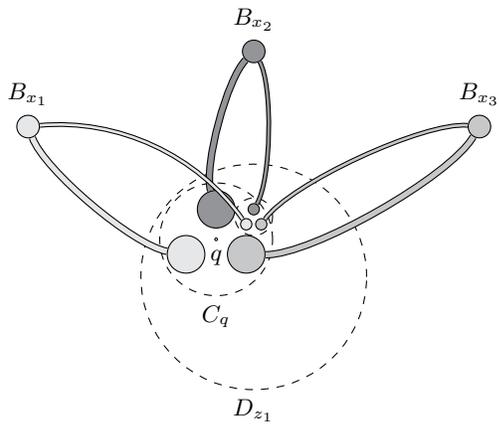
\begin{figure}[h]
\begin{center}
\begin{tikzpicture}[
clball/.style={circle,draw=black,minimum size=3mm,inner sep=0pt},
opball/.style={circle,dashed,draw=black,minimum size=30mm,inner sep=0pt}
]
\node [label=above:{\small $B_{x_1}$}] (x1) at (0,-1)[clball,fill=Gray!20] {};
\node [label=above:{\small $B_{x_2}$}] (x2) at (3,0)[clball,fill=Gray] {};
\node [label=above:{\small $B_{x_3}$}] (x3) at (6,-1)[clball,fill=Gray!50] {};
%
\node [label=below:{\small $D_{z_1}$}] (z1) at (3,-3)[opball] {};
\node [label=below:{\small $C_q$}](c) at (2.5,-2.5)[opball,minimum size=15mm] {};
\node [label=below:{\small $q$}](q) at (2.5,-2.5)[circle,inner sep=0pt,minimum size=1,draw=black] {};
\node (xc1) at (2.1,-2.7)[clball,fill=Gray!20,minimum size=5mm] {};
\node (xc2) at (2.5,-2.1)[clball,fill=Gray,minimum size=5mm] {};
\node (xc3) at (2.9,-2.7)[clball,fill=Gray!50,minimum size=5mm] {};
\draw[double=Gray!20,double distance=2pt] (xc1) to [bend left, looseness=0.5] (x1);
\draw[double=Gray,double distance=2pt] (xc2) to [bend left, looseness=0.5] (x2);
\draw[double=Gray!50,double distance=2pt] (xc3) to [bend right, looseness=0.5] (x3);
\node (cp) at (3,-2.2)[opball,minimum size=5mm] {};
\node (xcp1) at (2.9,-2.3)[clball,fill=Gray!20,minimum size=1.5mm] {};
\node (xcp2) at (3,-2.1)[clball,fill=Gray,minimum size=1.5mm] {};
\node (xcp3) at (3.1,-2.3)[clball,fill=Gray!50,minimum size=1.5mm] {};
\draw[double=Gray!20,double distance=1pt] (xcp1) to [bend right, looseness=0.9] (x1);
\draw[double=Gray,double distance=1pt] (xcp2) to [bend right, looseness=0.5] (x2);
\draw[double=Gray!50,double distance=1pt] (xcp3) to [bend left, looseness=0.5] (x3);
\end{tikzpicture}
\end{center}
\caption{The first two stages of filling $D_{z_1}$ with $B_{x_i}$, for $z_1 R x_i$, $i = 1,2,3$. (In $\R^3$, the sets $B_{x_1}$ and $B_{x_2}$ would not intersect.)}\label{fig:apollonian-app}
\end{figure}

Let $\RC(W,R)$ be the Boolean algebra of regular closed sets in $(W,R)$ and let $\RC(\R^n)$ be the Boolean algebra of regular closed sets in $\R^n$.
Define a map $f$ from $\RC(W,R)$ to $\RC(\R^n)$ by taking
\begin{equation*}
f(X) ~=~ \bigcup_{x\in X \cap W_0} B_x,\quad \text{ for $X \in \RC(W,R)$}.
\end{equation*}
By~(A), $f$ is an isomorphic embedding of $\RC(W,R)$ into $\RC(\R^n)$, that is, $f$ preserves
the operations $+$, $\cdot$ and $-$ and the constants $0$ and $1$.
Define an interpretation $\mathfrak{I}$ over $\RC(\R^n)$ by taking $r^\mathfrak{I} = f(r^\mathfrak{A})$. To show that $\mathfrak I \models \varphi$, it remains to prove that, for every $X \in \RC(W,R)$, $\ti{X}$ is connected if, and only if, $\ti{(f(X))}$ is connected. This equivalence follows from the fact that
\begin{equation*}
\ti{(f(X))} ~=~ \bigcup_{x\in X \cap W_0} \ti{B}_x \quad\cup \bigcup_{z\in X\cap W_1,\ V_z\subseteq X} D_z,
\end{equation*}
where $V_z\subseteq W_0$ is the set of all $R$-successors of $z$ of depth 0, which in turn is an immediate consequence of (B).
\end{proof}
\end{document}